\documentclass[11pt]{article}
\usepackage{jheppub}
\usepackage{bm,bbm}
\usepackage{booktabs,amsmath}
\usepackage{accents}
\usepackage{multirow}
\usepackage{graphics}
\usepackage{braket}
\usepackage{mathtools}
\usepackage{comment}
\usepackage{autobreak}
\usepackage{subcaption}
\usepackage{enumerate}
\usepackage{blkarray}
\usepackage{longtable}

\usepackage{tikz}
\usetikzlibrary{patterns}
\usetikzlibrary{arrows,shapes,positioning}
\usetikzlibrary{decorations.markings}
\usetikzlibrary{decorations.pathmorphing}
\tikzset{snake it/.style={decorate, decoration=snake}}

\hypersetup{
    colorlinks=true,
    linkcolor=blue,
    filecolor=magenta,      
    urlcolor=cyan,
    pdfpagemode=FullScreen,
    }

\usepackage{arydshln}
\usepackage{mathtools}

\edef\restoreparindent{\parindent=\the\parindent\relax}
\usepackage{parskip}
\restoreparindent
\usepackage{ascmac}
\usepackage{mathrsfs}

\usepackage{amsthm}
\newtheoremstyle{break}
  {\topsep}{\topsep}%
  {\upshape}{}%
  {\bfseries}{}%
  {\newline}{}%
\theoremstyle{break}
\newtheorem{theorem}{Theorem}[section]
\newtheorem{proposition}{Proposition}[section]
\newtheorem{corollary}{Corollary}[section]

\def\Tr{{\rm Tr}}
\def\d{{\rm d}}
\def\i{{\rm i}}

\def\CB{{\cal B}}
\def\CC{{\cal C}}

\def\CF{{\cal F}}

\def\CH{{\cal H}}

\def\CM{{\cal M}}
\def\CN{{\cal N}}
\def\CO{{\cal O}}

\def\BC{\mathbb{C}}
\def\BF{\mathbb{F}}

\def\BR{\mathbb{R}}

\def\BZ{\mathbb{Z}}

\def\Ba{\mathbf{a}}
\def\Bb{\mathbf{b}}
\def\Bn{\mathbf{n}}
\def\Bz{\mathbf{z}}

\def\Balpha{\boldsymbol{\alpha}}
\def\Bbeta{\boldsymbol{\beta}}
\def\BOmega{\boldsymbol{{\Omega}}}
\def\Bzero{\boldsymbol{{0}}}
\def\Bone{\mathbf{1}}
\def\Bdelta{\boldsymbol{{\delta}}}
\def\BDelta{\boldsymbol{{\Delta}}}
\def\SH{\mathsf{H}}
\def\SW{\mathsf{W}}

\def\d{\mathrm{d}}

\def\U{\mathrm{U}}

\def\U{\mathrm{U}}

\title{
Narain CFTs from quantum codes and their $\BZ_2$ gauging}

\author[a,b]{Kohki Kawabata,}
\author[b]{Tatsuma Nishioka}
\author[c]{and Takuya Okuda}

\affiliation[a]{Department of Physics, Faculty of Science,
The University of Tokyo,\\
Bunkyo-Ku, Tokyo 113-0033, Japan}

\affiliation[b]{Department of Physics, Osaka University,\\
Machikaneyama-Cho 1-1, Toyonaka 560-0043, Japan}

\affiliation[c]{Graduate School of Arts and Sciences, The University of Tokyo, Komaba,\\
Meguro-ku, Tokyo 153-8902, Japan}

\preprint{OU-HET-1196, UT-Komaba/23-10}

\abstract{
We investigate the gauging of a $\BZ_2$ symmetry in Narain conformal field theories (CFTs) constructed from qudit stabilizer codes. 
Considering both orbifold and fermionization, we establish a connection between $\BZ_2$ gauging procedures and modifications of the momentum lattice by vectors characterizing the $\BZ_2$ symmetry. 
We also provide three-dimensional interpretations of $\BZ_2$ gaugings through abelian Chern-Simons theories, which act as symmetry topological field theories. 
}

\begin{document} 
\maketitle
\flushbottom

\newpage

\section{Introduction}

The theory of error-correcting codes has played unexpectedly fruitful roles in mathematics and theoretical physics despite its initial development as a method of information communication.
One of the most significant outcomes is the construction of two-dimensional chiral conformal field theories (CFTs) from classical error-correcting codes through Euclidean lattices~\cite{frenkel1984natural,frenkel1989vertex,Dolan:1994st,Gaiotto:2018ypj,Kawabata:2023nlt}.
Such a class of chiral CFTs has played an important role in understanding the monstrous moonshine~\cite{conway1979monstrous} and the classification and construction of chiral CFTs with fixed central charges~\cite{Lam201971HV}.
Besides error-correcting codes, a key ingredient in such developments is the gauging procedure~\cite{Dolan:1989vr,Dolan:1994st,Dixon:1988qd}.
The orbifold and fermionization implement the gauging of a discrete symmetry in an original bosonic theory and constitute new consistent bosonic and fermionic theories.
By using the orbifolding technique, the Monster CFT~\cite{frenkel1984natural} has been obtained from the chiral CFT based on the Leech lattice closely related to the extended Golay code~\cite{conway2013sphere}.
Also, fermionization has demonstrated its utility in the construction of fermionic analogs of the Monster CFT~\cite{Dixon:1988qd} and found applications in searching for supersymmetric CFTs with large discrete symmetries \cite{Benjamin:2015ria,Harrison:2016hbq}.
Additionally, many entries in the Schellekens list~\cite{Schellekens:1992db},\footnote{See also \cite{Hohn:2017dsm,Hohn:2020xfe,Moller:2019tlx,vanEkeren:2020rnz} for recent developments.} which dictates the potential chiral CFTs with central charge $24$, have been constructed from classical error-correcting codes with and without orbifolding \cite{Dolan:1989vr,Dolan:1994st}.

Inspired by the construction of chiral CFTs from classical codes, it has recently turned out that binary quantum error-correcting codes naturally lead to a certain type of non-chiral bosonic CFTs~\cite{Dymarsky:2020qom,Dymarsky:2020bps}, called Narain CFTs~\cite{Narain:1985jj,Narain:1986am}.
Narain CFTs are characterized by Lorentzian lattices while the chiral CFTs are by Euclidean lattices.
The construction of Narain CFTs exploits the relationship between quantum codes, classical codes, Lorentzian lattices, and CFTs.
This procedure yields a discrete subset of Narain CFTs named Narain code CFTs, and they found various applications such as the realization of CFTs with large spectral gaps~\cite{Furuta:2022ykh,Angelinos:2022umf}, the validation of holography~\cite{Dymarsky:2020pzc}, and the search for solutions of the modular bootstrap~\cite{Dymarsky:2020bps,Henriksson:2022dnu,Dymarsky:2022kwb}. (See also~\cite{Dymarsky:2021xfc,Buican:2021uyp,Henriksson:2021qkt,Henriksson:2022dml,Yahagi:2022idq,Kawabata:2023nlt,Furuta:2023xwl} for other related progress.)
A similar construction also works in non-binary quantum codes based on finite fields $\BF_p$ of prime order $p$~\cite{Kawabata:2022jxt}, and on more general finite fields and finite rings~\cite{Alam:2023qac}.

In this paper, we investigate the gauging of a $\mathbb{Z}_2$ symmetry in Narain code CFTs.
In a modern perspective, orbifold and fermionization are different gauging procedures of the same $\BZ_2$ symmetry in a bosonic CFT, which yields another bosonic CFT and fermionic CFT, respectively~\cite{Tachikawalec,Karch:2019lnn,Hsieh:2020uwb,Kulp:2020iet}.
The basic strategy of the present paper is to use the modern description of the orbifold and fermionization to describe $\BZ_2$ gaugings of Narain CFTs as modifications of the momentum lattices.
Fermionization of Narain code CFTs has been explored in the recent paper~\cite{Kawabata:2023usr} with an aim to search for fermionic CFTs with supersymmetry systematically by using quantum stabilizer codes.
On the other hand, this paper aims to provide a unified perspective on the orbifold and fermionization in Narain code CFTs.

From the modern viewpoint, we establish the relationship between the $\BZ_2$ gaugings and modifications of the momentum lattice in Narain CFTs.
For a bosonic CFT with a $\BZ_2$ symmetry, one can extend the Hilbert space (untwisted sector) by adding another Hilbert space (twisted sector).
This leads to the decomposition of the Hilbert spaces into four subsectors graded by their charges under the $\BZ_2$ symmetry and whether they are twisted or not.
Then, the orbifold and fermionization can be regarded as the interchange of the subsectors in the Hilbert spaces as shown in Table~\ref{table:gauging}.
For Narain CFTs, one can take a certain type of $\BZ_2$ symmetry associated with a vector of the momentum lattice.
This choice boils down the construction of the twisted sector to a half shift of the momentum lattice by the corresponding vector. 
With this shifted lattice, the $\BZ_2$ gaugings of Narain CFTs can be realized as the modifications of the momentum lattice as summarized in Table~\ref{table:gauging_deforamtion}. 
With this general relationship between the $\BZ_2$ gaugings and the lattice modifications, we explore the orbifold and fermionization in Narain code CFTs based on finite fields of prime order.
By choosing a $\BZ_2$ symmetry that exists for any Narain code CFT, we compute the partition functions of the orbifolded and fermionized theories.
These partition functions are expressed in terms of the measure of spectrum in classical codes called the complete weight enumerator.
We illustrate the procedure with several examples including a quantum code that realizes an $\CN=4$ superconformal model~\cite{Gaberdiel:2013psa} by fermionization~\cite{Kawabata:2023usr}.

We give three-dimensional interpretations of Narain code CFTs and their $\BZ_2$ gauging in terms of abelian Chern-Simons theories.
We use the characterization of abelian Chern-Simons theories by lattices~\cite{Belov:2005ze} to relate them to Narain code CFTs.
Some of the three-dimensional bulk theories are invariant under the $\mathbb{Z}_2$ gauging on 2d CFTs, and the gauging only affects the topological boundary conditions, which suggests that the Chern-Simons theories play the role of symmetry topological field theories~\cite{Gaiotto:2020iye,Apruzzi:2021nmk,Freed:2022qnc} of Narain code CFTs.
In addition, with a special focus on an ensemble of quantum codes of the Calderbank-Shor-Steane (CSS) type~\cite{calderbank1996good,Steane:1996va}, we compute the averaged partition functions in the orbifold and fermionization of Narain code CFTs and discuss their spectrum in the large central charge limit.

This paper is organized as follows.
In section~\ref{sec:gauging-Z2-bosonic}, we start with a review of the $\mathbb{Z}_2$ gauging of bosonic CFTs with emphasis on the unification of the orbifold and fermionization.
We introduce the relationship between the $\BZ_2$ gaugings and the lattice modifications using a compact free boson theory as the simplest example.
Then, we proceed to the general formulation of the lattice modifications valid for any Narain CFT.
In section~\ref{sec:gauging_code_CFT}, we are focused on the $\mathbb{Z}_2$-gauging of Narain code CFTs and give a systematic way of computing the partition functions of their orbifold and fermionization.
Section~\ref{sec:examples} illustrates our general results with several examples containing a quantum code known to realize an $\CN=4$ supersymmetry.
In section~\ref{sec:chern-simons}, we give three-dimensional interpretations to the original, orbifolded, and fermionized Narain code CFTs.
We see that the stabilizer codes specify topological boundary conditions of the corresponding Chern-Simons theories.
In section~\ref{sec:ensemble}, we compute the averaged partition functions of the orbifolded and fermionized Narain code CFTs.
We conclude the paper with discussions in section~\ref{sec:discussion}.

\section{Gauging $\BZ_2$ symmetry of bosonic CFTs}
\label{sec:gauging-Z2-bosonic}

This section starts with a review of the gauging procedure with a global $\BZ_2$ symmetry in two-dimensional bosonic CFTs from a modern viewpoint. 
In section~\ref{ss:gauging_modification}, we consider $\BZ_2$ gauging of a compact free boson. 
This is the simplest example that allows us to describe the gauging of a $\BZ_2$ symmetry as a modification of the momentum lattice.
Section~\ref{ss:general_modification} is devoted to the general formulation of lattice modification, which leads to orbifold and fermionization of Narain CFTs by a certain $\BZ_2$ symmetry.

\subsection{Orbifold and fermionization}

For a 2d bosonic CFT with a global symmetry, we can construct a new bosonic CFT by gauging the symmetry of the original theory.
This procedure is called orbifold \cite{Dixon:1985jw,Dixon:1986jc,Dixon:1986qv} and has been studied with particular importance as it provides a consistent theory from the original one.
On the other hand, fermionization has taken much attention as a specific technique in 2d theories \cite{Coleman:1974bu,Mandelstam:1975hb} and was recently given by a modern description in analogy with the orbifold.
In this section, focusing on a global $\BZ_2$ symmetry, we review orbifold and fermionization in a parallel manner by following \cite{Ji:2019ugf,Hsieh:2020uwb,Kulp:2020iet}.

Consider a bosonic CFT $\CB$ with central charge $c$ and a non-anomalous global $\BZ_2$ symmetry $\sigma$.
The Hilbert space $\CH$ can be decomposed to the even and odd subsectors under the $\BZ_2$ symmetry: 
\begin{align}
    \CH = \CH^+\oplus \CH^- \ ,
\end{align}
where $\CH^\pm = \{\ket{\phi}\in\CH\mid \sigma\ket{\phi}=\pm\ket{\phi}\}$. 
This Hilbert space $\CH$ is called the untwisted sector.
To define the twisted sector, let us place the theory on a cylinder $S^1\times\BR$ where $S^1$ is the spacial circle and $\BR$ is the time direction, respectively. 
In the $\sigma$-twisted Hilbert space $\CH_\sigma$, any operator $\phi$ obeys the twisted boundary condition,
\begin{align}
   \CH_\sigma: \quad \phi(x + 2\pi) = \sigma\cdot \phi(x) \ ,
\end{align}
along the circle direction parametrized by $x\in S^1$ with $x\sim x+2\pi$.
The twisted sector also decomposes to the $\BZ_2$ even and odd subsectors:
\begin{align}
    \CH_\sigma = \CH_\sigma^+\oplus \CH_\sigma^- \ .
\end{align}

We have assumed that a global $\BZ_2$ symmetry $\sigma$ is non-anomalous.
Otherwise, the resulting gauged theories would be inconsistent and, for example, such a theory is not invariant under the modular transformation.
Let us consider the spin $s$ of an operator in the twisted Hilbert space $\CH_\sigma$.
Then, the spin selection rule in the twisted sector is given by (see \cite{Lin:2019kpn})
\begin{align}
    s\in
    \begin{dcases}
    \,\frac{\BZ}{2}, & (\text{non-anomalous})\ ,\\
    \,\frac{1}{4}+\frac{\BZ}{2}, & (\text{anomalous})\ .
    \end{dcases}
    \label{eq:spin_selection}
\end{align}
In the non-anomalous (anomalous) case, the spin of an operator in
$\CH_\sigma$ is $s\in\frac{\BZ}{2}$ ($s\in\frac{1}{4}+\frac{\BZ}{2}$).
Conversely, we can diagnose the anomaly of a global $\BZ_2$ symmetry by the spin selection rule: if the twisted sector $\CH_\sigma$ consists of operators with the spin $s\in\frac{\alpha}{4}+\frac{\BZ}{2}$ ($\alpha\in\BZ_2$), the $\BZ_2$ symmetry $\sigma$ is non-anomalous $\alpha=0$ (anomalous $\alpha=1$).

Let us place the bosonic theory $\CB$ on a torus with the modulus $\tau=\tau_1+\i\,\tau_2$.
The torus has two independent cycles and we specify them by
\begin{align}
    \mathrm{timelike:}\; w\sim w+2\pi\tau\,,\qquad \mathrm{spacial:}\; w\sim w+2\pi\,,
\end{align}
where $w$ is the cylindrical coordinate.
For each cycle of the torus, we can choose the periodicity condition $(a,b)\in\BZ_2\times\BZ_2$:
\begin{align}
\begin{aligned}
    \phi(w+2\pi\tau) = \sigma^a\cdot \phi(w)\,,\qquad \phi(w+2\pi) = \sigma^b\cdot \phi(w)\,.
\end{aligned}
\end{align}
While the spacial periodicity $b=0$ ($b=1$) specifies the untwisted (twisted) Hilbert space, the timelike periodicity $a\in\BZ_2$ imposes $a$ times insertion of the $\BZ_2$ symmetry operator $\sigma$ along the spacial direction. 
Then, we have four partition functions depending on the periodicity
\begin{align}
    \begin{aligned}
        Z &= \Tr_{\CH}\left[ q^{L_0 - \frac{c}{24}}\,\bar q^{\bar L_0 - \frac{c}{24}}\right] \ ,&\qquad\;
        Z^\sigma &= \Tr_{\CH}\left[\sigma\, q^{L_0 - \frac{c}{24}}\,\bar q^{\bar L_0 - \frac{c}{24}}\right] \ ,\\
        Z_\sigma &= \Tr_{\CH_\sigma}\left[q^{L_0 - \frac{c}{24}}\,\bar q^{\bar L_0 - \frac{c}{24}}\right] \ ,&\qquad
        Z_\sigma^\sigma &= \Tr_{\CH_\sigma}\left[\sigma\, q^{L_0 - \frac{c}{24}}\,\bar q^{\bar L_0 - \frac{c}{24}}\right] \ .
    \end{aligned}
\end{align}
In relation to the $\BZ_2$ gauging, it is convenient to introduce another set of partition functions associated to the four sectors shown in the top left panel of Table \ref{table:gauging}:
\begin{align}
    \begin{aligned}
        S &=\!\!\!\! & \Tr_{\CH^+}\left[ q^{L_0 - \frac{c}{24}}\,\bar q^{\bar L_0 - \frac{c}{24}}\right]  &= \Tr_{\CH}\left[ \frac{1+ \sigma}{2}\,q^{L_0 - \frac{c}{24}}\,\bar q^{\bar L_0 - \frac{c}{24}}\right] \ ,\\
        T &=\!\!\!\! & \Tr_{\CH^-}\left[ q^{L_0 - \frac{c}{24}}\,\bar q^{\bar L_0 - \frac{c}{24}}\right]  &= \Tr_{\CH}\left[\frac{1-\sigma}{2}\, q^{L_0 - \frac{c}{24}}\,\bar q^{\bar L_0 - \frac{c}{24}}\right] \ ,\\
        U &=\!\!\!\! & \Tr_{\CH_\sigma^+}\left[ q^{L_0 - \frac{c}{24}}\,\bar q^{\bar L_0 - \frac{c}{24}}\right]  &=  \Tr_{\CH_\sigma}\left[\frac{1+\sigma}{2}\,q^{L_0 - \frac{c}{24}}\,\bar q^{\bar L_0 - \frac{c}{24}}\right] \ ,\\
        V &=\!\!\!\! & \Tr_{\CH_\sigma^-}\left[ q^{L_0 - \frac{c}{24}}\,\bar q^{\bar L_0 - \frac{c}{24}}\right]  &= \Tr_{\CH_\sigma}\left[\frac{1-\sigma}{2}\, q^{L_0 - \frac{c}{24}}\,\bar q^{\bar L_0 - \frac{c}{24}}\right] \ .
    \end{aligned}
\end{align}
Here, we denote the four partition functions by $S,T,U,V$.
In the top left panel of Table \ref{table:gauging}, however, we slightly abuse the notation and use these symbols to refer to the corresponding sector in the Hilbert space. In what follows, we use this notation rather than $\CH^\pm, \CH_\sigma^\pm$ for simplicity.
Note that the two sets of the partition functions are related by
\begin{align}
    \begin{aligned}
    \label{eq:rel_four_partition}
        S = \frac{Z + Z^\sigma}{2} \ ,\qquad
        T = \frac{Z - Z^\sigma}{2}\ ,\qquad
        U = \frac{Z_\sigma + Z^\sigma_\sigma}{2}\ , \qquad
        V = \frac{Z_\sigma - Z^\sigma_\sigma}{2}\ .
    \end{aligned}
\end{align}

\begin{table}
  \centering
      \begin{subtable}[t]{0.45\textwidth}
        \centering
              \begin{tabular}{ccc}
              \toprule
                $\CB$  & untwisted  &  twisted  \\
                \hline 
                even  & $S$  & $U$ \\
                odd  & $T$   & $V$ \\\bottomrule
              \end{tabular}
        \caption{Bosonic CFT}
      \end{subtable}
      \begin{subtable}[t]{0.45\textwidth}
            \centering
                  \begin{tabular}{ccc}
                  \toprule
                    $\CO$  & untwisted  &  twisted  \\
                    \hline 
                    even  & $S$  & $T$ \\
                    odd  & $U$   & $V$ \\\bottomrule
                  \end{tabular}
            \caption{Orbifold CFT}
      \end{subtable}
      \vspace*{0.5cm}
  \\
  \vskip\baselineskip
        \begin{subtable}[t]{0.45\textwidth}
            \centering
              \begin{tabular}{ccc}
              \toprule
                $\CF$  & NS sector  &  R sector  \\
                \hline 
                even  & $S$  & $T$ \\
                odd  & $V$   & $U$ \\\bottomrule
              \end{tabular}
            \caption{Fermionic CFT}
        \end{subtable}
        \begin{subtable}[t]{0.45\textwidth}
            \centering
              \begin{tabular}{ccc}
              \toprule
                $\widetilde{\CF}$  & NS sector  &  R sector  \\
                \hline 
                even  & $S$  & $U$ \\
                odd  & $V$   & $T$ \\\bottomrule
              \end{tabular}
              \caption{The other fermionic CFT}
        \end{subtable}
  \vspace{0.5cm}
  \caption{The $\BZ_2$ gradings of the untwisted and twisted Hilbert space for the original bosonic theory $\CB$, the orbifold theory $\CO$, and the fermionized theory $\CF$, $\widetilde{\CF}$. The gradings are given by the $\BZ_2$ symmetry $\sigma$ for $\CB$, its dual symmetry $\hat{\sigma}$ for $\CO$, and the fermion parity $(-1)^F$ for $\CF$, $\widetilde{\CF}$.
  The four sectors are swapped with each other under orbifolding and fermionization.}
  \label{table:gauging}
\end{table}

\paragraph{$\BZ_2$ orbifold}

The $\BZ_2$ orbifold can be applied to the bosonic theory $\CB$ by choosing a non-anomalous global $\BZ_2$ symmetry $\sigma$. The standard procedure of the orbifold is the extension of the untwisted Hilbert space $\CH$ to the twisted Hilbert space $\CH_\sigma$ followed by the projection onto the $\BZ_2$ even sectors.
Table \ref{table:gauging} shows the $\BZ_2$ even sectors in the original theory $\CB$ are $S\oplus U$.
This gives the orbifold partition function
\begin{align}
    Z_\CO = S+U = \frac{1}{2}\left[ Z + Z^\sigma + Z_\sigma + Z^\sigma_\sigma \right]\,,
\end{align}
where we use the relationship \eqref{eq:rel_four_partition}.
The rightmost expression implies that the $\BZ_2$ orbifold sums over all possible $\BZ_2$ gauge field configurations on the torus where the original theory $\CB$ is defined.
Therefore, the $\BZ_2$ orbifold can be seen as the gauging of the $\BZ_2$ symmetry $\sigma$:
\begin{align}
    \text{Orbifold theory}~\CO ~=~ \frac{\text{Bosonic theory}~\CB}{\BZ_2^\sigma} \ .
\end{align}
After gauging the $\BZ_2$ symmetry $\sigma$, the orbifold theory $\CO$ no longer has the global symmetry $\sigma$.
However, the dual $\BZ_2$ symmetry $\hat{\sigma}$ (more precisely, the Pontrjagin dual of $\sigma$) emerges in the orbifold theory $\CO$ \cite{Vafa:1989ih}.\footnote{Taking the sum over all possible $\BZ_2$ configuration of the dual symmetry $\hat{\sigma}$, we can apply the orbifold by $\hat{\sigma}$ to the orbifold theory $\CO$ and it returns the original theory $\CB$. For more general discussion, readers can refer to \cite{Bhardwaj:2017xup}.}
Then, as in the original theory $\CB$, we can give the $\BZ_2$ grading in the orbifold Hilbert space by $\hat{\sigma}$ and also construct the twisted sector by $\hat{\sigma}$.

To introduce a more general partition function of the orbifold theory $\CO$, suppose that $t$ is a $\BZ_2$ gauge field configuration of $\hat{\sigma}$ on a spacetime where the theory lives.
The orbifold partition functions with the $\BZ_2$ configuration $t$ on the genus-$g$ Riemann surface $\Sigma_g$ can be written by
\begin{align}
    Z_\CO[t] = \frac{1}{2^g}\,\sum_{s\,\in\, H^1(\Sigma_g, \BZ_2)} Z_\CB[s]\,\exp\left[ \i\,\pi \int s\cup t\right] \ ,
\end{align}
where $s$ and $t$ take values in $H^1(\Sigma_g, \BZ_2)$, and $\cup$ is the cup product.
On a torus, the $\BZ_2$ configurations $s$, $t\in\BZ_2\times \BZ_2$ are specified by the periodicity $(a,b)$ along the temporal $a\in\BZ_2$ and spacial direction $b\in\BZ_2$.
Then, there are four types of partition functions $Z_\CO[ab]$ for the orbifold theory:
\begin{align}
    \begin{aligned}
    \label{eq:orbifold_part_gen}
        Z_\CO[00] &= \frac{1}{2}\left[ Z + Z^\sigma + Z_\sigma + Z^\sigma_\sigma \right] = S + U \ , \\
        Z_\CO[10] &= \frac{1}{2}\left[ Z + Z^\sigma - Z_\sigma - Z^\sigma_\sigma \right] = S - U \ , \\
        Z_\CO[01] &= \frac{1}{2}\left[ Z - Z^\sigma + Z_\sigma - Z^\sigma_\sigma \right] = T + V \ , \\
        Z_\CO[11] &= \frac{1}{2}\left[ Z - Z^\sigma - Z_\sigma + Z^\sigma_\sigma \right] = T - V \ ,
    \end{aligned}
\end{align}
where we used the identification:
\begin{align}
    \begin{aligned}
        Z_\CB[00] = Z \ , \quad
        Z_\CB[10] = Z^\sigma \ , \quad
        Z_\CB[01] = Z_\sigma \ , \quad
        Z_\CB[11] = Z^\sigma_\sigma \ .
    \end{aligned}
\end{align}
The relations \eqref{eq:orbifold_part_gen} show that the sector $S$ ($U$) is the even (odd) untwisted Hilbert space of the orbifold theory $\CO$. On the other hand, the sector $T$ ($V$) is the even (odd) sector of the Hilbert space twisted by the dual symmetry $\hat{\sigma}$.
Therefore, we can recapitulate the Hilbert space of the orbifold theory $\CO$ with respect to the dual $\BZ_2$ symmetry $\hat{\sigma}$ in the top right panel of Table \ref{table:gauging}.
Compared with the original theory $\CB$, taking the orbifold can be understood as the swapping of the two sectors $U\leftrightarrow V$.

\paragraph{Fermionization}

Starting with a bosonic theory $\CB$ with a non-anomalous $\BZ_2$ symmetry, we can turn it into a fermionic theory using the procedure called fermionization \cite{Tachikawalec,Karch:2019lnn}.
Fermionization exploits a non-trivial two-dimensional invertible spin topological theory known as the Kitaev Majorana chain \cite{Kitaev:2000nmw}, which has a global $\BZ_2$ symmetry.
By coupling $\CB$ to the Kitaev Majorana chain and taking the diagonal $\BZ_2$ orbifold, the bosonic theory $\CB$ turn into a fermionic theory $\CF$:
\begin{align}
    \text{Fermionized theory}~\CF ~=~ \frac{\text{Bosonic theory}~\CB\times \text{Kitaev}}{\BZ_2} \ .
\end{align}
After the fermionization, the fermionic theory does not have the diagonal $\BZ_2$ symmetry. Instead, the fermion parity $(-1)^F$ emerges and we can give define the untwisted and twisted sectors by the $\BZ_2$ grading with respect to the fermion parity.
The resulting untwisted sector and the twisted sector are called the Neveu-Schwarz (NS) sector and the Ramond (R) sector, respectively.

On a genus-$g$ Riemann surface $\Sigma_g$, the fermionized partition function with spin structure $\rho$ is given by
\begin{align}
    Z_\CF[t+\rho] = \frac{1}{2^g}\,\sum_{s\,\in\, H^1(\Sigma_g, \BZ_2)} Z_\CB[s]\,\exp\left[ \i\,\pi \left(\text{Arf}\,[s+\rho] + \text{Arf}\,[\rho] + \int s\cup t\right)\right] \ ,
\end{align}
where $t\in H^1(\Sigma_g,\BZ_2)$ is the $\BZ_2$ configuration of the fermion parity.
Here, $\text{Arf}$ is called the Arf invariant and depends on the choice of spin structures on $\Sigma_g$.
The resulting fermionic theory has partition functions that depend on the choice of spin structures.

On a torus, there are four spin structures by depending on whether fermions are anti-periodic (A) or periodic (P) along each cycle.
The Arf invariant takes the values:
\begin{align}
    \text{Arf}\,[\rho] 
        =
        \begin{cases}
            0 & \quad \rho = 00~ (\text{AA}),~ 01~ (\text{AP}),~ 10~ (\text{PA}) \\
            1 & \quad  \rho = 11~ (\text{PP})
        \end{cases}
\end{align}
where the first and second letters stand for the periodicity along the temporal and spacial direction, respectively.
There are four types of fermionic partition functions:
\begin{align}
    \begin{aligned}
    \label{eq:fermionized_part_gen}
        Z_\CF[00] 
            &=
            \frac{1}{2}\left[ Z + Z^\sigma + Z_\sigma - Z^\sigma_\sigma \right] = S + V\ , \\
        Z_\CF[10] 
            &=
            \frac{1}{2}\left[ Z  + Z^\sigma - Z_\sigma + Z^\sigma_\sigma \right] = S - V \ , \\
        Z_\CF[01] 
            &=
            \frac{1}{2}\left[ Z  - Z^\sigma + Z_\sigma + Z^\sigma_\sigma \right] = T + U\ , \\
        Z_\CF[11] 
            &=
            \frac{1}{2}\left[ Z - Z^\sigma - Z_\sigma - Z^\sigma_\sigma \right] = T - U\ .
    \end{aligned}
\end{align}
Note that the second letter specifies the NS sector for the anti-periodic boundary and the R sector for the periodic boundary.
From \eqref{eq:fermionized_part_gen}, the sector $U$ ($V$) has even (odd) fermion parity in the NS sector. Similarly, the sector $T$ ($U$) is even (odd) with respect to the fermion parity in the R sector.
Hence, the Hilbert space of the fermionized theory $\CF$ can be summarized as the bottom left panel in Table \ref{table:gauging}.
Comparing the bosonic theory $\CB$ and the fermionized counterpart $\CF$ in Table \ref{table:gauging}, fermionization can be regarded as the cyclic permutation of the three sectors $T\to U\to V\to T$.

From a bosonic theory $\CB$, we can construct the other fermionic CFT denoted by $\widetilde{\CF}$. 
These fermionic theories $\CF$ and $\widetilde{\CF}$ are related by the stacking of the Kitaev Majorana chain:
\begin{align}
    \text{Fermionized theory}~\widetilde{\CF} ~=~ \text{Fermionized theory}~\CF \times \text{Kitaev}\,,
\end{align}
which results in the difference of the $\BZ_2$ grading in the R sector. Therefore, the Hilbert space of $\widetilde{\CF}$ can be given by the bottom right panel in Table \ref{table:gauging}.

\subsection{$\BZ_2$ gauging from lattice modification}
\label{ss:gauging_modification}

In this section, we consider a 2d bosonic CFT $\CB$ with a particular type of a non-anomalous $\BZ_2$ symmetry and their gauging by the $\BZ_2$ symmetry.
In particular, we treat a bosonic theory characterized by a momentum lattice and its $\BZ_2$ symmetry related to a shift of momentum lattice.
This exemplifies the relationship between the $\BZ_2$ gauging and lattice modification.
For the $\BZ_2$ orbifold, the following relation with lattice modification is demonstrated in~\cite{Dymarsky:2020qom}.

More specifically, we consider a single compact boson whose momenta are characterized by a two-dimensional Lorentzian lattice $\Lambda\subset\BR^{1,1}$.
Let $X(z,\bar{z}) = X_L(z) + X_R(\bar{z})$ be a compact free boson with radius $R$. 
The compact boson is normalized as $X(z,\bar{z})\,X(0,0)\sim-\log|z|^2$, which corresponds to $\alpha'=2$ in the notation of \cite{Polchinski:1998rq}.
In this setup, the vertex operators are
\begin{align}\label{Vertex_Operator}
    V_{p_L,\,p_R}(z,\bar{z}) = :e^{\i\, p_L X_L(z) +\i\, p_R X_R(\bar{z})}:\,,
\end{align}
where the left- and right-moving momenta $(p_L,p_R)$ are given by
\begin{align}
\label{eq:deformation_momentum}
    p_L =\frac{m}{R}+\frac{w\,R}{2}\,,\qquad
    p_R =\frac{m}{R}-\frac{w\,R}{2}\,,\qquad m,w\in\BZ\,.
\end{align}
Here, $m$ and $w$ are called the momentum and winding numbers, respectively.
For any vertex operator $V_{p_L,\,p_R}(z,\bar{z})$, the spin is an integer:
\begin{align}
    s = h-\bar{h} = \frac{p_L^2}{2}-\frac{p_R^2}{2} = m\, w\in\BZ\,.
\end{align}
The spin-statistics theorem shows that the theory $\CB$ has only bosonic excitations.

Generally, a bosonic theory should be invariant under the modular transformation.
In an $n$-dimensional compact boson, the modular invariance requires a set of left- and right-moving momenta $(p_L,p_R)$ to form an even self-dual lattice $\Lambda$ with respect to the diagonal Lorentzian metric $\tilde{\eta} = \mathrm{diag}(\Bone_n,-\Bone_n)$ (the corresponding inner product is denoted by $\circ$ in the following):
\begin{itemize}
    \item (Even) $\qquad\quad(p_L,p_R)\circ(p_L,p_R)\in2\BZ\qquad \text{for }~ (p_L,p_R)\in\Lambda$.
    \item (Self-dual) $\quad \Lambda^* := \{\,(p_L',p_R')\in\BR^{n,n}\mid (p_L,p_R)\circ (p_L',p_R')\in\BZ\,,\,(p_L,p_R)\in\Lambda\} = \Lambda$.
\end{itemize}
In a general dimension $n$, an even self-dual lattice specifies a compactification of bosons and defines a bosonic CFT of Narain type~\cite{Narain:1985jj,Narain:1986am}.

For a single compact boson ($n=1$) of radius $R$, the corresponding momentum lattice is
\begin{align}
    \Lambda(R) = \left\{\left(\frac{m}{R}+\frac{w\,R}{2},\,\frac{m}{R}-\frac{w\,R}{2}\right)\in\BR^2\;\middle|\; m,w\in\BZ\right\}\,.
\end{align}
One can easily check that $\Lambda(R)$ is even self-dual with respect to $\tilde{\eta} = \mathrm{diag}(1,-1)$.
The compact boson theory consists of the operators
\begin{align}
    \partial X(z)\,,\qquad T(z) = \partial X(z)\, \partial X(z) \,,\qquad V_{p_L,\,p_R}(z,\bar{z})\,.
\end{align}
Using the state-operator isomorphism, the Hilbert space of the compact boson is
\begin{align}
    \CH = \left\{\alpha_{-k_1}\cdots\alpha_{-k_r}\,\tilde{\alpha}_{-l_1}\cdots\tilde{\alpha}_{-l_s}\ket{p_L,p_R} \;\middle|\; (p_L,p_R)\in\Lambda(R)\right\}\,,
\end{align}
where $\alpha_k$ and $\tilde{\alpha}_k$ are the bosonic oscillators and $k_1,\cdots,k_r\in \BZ_{>0}$ and $l_1,\cdots,l_s\in \BZ_{>0}$.

The compact free boson at generic radius $R$ has a global $\U(1)_m\times \U(1)_w$ symmetry where the momentum $\U(1)_m$ and winding $\U(1)_w$ act on the compact boson as (see e.g., \cite{Ji:2019ugf})
\begin{align}
\begin{aligned}
    \U(1)_m: \quad X_L(z)&\to X_L(z) + \frac{R}{2}\,\theta_m\,, &\quad X_R(z)&\to X_R(z) + \frac{R}{2}\,\theta_m\,,\\
    \U(1)_w: \quad X_L(z)&\to X_L(z) + \frac{1}{R}\,\theta_w\,, &\quad X_R(z)&\to X_R(z) - \frac{1}{R}\,\theta_w\,,
\end{aligned}
\end{align}
where $\theta_{m,w}\sim \theta_{m,w}+2\pi$.
Then, the symmetries act on the vertex operators as
\begin{align}
    \begin{aligned}
    \U(1)_m:\quad V_{p_L,\,p_R}(z,\bar{z})&\to \!\!\!&e^{\i\, m\, \theta_m}\,&V_{p_L,\,p_R}(z,\bar{z})\,,\\
    \U(1)_w:\quad V_{p_L,\,p_R}(z,\bar{z})&\to \!\!\!&e^{\i\, w\, \theta_w}\,&V_{p_L,\,p_R}(z,\bar{z})\,,
    \end{aligned}
\end{align}
By setting $\theta_m=\pi$ and $\theta_w=\pi$, these $\U(1)$ symmetries reduce to the momentum $\BZ_2$ and winding $\BZ_2$ symmetries.

Let us focus on the momentum $\BZ_2$ symmetry and identify the untwisted and twisted Hilbert spaces. 
The momentum $\BZ_2$ symmetry acts on the vertex operators as
\begin{align}
    V_{p_L,\,p_R}(z,\bar{z})\to(-1)^m\,\, V_{p_L,\,p_R}(z,\bar{z})\,.
\end{align}
This action gives rise to the $\BZ_2$ grading of the momentum lattice $\Lambda(R)=\Lambda_0\cup\Lambda_1$: 
\begin{align}
\begin{aligned}
\label{eq:momenta01single}
    \Lambda_0 &= \left\{\left(\frac{m}{R}+\frac{w\,R}{2},\, \frac{m}{R}-\frac{w\,R}{2}\right)\in\BR^2\;\middle|\; m\in 2\BZ\,,\;w\in\BZ\right\}\,,\\
    \Lambda_1 &= \left\{\left(\frac{m}{R}+\frac{w\,R}{2},\, \frac{m}{R}-\frac{w\,R}{2}\right)\in\BR^2\;\middle|\; m\in2\BZ+1\,,\;w\in\BZ\right\}\,.
\end{aligned}
\end{align}
Correspondingly, the untwisted Hilbert space $\CH$ can be decomposed into the $\BZ_2$ even sector $S$ and odd sector $T$:
\begin{align}
\begin{aligned}
\label{eq:untwisted_hilbert_single}
    S &= \left\{\alpha_{-k_1}^{i_1}\cdots\alpha_{-k_r}^{i_r}\,\tilde{\alpha}_{-l_1}^{j_1}\cdots\tilde{\alpha}_{-l_s}^{j_s}\ket{p_L,p_R} \;\middle|\; (p_L,p_R)\in\Lambda_0\right\}\,,\\
    T &= \left\{\alpha_{-k_1}^{i_1}\cdots\alpha_{-k_r}^{i_r}\,\tilde{\alpha}_{-l_1}^{j_1}\cdots\tilde{\alpha}_{-l_s}^{j_s}\ket{p_L,p_R} \;\middle|\; (p_L,p_R)\in\Lambda_1\right\}\,.
\end{aligned}
\end{align}

By gauging the momentum $\BZ_2$ symmetry of the compact free boson, we obtain the following operators in the twisted sector (refer to Appendix A in \cite{Lin:2019kpn}):
\begin{align}
    V_{\tilde{p}_L,\,\tilde{p}_R}(z,\bar{z}) = :e^{\i\,\tilde{p}_L X_L(z) +\i\,\tilde{p}_R X_R(\bar{z})}:\,,
\end{align}
which are labeled by
\begin{align}
    \tilde{p}_L = \frac{k}{R} + \frac{l R}{2}\,,\qquad \tilde{p}_R = \frac{k}{R} - \frac{l R}{2}\,,\qquad k\in\BZ\,,~~l\in\BZ+\frac{1}{2}\,.
    \label{eq:twisted_momentum}
\end{align}
Note that the spin of the vertex operator $V_{\tilde{p}_L,\,\tilde{p}_R}(z,\bar{z})$ is $s=k\,l$.
Then, the spin of any operator in the twisted sector is at least a half-integer. Hence, the spin selection rule \eqref{eq:spin_selection} guarantees that the momentum $\BZ_2$ symmetry is non-anomalous, with which we can perform the orbifold and fermionization.

On the twisted Hilbert space, the momentum $\BZ_2$ symmetry acts as
\begin{align}\label{eq:orb-act-single}
    V_{\tilde{p}_L,\,\tilde{p}_R}(z,\bar{z})\to(-1)^k\,\,V_{\tilde{p}_L,\,\tilde{p}_R}(z,\bar{z})\,.
\end{align}
Then, we can decompose the momenta $(\tilde{p}_L,\tilde{p}_R)$ into two sectors by the $\BZ_2$ grading:
\begin{align}
\begin{aligned}
\label{eq:momenta23single}
    \Lambda_2 &= \left\{\left(\frac{k}{R} + \frac{l\, R}{2},\, \frac{k}{R} - \frac{l\, R}{2}\right)\in\BR^2\;\middle|\; k\in 2\BZ+1\,,\;l\in\BZ+\frac{1}{2}\right\}\,,\\
    \Lambda_3 &=\left\{\left(\frac{k}{R} + \frac{l\, R}{2},\, \frac{k}{R} - \frac{l\, R}{2}\right)\in\BR^2\;\middle|\; k\in 2\BZ\,,\;l\in\BZ+\frac{1}{2}\right\}\,.
\end{aligned}
\end{align}
This induces the decomposition of the twisted Hilbert sector into the $\BZ_2$ even sector $U$ and odd sector $V$:
\begin{align}
    \begin{aligned}
    \label{eq:twisted_hilbert_single}
    U &= \left\{\alpha_{-k_1}^{i_1}\cdots\alpha_{-k_r}^{i_r}\,\tilde{\alpha}_{-l_1}^{j_1}\cdots\tilde{\alpha}_{-l_s}^{j_s}\ket{p_L,p_R} \;\middle|\; (p_L,p_R)\in\Lambda_3\right\}\,,\\
    V &= \left\{\alpha_{-k_1}^{i_1}\cdots\alpha_{-k_r}^{i_r}\,\tilde{\alpha}_{-l_1}^{j_1}\cdots\tilde{\alpha}_{-l_s}^{j_s}\ket{p_L,p_R} \;\middle|\; (p_L,p_R)\in\Lambda_2\right\}\,.
    \end{aligned}
\end{align}

\begin{table}
  \centering
          \begin{subtable}[t]{0.45\textwidth}
        \centering
              \begin{tabular}{ccc}
              \toprule
                $\CB$  & untwisted  &  twisted  \\
                \hline 
                even  & $\Lambda_0$  & $\Lambda_3$ \\
                odd  & $\Lambda_1$   & $\Lambda_2$ \\ \bottomrule
              \end{tabular}
        \caption{Bosonic CFT}
      \end{subtable}
      \begin{subtable}[t]{0.45\textwidth}
            \centering
                  \begin{tabular}{ccc}
                  \toprule
                    $\CO$  & untwisted  &  twisted  \\
                    \hline 
                    even  & $\Lambda_0$  & $\Lambda_1$ \\
                    odd  & $\Lambda_3$   & $\Lambda_2$  \\\bottomrule
                  \end{tabular}
            \caption{Orbifold CFT}
      \end{subtable}
      \vspace*{0.5cm}
  \\
  \vskip\baselineskip
        \begin{subtable}[t]{0.45\textwidth}
            \centering
              \begin{tabular}{ccc}
              \toprule
                $\CF$  & NS sector  &  R sector  \\
                \hline 
                even  & $\Lambda_0$  & $\Lambda_1$ \\
                odd  & $\Lambda_2$   & $\Lambda_3$ \\\bottomrule
              \end{tabular}
            \caption{Fermionic CFT}
        \end{subtable}
        \begin{subtable}[t]{0.45\textwidth}
            \centering
              \begin{tabular}{ccc}
              \toprule
                $\widetilde{\CF}$  & NS sector  &  R sector  \\
                \hline 
                even  & $\Lambda_0$  & $\Lambda_3$ \\
                odd  & $\Lambda_2$   & $\Lambda_1$  \\\bottomrule
              \end{tabular}
              \caption{The other fermionic CFT}
        \end{subtable}
  \vspace{0.5cm}
  \caption{The $\BZ_2$ grading of the untwisted and twisted Hilbert space in terms of the underlying sets of momenta. Each Hilbert space can be constructed from the vertex operators associated with $\Lambda_i$ $(i=0,1,2,3)$ and the bosonic excitations. The orbifold and fermionization reduce to the swapping of the four sets $\Lambda_i$ of momenta.}
  \label{table:gauging_deforamtion}
\end{table}

The only difference between the four sectors $S, T, U, V$ in \eqref{eq:untwisted_hilbert_single} and \eqref{eq:twisted_hilbert_single} is the underlying set of left- and right-moving momenta.
Hence, the $\BZ_2$ grading in the untwisted and twisted Hilbert space can be understood in terms of $\Lambda_i$ ($i=0,1,2,3$):
\begin{align}
    S \leftrightarrow \Lambda_0\,,\qquad
    T \leftrightarrow \Lambda_1\,,\qquad
    U \leftrightarrow \Lambda_3\,,\qquad
    V \leftrightarrow \Lambda_2\,.\qquad
\end{align}
To compare with Table \ref{table:gauging}, we summarize the orbifold and fermionization by the momentum $\BZ_2$ symmetry in terms of the underlying sets of momenta in Table \ref{table:gauging_deforamtion}.
While the momentum lattice of the orbifold theory is
\begin{align}
    \Lambda_\CO = \Lambda_0\cup \Lambda_3\,,
\end{align}
that of the NS sector in the fermionized theory is
\begin{align}
\Lambda_\mathrm{NS} = \Lambda_0\cup \Lambda_2\,.  
\end{align}
One can check that $\Lambda_\CO$ is even self-dual and $\Lambda_\mathrm{NS}$ is odd self-dual with respect to the inner product $\circ$.
This meets the general requirement that a bosonic theory has an even self-dual momentum lattice and a fermionic theory has an odd self-dual lattice.

Let us reproduce the $\BZ_2$ grading of the Hilbert spaces more systematically.
Suppose an element $\chi = (R/2,-R/2)\in \Lambda(R)$.
Then, each $\Lambda_i$ ($i=0,1$) in \eqref{eq:momenta01single} can be characterized by 
\begin{align}
    \begin{aligned}
        \Lambda_i = \left\{\, (p_L,p_R)\in\Lambda(R)\mid (p_L,p_R)\circ \chi=i\;\;\mathrm{mod}\,\;2\right\}\,.
    \end{aligned}
\end{align}
Furthermore, $\Lambda_2$ and $\Lambda_3$ in \eqref{eq:momenta23single} can be endowed with
\begin{align}
   \begin{aligned}
    \Lambda_2 = \Lambda_1 + \frac{\chi}{2}\,,\qquad \Lambda_3 = \Lambda_0 +\frac{\chi}{2}\,.
\end{aligned}
\end{align}
Therefore, $\Lambda_2$ and $\Lambda_3$, which specify the $\BZ_2$ grading of the twisted sector, are a shift of $\Lambda_1$ and $\Lambda_2$ by $\chi/2$, respectively. 
Once a vector $\chi$ is given, then it uniquely determines $\Lambda_i$ ($i=0,1,2,3$) as well as the $\BZ_2$ grading of the untwisted and twisted Hilbert space.

In general, we can choose a non-anomalous $\BZ_2$ symmetry in $\U(1)_m\times \U(1)_w$ and take the orbifold and fermionization by the symmetry.
In terms of the momentum lattice, the choice of a non-anomalous $\BZ_2$ symmetry corresponds to the choice of $\chi\in\Lambda(R)$. 
For example, the momentum and winding $\BZ_2$ symmetries are specified by $\chi = (R/2,-R/2)$ and $\chi = (R/2,R/2)$, respectively.
In the next section, we generalize this discussion into higher-dimensional lattices and formulate the lattice modification, which leads to the $\BZ_2$ gauging of Narain CFTs.

\subsection{General formulation of lattice modification}
\label{ss:general_modification}

In the previous section, we introduced the modifications of the momentum lattice as orbifold and fermionization in a compact boson. 
We call this operation lattice modification, which is often used by mathematicians to make new self-dual lattices out of existing ones \cite{conway2013sphere}.
In relation to string theory, lattice modification also has been studied in \cite{Lerche:1988np,gannon1991lattices}.
The main aim of this section is to provide a modern description of lattice modification related to orbifold and fermionization of Narain CFTs.

Before discussing lattice modification, we clarify several notions and properties of lattices.
Let $\Lambda\subset\BR^N$ be a lattice with a non-degenerate symmetric bilinear form $g$. We denote the inner product with respect to the metric $g$ as $\oslash$. 
We do not specify its signature so that Euclidean and Lorentzian metrics can be considered.

Associated with the inner product $\oslash$, the dual lattice is
\begin{align}
    \Lambda^* = \left\{\lambda'\in\BR^N\mid \lambda\oslash \lambda'\in\BZ\, ,\,\;\lambda\in\Lambda\right\}\,.
\end{align}
If $\Lambda\subset\Lambda^*$, $\Lambda$ is called \emph{integral}. If $\Lambda = \Lambda^*$, $\Lambda$ is called \emph{self-dual}.
An integral lattice $\Lambda$ can be further classified into an even and odd lattice.
\begin{itemize}
    \item If any $\lambda\in \Lambda$ satisfies $\lambda\oslash \lambda\in2\BZ$, $\Lambda$ is called \emph{even}.
    \item If there exists $\lambda\in\Lambda$ such that $\lambda\oslash\lambda\in2\BZ+1$, $\Lambda$ is called \emph{odd}.
\end{itemize}
By regarding $\Lambda$ as a momentum lattice, an even self-dual lattice $\Lambda$ defines a bosonic theory and an odd self-dual lattice $\Lambda$
defines a fermionic theory \cite{Lerche:1988np}.

Since we are interested in Narain CFTs, we suppose that $\Lambda$ is even self-dual with respect to the metric $g$.
Let us pick up a lattice vector $\chi\in\Lambda$ whose half is not an element of $\Lambda$: $\delta = \frac{\chi}{2}\notin\Lambda$.
We assume $\delta\oslash\delta \in \BZ$ so that the CFT associated with the lattice $\Lambda$ has a non-anomalous $\BZ_2$ symmetry.
Otherwise, the twisted sector would consist of operators with the spin $s\in\frac{1}{4}+\frac{\BZ}{2}$ and it implies the $\BZ_2$ symmetry is anomalous from the spin selection rule \eqref{eq:spin_selection}.
Note that, in the non-anomalous case,  $\chi\oslash\chi\in 4\BZ$ holds.

Generalizing the case of a single compact boson in the previous subsection, we divide $\Lambda$ into two parts as 
\begin{align}
    \Lambda = \Lambda_0\cup\Lambda_1\,,
\end{align}
where
\begin{align}
\begin{aligned}
\label{eq:lattice_divide}
    \Lambda_0 &= \left\{\lambda\in\Lambda\,|\,\chi\oslash\lambda = 0 \;\;\mathrm{mod}\;\,2\right\}\,,\\
    \Lambda_1 &= \left\{\lambda\in\Lambda\,|\,\chi\oslash\lambda = 1 \;\;\mathrm{mod}\;\,2\right\}\,.
\end{aligned}
\end{align}
From the assumption $\chi\oslash\chi\in4\BZ$, the vector $\chi\in\Lambda$ is contained in $\Lambda_0$: $\chi\in\Lambda_0$.
Note that $\Lambda_0$ is closed under addition, so it forms a lattice by itself.
Using the shift vector $\delta$, we introduce the following two sets: 
\begin{align}
    \begin{aligned}
    \label{eq:lattice_shift_23}
    \Lambda_2 &= 
    \begin{dcases}
    \Lambda_1 + \delta &\quad (\delta\oslash\delta \in 2\BZ)\,,\\
    \Lambda_0 + \delta &\quad (\delta\oslash\delta\in 2\BZ+1)\,,
    \end{dcases}\\[0.1cm]
    \Lambda_3 &= 
    \begin{dcases}
    \Lambda_0 + \delta &\quad (\delta\oslash\delta\in 2\BZ)\,,\\
    \Lambda_1 + \delta &\quad (\delta\oslash\delta\in 2\BZ+1)\,.
    \end{dcases}
    \end{aligned}
\end{align}

As in the previous subsection, we define
\begin{align}
\label{eq:lattice_def}
    \Lambda_\CO = \Lambda_0\cup\Lambda_3\,,\qquad
    \Lambda_\mathrm{NS} = \Lambda_0\cup\Lambda_2\,.
\end{align}
In what follows, we verify that $\Lambda_\CO$ is even self-dual and $\Lambda_\mathrm{NS}$ is odd self-dual for a general metric $g$.
The following proposition is fundamental.

\begin{proposition}
The dual lattice $\Lambda_0^*$ of $\Lambda_0$ with respect to the metric $g$ is given by
\begin{align}
    \Lambda_0^* = \Lambda_0\cup\Lambda_1\cup\Lambda_2\cup\Lambda_3\,,
\end{align}
where $\Lambda_i$ $(i=0,1,2,3)$ are defined by \eqref{eq:lattice_divide} and \eqref{eq:lattice_shift_23}.
\label{prop:dual_lattice_lambda0}
\end{proposition}
\begin{proof}
First, we prove $\Lambda_0^*\supset\Lambda_0\cup\Lambda_1\cup\Lambda_2\cup\Lambda_3$.
Since $\Lambda = \Lambda_0\cup\Lambda_1$ is self-dual,
\begin{align}
    \Lambda_0^*\supset (\Lambda_0 \cup \Lambda_1)^\ast = \Lambda_0 \cup \Lambda_1\ .
\end{align}
By taking the inner product with $\lambda_0\in \Lambda_0$, we check that any element in $\Lambda_0+\delta$ and $\Lambda_1 + \delta$ is in $\Lambda_0^*$.
Suppose $\lambda'_0 + \delta\in\Lambda_0 + \delta$.
The inner product between them is
\begin{align}
    (\lambda'_0 + \delta)\oslash\lambda_0 = \lambda'_0\oslash\lambda_0 + \delta\oslash\lambda_0\in \BZ\ ,
\end{align}
where the first term is even and the second term is integer by definition of $\Lambda_0$.
Suppose $\lambda_1 + \delta\in \Lambda_1 + \delta$. Then
\begin{align}
    (\lambda_1 + \delta)\oslash\lambda_0 = \lambda_1 \oslash\lambda_0 + \delta\oslash\lambda_0 \in\BZ\ ,
\end{align}
where the first term is integer due to the integral lattice $\Lambda$ and the second term is also integer by definition of $\Lambda_0$.
Then, any lattice vector in $\Lambda_0\cup\Lambda_1\cup\Lambda_2\cup\Lambda_3$ has an integer inner product with all vectors in $\Lambda_0$.
This means $\Lambda_0^*\supset\Lambda_0\cup\Lambda_1\cup\Lambda_2\cup\Lambda_3$.

From now on, we consider $\Lambda_0^*\subset\Lambda_0\cup\Lambda_1\cup\Lambda_2\cup\Lambda_3$.
Suppose $\lambda\in\Lambda_0^*$.
We will consider the following two cases: (i) $\lambda\oslash\lambda_1\in\BZ$ for all $\lambda_1\in\Lambda_1$ and (ii) there exist $\lambda_1\in\Lambda_1$ such that $\lambda\oslash\lambda_1\notin\BZ$.

In the first case (i), the vector $\lambda\in\Lambda_0^*$ has an integer inner product with a lattice vector in $\Lambda_0$ and $\Lambda_1$. Then, we obtain
\begin{align}
    \lambda\in(\Lambda_0\cup\Lambda_1)^* = \Lambda_0\cup\Lambda_1\ .
\end{align}

In the second case (ii), let us consider a lattice vector $\lambda_1\in\Lambda_1$ that satisfies $\lambda\oslash\lambda_1\notin\BZ$.
Since $2\lambda_1\in\Lambda_0$, we have $\lambda\oslash2\lambda_1\in\BZ$.
This implies 
\begin{align}
\label{eq:dual_half_odd}
    \lambda\oslash\lambda_1\in\BZ +\frac{1}{2}\ .
\end{align}
Let $\lambda'$ be a vector $\lambda' = \lambda-\delta$ and $\nu\in\Lambda_0\cup\Lambda_1$.
Then, the inner product between them is
\begin{align}
\label{eq:dual_inner_prod}
    \lambda'\oslash\nu = \lambda\oslash\nu - \delta\oslash\nu\ .
\end{align}
The first term is
\begin{align}
    \lambda\oslash\nu  \in
    \begin{dcases}
    \BZ &\quad (\nu\in\Lambda_0)\,,\\
    \BZ + \frac{1}{2} &\quad (\nu\in\Lambda_1)\,,
    \end{dcases}
\end{align}
where, in the second case $(\nu\in\Lambda_1)$, we use the fact that any lattice vector $\nu\in\Lambda_1$ can be shifted to the vector $\lambda_1\in\Lambda_1$ defined above by an appropriate vector $\rho\in\Lambda_0$: $\nu = \lambda_1 + \rho$. Since $\Lambda_0$ is integral and \eqref{eq:dual_half_odd} holds, we obtain the result $\lambda\oslash\nu\in\BZ + \frac{1}{2}$.
On the other hand, the second term in \eqref{eq:dual_inner_prod} is
\begin{align}
    \delta\oslash\nu \in
    \begin{dcases}
    \BZ & (\nu\in\Lambda_0)\,,\\
    \BZ + \frac{1}{2} & (\nu\in\Lambda_1)\,.
    \end{dcases}
\end{align}
Therefore, we have $\lambda'\oslash\nu\in\BZ$ for all vectors $\nu\in\Lambda_0\cup\Lambda_1$, implying $\lambda'=\lambda-\delta\in(\Lambda_0\cup\Lambda_1)^* = \Lambda_0\cup\Lambda_1$ or equivalently $\lambda\in\Lambda_2\cup\Lambda_3$.
Since the above two cases fill in all cases, if the vector $\lambda\in\Lambda_0^*$, then $\lambda\in\Lambda_0\cup\Lambda_1\cup\Lambda_2\cup\Lambda_3$. This results in $\Lambda_0^*\subset\Lambda_0\cup\Lambda_1\cup\Lambda_2\cup\Lambda_3$. 
\end{proof}

\begin{proposition}
Let $\Lambda_\CO$ and $\Lambda_\mathrm{NS}$ be the shifted lattices defined by \eqref{eq:lattice_def}. Then, $\Lambda_\CO$ and $\Lambda_\mathrm{NS}$ are self-dual with respect to the metric $g$.
\label{prop:self-dual_general}
\end{proposition}
\begin{proof}
We divide the original lattice $\Lambda$ into two parts $\Lambda = \Lambda_0 \cup \Lambda_1$. 
The dual lattice of $\Lambda_0$ is given by Proposition \ref{prop:dual_lattice_lambda0}:
\begin{align}
    \Lambda_0^* = \Lambda_0\cup\Lambda_1\cup\Lambda_2\cup\Lambda_3\ .
\end{align}
First, we consider the shifted lattice $\Lambda_\mathrm{NS} = \Lambda_0\cup\Lambda_2$.
It is straightforward to verify that the shifted lattice is integral, so we have
\begin{align}
    \Lambda_0\cup\Lambda_2 = \Lambda_\mathrm{NS}\subset \Lambda_\mathrm{NS}^*\ .
\end{align}
The dual lattice $\Lambda_\mathrm{NS}^*$ of the shifted lattice $\Lambda_\mathrm{NS}$ satisfies
\begin{align}
\label{eq:constraint_selfdual_general}
    \Lambda^*_\mathrm{NS} = (\Lambda_0\cup\Lambda_2)^* \subset \Lambda_0^* = \Lambda_0\cup\Lambda_1\cup\Lambda_2\cup\Lambda_3\ .
\end{align}
To specify the dual lattice $\Lambda_\mathrm{NS}^*$, we consider the inner products between the shifted lattice $\Lambda_\mathrm{NS}$ and $\Lambda_i$ ($i=0,1,2,3$).
Since $\Lambda_0$ has integer inner products with $\Lambda_i$ for all $i$ by Proposition \ref{prop:dual_lattice_lambda0}, we only need to examine the inner products between $\Lambda_2$ and $\Lambda_i$.

Suppose $\lambda\in\Lambda_2$ and $\lambda_i\in \Lambda_i$.
For $i=0,2$, we find the inner product $\lambda\oslash\lambda_i \in \BZ$ ($i=0,2$) because $\Lambda_\mathrm{NS}=\Lambda_0\cup\Lambda_2$ is an integral lattice.
For $i=1$, we consider two cases: $\delta\oslash\delta\in2\BZ$ and $\delta\oslash\delta\in2\BZ+1$.
In the case with $\delta\oslash\delta\in2\BZ$, a lattice vector is given by $\lambda_2=\lambda_1' +\delta\in\Lambda_1 + \delta$.
In this case, we obtain the following:
\begin{align}
    (\lambda_1' + \delta)\oslash\lambda_1 = \lambda_1'\oslash\lambda_1 + \delta\oslash\lambda_1 \in \BZ+\frac{1}{2}\ ,
\end{align}
which follows from the definition of $\Lambda_1$.
For the other case ($\delta\oslash\delta\in2\BZ+1$), we have $\lambda_2 = \lambda_0' + \delta\in\Lambda_0 + \delta$. Then, we obtain
\begin{align}
    (\lambda_0' +\delta)\oslash\lambda_1 = \lambda_0'\oslash\lambda_1 + \delta\oslash\lambda_1\in \BZ+\frac{1}{2}\ ,
\end{align}
where the first term is integer due to the integrality of the lattice $\Lambda=\Lambda_0\cup\Lambda_1$ and the second term is half-odd by definition of $\Lambda_1$.
For $i=3$, there are also two cases: $\delta\oslash\delta\in2\BZ$ and $\delta\oslash\delta\in2\BZ+1$.
For $\delta\oslash\delta\in 2\BZ$, we can write $\lambda_2 =\lambda_1 + \delta \in\Lambda_1 + \delta$ and $\lambda_3 = \lambda_0 + \delta\in\Lambda_0 + \delta$.
Then, the inner product between them is
\begin{align}
\label{eq:innerprd_self_general}
    \lambda_2\oslash\lambda_3 = \lambda_1\oslash\lambda_0 + \lambda_1\oslash\delta + \delta\oslash\lambda_0 + \delta\oslash\delta\in \BZ + \frac{1}{2}\ ,
\end{align}
where the second term is half-odd by definition of $\Lambda_1$.
For the other case ($\delta\oslash\delta\in2\BZ+1$),
we have $\lambda_2 =\lambda_0+\delta\in\Lambda_0 + \delta$ and $\lambda_3 = \lambda_1 + \delta\in\Lambda_1 + \delta$. Then we get the same result as \eqref{eq:innerprd_self_general}: $\lambda_2\oslash\lambda_3 \in \BZ+\frac{1}{2}$.

Then, the inner product between an element in $\Lambda_2$ and an element in $\Lambda_1\cup\Lambda_3$ is always half-odd. 
Therefore, any lattice vector in $\Lambda_1\cup \Lambda_3$ is not an element of the dual lattice $\Lambda_\mathrm{NS}^*$, so we can strengthen the constraint \eqref{eq:constraint_selfdual_general}:
\begin{align}
    \Lambda_\mathrm{NS}^* \subset \Lambda_0\cup\Lambda_2\ .
\end{align}
In summary, we have the following relation:
\begin{align}
    \Lambda_0\cup\Lambda_2\subset \Lambda_\mathrm{NS}^*\subset\Lambda_0\cup\Lambda_2\ ,
\end{align}
which implies that the integral lattice $\Lambda_\mathrm{NS}$ is self-dual: $\Lambda_\mathrm{NS}^* = \Lambda_0\cup\Lambda_2 = \Lambda_\mathrm{NS}$.
A similar discussion can be applied for the other shifted lattice $\Lambda_\CO$.
\end{proof}

From Proposition \ref{prop:self-dual_general}, the following corollary is obvious.

\begin{corollary}
The shifted lattices $\Lambda_\CO$ and $\Lambda_\mathrm{NS}$ are integral with respect to the metric $g$.
\label{corollary:integral}
\end{corollary}

Both the shifted lattices $\Lambda_\CO$ and $\Lambda_\mathrm{NS}$ share the self-duality with respect to $g$.
The following proposition clarifies the important difference between them.

\begin{proposition}
The shifted lattice $\Lambda_\CO$ is even and $\Lambda_\mathrm{NS}$ is odd with respect to the metric $g$.
\label{prop:norm_general}
\end{proposition}

\begin{proof}
From Corollary \ref{corollary:integral}, both $\Lambda_\CO$ and $\Lambda_\mathrm{NS}$ are integral, so we can classify them into even and odd lattices.
The norm of a lattice vector in $\Lambda_0$ is even due to the evenness of the original lattice $\Lambda$.
In what follows, we consider the norm of an element in $\Lambda_2$ and $\Lambda_3$.

First, we focus on the lattice $\Lambda_\CO = \Lambda_0\cup\Lambda_3 = \Lambda_0\cup(\Lambda_0+\delta)$ when $\delta\oslash\delta\in 2\BZ$.
Suppose $\lambda_0 + \delta\in\Lambda_0 + \delta$.
Its norm is given by 
\begin{align}
    (\lambda_0 + \delta)\oslash(\lambda_0 + \delta) = \lambda_0\oslash\lambda_0 + \delta\oslash\delta+ \chi\oslash\lambda_0 \in 2\BZ\ ,
\end{align}
where the first term is even due to the evenness of $\Lambda$, the second term is even by assumption, and the third term is also even due to the definition of $\Lambda_0$.
If $\delta\oslash\delta\in 2\BZ+1$, then the shifted lattice $\Lambda_\CO$ is $\Lambda_\CO = \Lambda_0\cup(\Lambda_1 + \delta)$.
Suppose $\lambda_1 + \delta\in\Lambda_1 + \delta$.
The norm of this vector is 
\begin{align}
    (\lambda_1 + \delta)\oslash(\lambda_1 + \delta) = \lambda_1\oslash\lambda_1 + \delta\oslash\delta +\chi\oslash\lambda_1 \in 2\BZ\ ,
\end{align}
where the first term is even due to the evenness of $\Lambda$ and the second term is odd by assumption, the third term is also odd by definition of $\Lambda_1$.
Since the norm of any lattice vector in $\Lambda_\CO$ is always even, the shifted lattice $\Lambda_\CO$ is an even lattice.

Next, we treat the other shifted lattice $\Lambda_\mathrm{NS}$. For $\delta\oslash\delta\in2\BZ$, the shifted lattice is given by $\Lambda_\mathrm{NS} = \Lambda_0\cup(\Lambda_1 + \delta)$. 
Suppose $\lambda_1 + \delta \in \Lambda_1 + \delta$.
The norm of this vector is
\begin{align}
    (\lambda_1 + \delta) \oslash (\lambda_1 + \delta) = \lambda_1\oslash\lambda_1 + \delta\oslash\delta + \chi\oslash\lambda_1 \in 2\BZ + 1\ ,
\end{align}
where the first term is even by the evenness of $\Lambda$ and the second term is even by assumption, the third term is odd due to the definition of $\Lambda_1$.
For $\delta\oslash\delta\in 2\BZ + 1$, $\Lambda_\mathrm{NS} = \Lambda_0 \cup (\Lambda_0 + \delta)$. Suppose a lattice vector $\lambda_0 + \delta \in \Lambda_0 + \delta$. Its norm is
\begin{align}
    (\lambda_0 + \delta)\oslash(\lambda_0 + \delta) = \lambda_0 \oslash\lambda_0 + \delta\oslash\delta + \chi\oslash\lambda_0\in2\BZ +1\ ,
\end{align}
where $\lambda_0\oslash\lambda_0\in2\BZ$ by the evenness of $\Lambda$ and the second term is odd by assumption, the third term is even by definition of $\Lambda_0$.
The norm of any element in $\Lambda_2$ is odd, so the shifted lattice $\Lambda_\mathrm{NS}$ is an odd lattice with respect to the metric $g$.
\end{proof}

Combining Proposition \ref{prop:self-dual_general} and \ref{prop:norm_general}, we arrive at the following theorem.
\begin{theorem}
Let $\Lambda_\CO$ and $\Lambda_\mathrm{NS}$ be the shifted lattices defined by \eqref{eq:lattice_def}.
Then $\Lambda_\CO$ is even self-dual and $\Lambda_\mathrm{NS}$ is odd self-dual with respect to the metric $g$.
\label{prop:deformation_theorem}
\end{theorem}

Theorem \ref{prop:deformation_theorem} guarantees that lattice modifications of an even self-dual lattice $\Lambda$ yield new even self-dual lattice $\Lambda_\CO$ and odd self-dual lattice $\Lambda_\mathrm{NS}$.
In terms of CFTs, $\Lambda_\CO$ and $\Lambda_\mathrm{NS}$ correspond to the momentum lattices of the new bosonic and fermionic CFTs, which are obtained by orbifold and fermionization of the original bosonic CFT.
Theorem \ref{prop:deformation_theorem} suggests that this procedure of the orbifold and fermionization holds for any even self-dual lattice $\Lambda$ given a non-anomalous $\BZ_2$ symmetry.
In section \ref{sec:gauging_code_CFT}, we apply the lattice formulation of orbifold and fermionization to Narain CFTs constructed from quantum stabilizer codes and systematically compute the partition functions of orbifold and fermionized theory.

\section{Gauging $\BZ_2$ symmetry of Narain code CFTs}
\label{sec:gauging_code_CFT}

This section is devoted to the orbifold and fermionization of Narain code CFTs. Narain code CFTs are a class of Narain CFTs constructed from quantum codes. In section \ref{ss:review_code_cft}, we review the stabilizer formalism of quantum codes based on finite fields $\BF_p$ and the construction of Narain code CFTs for the cases with an odd prime number $p$ and $p=2$.
Section \ref{ss:gauging_code_oddprime} and \ref{ss:gauging_code_p=2} consider the $\BZ_2$ gauging for the case with an odd prime $p$ and $p=2$, respectively.
We also give a systematic way to compute the partition functions of the orbifold and fermionized theory.

\subsection{Narain code CFTs}
\label{ss:review_code_cft}

In this section, we review the construction of bosonic CFTs from quantum stabilizer codes \cite{Dymarsky:2020qom,Kawabata:2022jxt}.
Let $\BF_p = \BZ/p\,\BZ$ be a finite field for a prime number $p$ and $\{\, \ket{x}\,\big|\, x\in \BF_p\, \}$ be an orthonormal basis for a Hilbert space $\BC^p$. 
The Pauli group acting on $\BC^p$ is generated by the operators $X$ and $Z$ which satisfy the relations \cite{Gottesman:1997zz}: $X\,\ket{x} = \ket{x+1},~  Z\,\ket{x} = \omega^x\,\ket{x}$,
where $\omega = \exp(2\pi \i/p)$ and $x$ is defined modulo $p$.
$X$ and $Z$ are generalizations of the Pauli matrices acting on a qubit system ($p=2$) to a qudit system ($p\ge 2$).
Let $\Ba = (a_1, \cdots, a_n), \Bb = (b_1, \cdots, b_n)$ be vectors in $\BF_p^n$ and define an operator $g(\Ba, \Bb)$ acting on $n$-qudit system by
\begin{align}
    g(\Ba, \Bb) \equiv X^{a_1} Z^{b_1}\otimes \cdots \otimes X^{a_n} Z^{b_n}\ .
\end{align}
Here we omit a phase factor for $g(\Ba, \Bb)$ for simplicity.
A pair of two error operators do not commute in general but satisfy
\begin{align}
    g(\Ba, \Bb) \, g(\Ba', \Bb')  = \omega^{-\Ba\cdot \Bb' + \Ba'\cdot \Bb}\,g(\Ba', \Bb') \, g(\Ba, \Bb) \ ,
\end{align}
where $\Ba\cdot\Bb = \sum_{i=1}^n a_i\,b_i$.
Thus $g(\Ba, \Bb)$ and $g(\Ba', \Bb')$ commute if $\Ba\cdot \Bb' - \Ba'\cdot \Bb = 0$.
An $[[n, k]]$ quantum stabilizer code $V_\mathsf{S}$ is a subspace of $(\BC^p)^{\otimes n}$ fixed by the stabilizer group $\mathsf{S} = \langle\, g_1, \cdots, g_{n-k}\,\rangle$ generated by a commuting set of operators $g_i = g(\Ba^{(i)}, \Bb^{(i)})~(i=1, \cdots, n-k)$ \cite{Gottesman:1996rt,Gottesman:1998se,knill1996non,knill1996group,rains1999nonbinary}:
\begin{align}
    V_\mathsf{S} = \left\{\, \ket{\psi}\in (\BC^p)^{\otimes n}\,\, \big|\,\, g\,\ket{\psi} = \ket{\psi},\, \forall g \in \mathsf{S} \,\right\} \ .
\end{align}
Since $V_\mathsf{S}$ is the simultaneous eigenspace of the $n-k$ independent operators $g_i$, it is a $p^k$-dimensional subspace in $(\BC^p)^{\otimes n}$.
The stabilizer group $\mathsf{S}$ can be encoded into a $(n-k) \times 2n$ matrix of rank $n-k$ over $\BF_p$:
\begin{align}
    \SH = 
    \left[
        \begin{array}{c|c}
    	~\Ba^{(1)} ~& ~\Bb^{(1)}~ \\
            ~\Ba^{(2)} ~& ~\Bb^{(2)}~ \\
            ~\vdots ~& ~\vdots~ \\
            ~\Ba^{(n-k)}~ & ~\Bb^{(n-k)} ~
        \end{array}
    \right] \ .
\end{align}
The vertical line indicates that the left and right vectors in each row correspond to the $X$ and $Z$ operators for the stabilizer generator $g(\Ba, \Bb)$, respectively.
The commutativity of the generators $g_i$ yields the condition:
\begin{align}\label{stabilizer_condition}
    \SH\,\SW\,\SH^T = 0 \quad \text{mod}~p \ ,
    \qquad \SW 
        =
        \left[
        \begin{array}{cc}
    	0 & I_{n} \\
            -I_n & 0
        \end{array}
        \right]\ ,
\end{align}
where $I_n$ is the $n \times n$ identity matrix and $0$ is a zero matrix.

The matrix representation of the stabilizer generators manifests an intriguing relation between quantum stabilizer codes and classical codes.
An $[n,k]$ linear code $C$ over $\BF_p$ is defined as a $p^k$-dimensional linear subspace of $\BF_p^{n}$ generated by $k$ row vectors of a $k \times n$ matrix $G$ of rank $k$ over $\BF_p$:
\begin{align}
    C = \left\{\, c\in \BF_p^n\,\big|\, c = x\,G, ~x \in \BF_p^k\,\right\} \ .
\end{align}
Here, $G$ is called the generator matrix of the code $C$.
Now consider a classical code $\CC$ whose generator matrix is given by the matrix $G_\SH = \SH$ associated with a stabilizer code:
\begin{align}\label{eq:CC-stabilizer}
    \CC = \left\{\, c\in \BF_p^{2n}\,\big|\, c = x\,G_\SH, ~x \in \BF_p^{n-k}\,\right\} \ .
\end{align}
Having the construction of Narain CFTs in mind, we introduce the Lorentzian inner product $\odot$ for a pair of codewords $c, c' \in \CC$ by
\begin{align}\label{eq:eta-def}
     u\odot v = u\,\eta\,v^T \ , \qquad
    \eta 
        =
        \left[
        \begin{array}{cc}
    	0 & I_{n} \\
            I_n & 0
        \end{array}
        \right] \ .
\end{align}
The dual code $\CC^\perp$ with respect to the Lorentzian inner product $\odot$ is defined by
\begin{align}
    \CC^\perp =  \left\{\, c'\in \BF_p^{2n}\,\big|\, c'\odot c = 0~\text{mod}~p\,, ~c \in \CC\,\right\} \ .
\end{align}
The Lorentzian linear code $\CC$ is called \emph{self-orthogonal} if $\CC\subset\CC^\perp$ and \emph{self-dual} if $\CC = \CC^\perp$.
For the binary case $p=2$, we can further classify self-orthogonal codes.
In particular, we call $\CC$ \emph{doubly-even} if any codeword $c\in\CC$ satisfies $c\cdot c\in4\BZ$.

We can further extend the code $\CC$ to a $2n$-dimensional Lorentzian lattice $\Lambda(\CC)$ through the so-called Construction A \cite{conway2013sphere}:
\begin{align}\label{eq:constructon-A-lattice}
    \Lambda(\CC) = \left\{ \frac{c + p\,m}{\sqrt{p}}\in \BR^{n,n}\; \middle|\; c\in \CC, ~ m\in \BZ^{2n} \right\}\ ,
\end{align}
which is equipped with the Lorentzian inner product $\odot$.
In the construction of Narain CFTs from Lorentzian even self-dual lattices, the following theorems are essential.

\begin{theorem}[{\cite{Yahagi:2022idq}}]
Let $p$ be an odd prime and $\CC$ a self-dual code over $\BF_p$ with respect to $\eta$. Then, the Construction A lattice $\Lambda(\CC)$ is even self-dual with respect to $\eta$.
\label{prop:even_self_dual_odd}
\end{theorem}

\begin{theorem}[{\cite{Dymarsky:2020qom,Kawabata:2022jxt}}]
    Let $\CC$ be a doubly-even self-dual code over $\BF_2$ with respect to $\eta$. Then, the Construction A lattice $\Lambda(\CC)$ is even self-dual with respect to $\eta$.
    \label{prop:even_self_dual_p=2}
\end{theorem}

A Lorentzian even self-dual lattice $\Lambda(\CC)$ is related to the momentum lattice $\widetilde{\Lambda}(\CC)$ of a Narain CFT by the orthogonal transformation
\begin{align}
\label{eq:orthogonal_transf}
    p_L = \frac{\lambda_1 + \lambda_2}{\sqrt{2}}\,,\qquad
    p_R = \frac{\lambda_1 - \lambda_2}{\sqrt{2}}\,,
\end{align}
where $(\lambda_1,\lambda_2)\in\Lambda(\CC)$.
The set of vertex operators of the Narain code CFT is given by $V_{p_L,\,p_R} (z,\bar{z}) = e^{\i p_L\cdot X(z) +\i p_R \cdot \bar{X}(\bar{z})}$ where $(p_L,p_R)\in\widetilde{\Lambda}(\CC)$. The corresponding momentum state is given by $\ket{p_L,p_R}$ through the state-operator isomorphism.
These are eigenstates of the Virasoro generators $L_0$ and $\bar{L}_0$ with eigenvalues $h=p_L^2/2$ and $\bar{h}=p_R^2/2$.
Combining with the bosonic excitations, we obtain the whole Hilbert space of the Narain code CFT:
\begin{align}
    \CH(\CC) = \left\{\alpha_{-k_1}^{i_1}\cdots\alpha_{-k_r}^{i_r}\,\tilde{\alpha}_{-l_1}^{j_1}\cdots\tilde{\alpha}_{-l_s}^{j_s}\ket{p_L,p_R}\;\middle|\;(p_L,p_R)\in\widetilde{\Lambda}(\CC)\right\}\,,
\end{align}
with $k_1,\cdots,k_r\in \BZ_{>0}$ and $l_1,\cdots,l_s\in \BZ_{>0}$.

The partition function of the Narain code CFT can be written in terms of the complete weight enumerator polynomial of a Lorentzian linear code $\CC$ and the lattice theta function of the Construction A lattice $\widetilde{\Lambda}(\CC)$.
We define the complete weight enumerator polynomial of $\CC$
\begin{align}
\label{eq:complete_weight}
    W_\CC\left(\{x_{ab}\}\right) = \sum_{c\,\in\,\CC} \;\prod_{(a,b)\,\in\,\BF_p\times\BF_p} x_{a b}^{\mathrm{wt}_{ab}(c)}\,,
\end{align}
where $\mathrm{wt}_{ab}(c)$ is the number of components $c_i = (\alpha_i,\beta_i)\in\BF_p\times\BF_p$ that equal to $(a,b)\in\BF_p\times\BF_p$ for a codeword $c\in\CC$:
\begin{align}
    \mathrm{wt}_{ab}(c) = \left|\,\left\{i\,|\, c_i = (a,b)\right\}\,\right|\,.
\end{align}
Also, we define the lattice theta function of a momentum lattice $\widetilde{\Lambda}(\CC)$ by
\begin{align}
\label{eq:lattice_theta}
    \Theta_{\widetilde{\Lambda}(\CC)}(\tau,\bar{\tau}) = \sum_{(p_L,p_R)\,\in\,\widetilde{\Lambda}(\CC)} q^{\frac{p_L^2}{2}}\,\bar{q}^{\,\frac{p_R^2}{2}}\,,
\end{align}
where $q=e^{2\pi\i\tau}$ and $\tau$ is the torus modulus.
Then, these quantities are related to the partition function in the following way.

\begin{proposition}[{\cite{Yahagi:2022idq,Angelinos:2022umf,Kawabata:2022jxt}}]
Let $\CC\subset\BF_p^n\times\BF_p^n$ be a classical code with the complete enumerator polynomial $W_\CC$.
Then, the partition function of the Narain CFT constructed from the code $\CC$ is
\begin{align}
\label{eq:partition_theta_enumerator}
    Z_\CC (\tau,\bar{\tau}) =\frac{\Theta_{\widetilde{\Lambda}(\CC)}(\tau,\bar{\tau})}{|\eta(\tau)|^{2n}} =  \frac{1}{|\eta(\tau)|^{2n}} \,W_\CC(\{\psi_{ab}^+\})\,,
\end{align}
where the variables $x_{ab}$ in the complete enumerator polynomial are replaced by
\begin{align}
    \psi_{a b}^+(\tau,\bar{\tau}) =  \sum_{k_1,k_2\in\BZ} q^{\frac{p}{4}\left(\frac{a+b}{p}+k_1+k_2\right)^2}\bar{q}^{\frac{p}{4}\left(\frac{a-b}{p}+k_1-k_2\right)^2}\,.
    \label{eq:psiab-def}
\end{align}
\label{prop:enumerator=partition}
\end{proposition}

Note that the functions $\psi_{ab}^+$ are independent of $\CC$.
We sometimes use different representations of $\psi_{ab}^+$.
One useful representation is the form
\begin{align}
\label{eq:psi_theta_function}
    \psi_{ab}^+(\tau,\bar{\tau}) = \Theta_{a+b,\,p}(\tau)\,\bar{\Theta}_{a-b,\,p}(\bar{\tau}) + \Theta_{a+b-p,\,p}(\tau)\,\bar{\Theta}_{a-b-p,\,p}(\bar{\tau})\ ,
\end{align}
where $(a,b)\in\BF_p\times\BF_p$ and $\Theta_{m,\,k}(\tau)$ is the theta function
\begin{align}\label{eq:Theta-m-k-def}
    \Theta_{m,\,k}(\tau) = \sum_{n\in\BZ}\, q^{k\,\left(n+\frac{m}{2k}\right)^2}.
\end{align}

There is an important class of quantum stabilizer codes called \emph{Calderbank-Shor-Steane (CSS) codes} \cite{calderbank1996good,steane1996multiple}.
This type of quantum code can be constructed from a pair of classical Euclidean codes.
Let $C$ be an $[n,k]_p$ linear code over $\BF_p$ equipped with the standard Euclidean inner product $\cdot$.
A linear code generated by a $k\times n$ matrix $G_C$ can be characterized by a $(k-n)\times n$ matrix $H_C$ satisfying $G_C\,H_C = 0 $ mod $p$, which is called a parity check matrix.
A parity check matrix generates the dual code $C^\perp$\footnote{Here, we slightly abuse the use of the symbol $\perp$. While we use it for the dual code $\CC^\perp$ with respect to the Lorentzian inner product, we also use it for $C^\perp$ with respect to the Euclidean one.}
\begin{align}
    C^\perp = \left\{c'\in \BR^n\;\middle|\; c\cdot c'=0 \mod p\,,\;c\in C\right\}\,.
\end{align}
For a linear code $C$ with a generator matrix $G_C$ and a parity check matrix $H_C$, the dual code $C^\perp$ has a generator matrix $H_C$ and a parity check matrix $G_C$.

Suppose that $C_X$ and $C_Z$ are $[n,k_X]_p$ and $[n,k_Z]_p$ linear codes with the generator matrices $G_X\,,G_Z$ and the parity check matrices $H_X\,,H_Z$, respectively. Moreover, we assume $C_X^{\perp}\subseteq C_Z$.
Then, the check matrix
\begin{align}
\label{eq:CSS}
    \mathsf{H}_{(C_X,\,C_Z)} = 
    \left[
    \begin{array}{c|c}
        H_X\, & 0  \\
        0 &\, H_Z
    \end{array}
    \right],
\end{align}
satisfies the commutativity condition \eqref{stabilizer_condition}. Therefore, the corresponding operators $g(\Ba^{(i)},\Bb^{(i)})$ commute with each other and generate an abelian group, which can be regarded as a stabilizer group of quantum codes. 

To construct a Narain code CFT, we exploit CSS construction in the case with $C_X=C$ and $C_Z = C^\perp$.
Then, we have the $n\times 2n$ check matrix
\begin{align}
\label{eq:general_CSS_check}
    \mathsf{H}_{(C,C^\perp)} = \left[
    \begin{array}{cc}
        H_C & 0 \\
        0 & G_C
    \end{array}
    \right].
\end{align}
As shown in the following theorem, we can show that the check matrix of this form provides a Lorentzian even self-dual lattice via Construction A.
This gives a systematic construction of Narain code CFTs using classical linear codes.

\begin{theorem}[{\cite{Kawabata:2022jxt}}]
Suppose that a CSS code has a check matrix \eqref{eq:general_CSS_check} with a classical linear code $C$ and the dual code $C^\perp$.
Let $\CC$ be the classical code with the generator matrix $\mathsf{H}_{(C,C^\perp)}$.
Then, the Construction A lattice $\Lambda(\CC)$ is even self-dual with respect to the metric $\eta$.
\label{prop:general_CSS_construction}
\end{theorem}

Of course, we can take a classical linear code $C$ to be self-dual $C=C^\perp$. Then, we get the reduced version of the CSS construction.

\begin{corollary}[{\cite{Kawabata:2022jxt}}]
Suppose a CSS code with a linear self-dual code $C$. Let $\CC$ be the classical code with the generator matrix $\mathsf{H}_{(C,C)}$. Then, the Construction A lattice $\Lambda(\CC)$ is even self-dual with respect to the off-diagonal Lorentzian metric $\eta$.
\end{corollary}

For the CSS construction, the partition function can be written in terms of classical linear codes.
Let $\mathsf{H}_{(C, C)}$ be a check matrix of a CSS code constructed from a single self-dual code $C$ over $\BF_p$.
Suppose that $\CC$ is a classical code generated by the matrix $G_\mathsf{H} = \mathsf{H}_{(C,C)}$.
Then, the complete enumerator polynomial of $\CC$ reduces to 
\begin{align}
\label{eq:CSS_complete_enumerator}
    W_{C,C}^{(\mathrm{CSS})}(\{x_{ab}\})= \sum_{(c,\,c')\,\in\,C^2}\,\prod_{(a,b)\,\in\,\BF_p\times\BF_p}\,x_{ab}^{\mathrm{wt}_{ab}(c,c')}\,,
\end{align}
where, for codewords $c=(c_1,\cdots,c_n)\in C$ and $c'=(c_1',\cdots,c_n')\in C$, we define 
\begin{align}
\label{eq:2weight}
    \mathrm{wt}_{ab}(c,c') = \left|\,\{j\in\{1,\cdots,n\}\;\middle|\;c_j = a\,,\,\,c_j' = b\}\,\right|\,.
\end{align}
This representation is useful when we take the ensemble average of Narain code CFTs.

\subsection{$\BZ_2$ gauging of Narain code CFTs for $p\neq2$}
\label{ss:gauging_code_oddprime}

We consider a lattice vector of length $2n$
\begin{align}
\label{eq:choice_oddprime}
    \chi = \sqrt{p}\,(1,1,\cdots,1)\in\Lambda(\CC)\ ,
\end{align}
whose half is not in the Construction A lattice: $\delta: =\frac{\chi}{2}\notin\Lambda(\CC)$.
In terms of CFTs, this choice corresponds to the $\BZ_2$ action on the vertex operators defined by
\begin{align}\label{eq:Z2-sym-def}
    V_{p_L,\,p_R}(z,\bar{z}) \to (-1)^{\chi\odot \lambda} \,V_{p_L,\,p_R}(z,\bar{z})\,,
\end{align}
where $\lambda = (\lambda_1,\lambda_2)$ is related to the left- and right-moving momenta $(p_L, p_R)$ by \eqref{eq:orthogonal_transf}.
To use Theorem \ref{prop:deformation_theorem}, the Construction A lattice $\Lambda(\CC)$ is divided into the following two parts:
\begin{align}
    \Lambda(\CC) = \Lambda_0\cup \Lambda_1\,,
\end{align}
where
\begin{align}
    \begin{aligned}
    \label{eq:lambda0lambda1}
    \Lambda_0 &= \left\{\lambda\in\Lambda(\CC)\,|\,\chi\odot \lambda = 0 \;\;\mathrm{mod}\;2\right\},\\
    \Lambda_1 &= \left\{\lambda\in\Lambda(\CC)\,|\,\chi\odot\lambda = 1 \;\;\mathrm{mod}\;2\right\}.
    \end{aligned}
\end{align}
This provides the $\BZ_2$-grading of the untwisted Hilbert space for Narain CFTs.
The $\BZ_2$ symmetry is associated with the parity of $\chi\odot\lambda$.
Now, we aim to gauge this $\BZ_2$ symmetry and construct the twisted Hilbert space.
Below, we assume $n\in2\BZ$ to ensure that $\BZ_2$ symmetry given by $\chi$ is non-anomalous. We define the following sets for $n\in2\BZ$:
\begin{align}
    \begin{aligned}
    \label{eq:lambda2lambda3}
    \Lambda_2 &= 
    \begin{dcases}
    \Lambda_1 + \delta & \quad (n\in4\BZ)\,,\\
    \Lambda_0 + \delta & \quad (n\in 4\BZ+2)\,,
    \end{dcases}\\[0.1cm]
    \Lambda_3 &= 
    \begin{dcases}
    \Lambda_0 + \delta & \quad (n\in4\BZ)\,,\\
    \Lambda_1 + \delta & \quad (n\in 4\BZ+2)\,,
    \end{dcases}
    \end{aligned}
\end{align}
where we used \eqref{eq:lattice_shift_23}.
In the momentum basis \eqref{eq:orthogonal_transf}, we denote them by $\widetilde{\Lambda}_i$ ($i=0,1,2,3$).
We define the theta functions of each set $\widetilde{\Lambda}_i$ for $i=0,1,2,3$ as
\begin{align}
    \Theta_{\widetilde{\Lambda}_i}(\tau,\bar{\tau}) = \sum_{(p_L,\,p_R)\,\in\,\widetilde{\Lambda}_i}\,q^{\frac{p_L^2}{2}}\,\bar{q}^{\frac{p_R^2}{2}}\,,\qquad q=e^{2\pi\i\tau}\,.
\end{align}

From Theorem~\ref{prop:deformation_theorem}, we obtain a momentum lattice $\Lambda_\CO$ for the orbifold theory and a momentum lattice $\Lambda_\mathrm{NS}$ for the NS sector in the fermionized theory
\begin{align}
\label{eq:shifted_lattice}
    \begin{aligned}
    \Lambda_\CO = \Lambda_0\cup\Lambda_3\,,\qquad
    \Lambda_\mathrm{NS} = \Lambda_0\cup\Lambda_2\,.
    \end{aligned}
\end{align}
Additionally, we define the set of momenta of the R sector in the fermionized theory by
\begin{align}
    \Lambda_\mathrm{R} = \Lambda_1 \cup \Lambda_3\,.
\end{align}
Note that while we use the notation of $\Lambda_\mathrm{R}$, it is not a lattice by itself.
These are denoted by $\widetilde{\Lambda}_\CO$, $\widetilde{\Lambda}_\mathrm{NS}$, and $\widetilde{\Lambda}_\mathrm{R}$ as a set of left- and right-moving momenta \eqref{eq:orthogonal_transf}.
For these sets, we define the theta function of $\widetilde{\Lambda}_*$ ($*=\mathrm{\CO},\mathrm{NS},\mathrm{R}$) by
\begin{align}
    \Theta_{\widetilde{\Lambda}_*} = \sum_{(p_L,\,p_R)\,\in\,\widetilde{\Lambda}_*}\;q^{\frac{p_L^2}{2}}\,\bar{q}^{\frac{p_R^2}{2}}\,.
\end{align}

In what follows, we give a way of computing the theta functions of $\widetilde{\Lambda}_\CO$, $\widetilde{\Lambda}_\mathrm{NS}$, and $\widetilde{\Lambda}_\mathrm{R}$.
As these are a disjoint union of $\widetilde{\Lambda}_i$, then we have
\begin{align}
\begin{aligned}
    \Theta_{\widetilde{\Lambda}_\CO}(\tau,\bar{\tau}) &= \Theta_{\widetilde{\Lambda}_0}(\tau,\bar{\tau}) + \Theta_{\widetilde{\Lambda}_3}(\tau,\bar{\tau})\,,\\
    \Theta_{\widetilde{\Lambda}_\mathrm{NS}}(\tau,\bar{\tau}) &= \Theta_{\widetilde{\Lambda}_0}(\tau,\bar{\tau}) + \Theta_{\widetilde{\Lambda}_2}(\tau,\bar{\tau})\,,\\
    \Theta_{\widetilde{\Lambda}_\mathrm{R}}(\tau,\bar{\tau}) &= \Theta_{\widetilde{\Lambda}_1}(\tau,\bar{\tau}) + \Theta_{\widetilde{\Lambda}_3}(\tau,\bar{\tau})\,.
\end{aligned}
\end{align}
The following proposition gives the theta function of $\widetilde{\Lambda}_0$ and $\widetilde{\Lambda}_1$ in terms of the complete weight enumerator.

\begin{proposition}
The theta functions of $\widetilde{\Lambda}_i$ $(i=0,1)$ are given by
\begin{align}
    \begin{aligned}
    \Theta_{\widetilde{\Lambda}_0}(\tau,\bar{\tau}) &= \frac{1}{2} \left[ W_\CC (\{\psi_{ab}^+\}) + W_\CC(\{\psi_{ab}^-\})\right]\,,\\
    \Theta_{\widetilde{\Lambda}_1}(\tau,\bar{\tau}) &= \frac{1}{2} \left[ W_\CC (\{\psi_{ab}^+\}) - W_\CC(\{\psi_{ab}^-\})\right]\,,
    \end{aligned}
\end{align}
where
\begin{align}
    \begin{aligned}
    \psi_{ab}^+(\tau,\bar{\tau}) &= \Theta_{a+b,p}(\tau)\,\bar{\Theta}_{a-b,p}(\bar{\tau}) + \Theta_{a+b-p,p}(\tau)\,\bar{\Theta}_{a-b-p,p}(\bar{\tau})\,,\\
    \psi_{ab}^-(\tau,\bar{\tau}) &= (-1)^{a+b}\left( \Theta_{a+b,p}(\tau)\,\bar{\Theta}_{a-b,p}(\bar{\tau}) - \Theta_{a+b-p,p}(\tau)\,\bar{\Theta}_{a-b-p,p}(\bar{\tau})\right)\,.
    \end{aligned}
\end{align}

\label{prop:odd_p_theta01}
\end{proposition}

\begin{proof}
We already know lattice theta functions of the Construction A lattice, which are expressed by the complete enumerator polynomial in Proposition \ref{prop:enumerator=partition}.
Our aim is to divide it into two sets $\widetilde{\Lambda}_0$ and $\widetilde{\Lambda}_1$.
From the definition of these sets, it is natural to introduce the $\BZ_2$ grading into the Construction A lattice according to the mod $2$ value of $\chi\odot\lambda$ for $\lambda\in\Lambda(\CC)$.
Suppose $\lambda=(\lambda_1,\lambda_2)\in\Lambda(\CC)$, where
\begin{align}
    \lambda_1 = \frac{\alpha+p\,k_1}{\sqrt{p}}\,,\qquad
    \lambda_2 = \frac{\beta+p\,k_2}{\sqrt{p}}\,,\qquad
    k_1,\;k_2\in\BZ^n\,,
\end{align}
for a codeword $(\alpha,\beta)\in\CC$.
Then, the $\BZ_2$ grading is determined by
\begin{align}
    \chi\odot \lambda = \Bone_n\cdot (\alpha+\beta) + p \,\Bone_n\cdot (k_1+k_2)\,.
\end{align}
Let us introduce the lattice theta function weighted by $(-1)^{\chi\odot\lambda}$
\begin{align}
    \begin{aligned}
    \label{eq:psi-}
    \Theta'_{\widetilde{\Lambda}(\CC)}(\tau,\bar{\tau}) 
        &= \sum_{(\alpha,\beta)\,\in\,\CC}
        \,\sum_{k_1,k_2\,\in\,\BZ^n}\,(-1)^{\chi\odot\lambda}\,
        q^{\frac{p}{4}\left(\frac{\alpha+\beta}{p}+k_1+k_2\right)^2} \bar{q}^{\frac{p}{4}\left(\frac{\alpha-\beta}{p}+k_1-k_2\right)^2} \\
        &= \sum_{(\alpha,\beta)\,\in\,\CC} \,\prod_{i=1}^n \,\psi^-_{\alpha_i\beta_i}(\tau,\bar{\tau}) \\
        &= \sum_{c\,\in\,\CC}\, \prod_{(a,b)\,\in\,\BF_p\times\BF_p}\,\left(\psi^-_{ab}(\tau,\bar{\tau}) \right)^{\mathrm{wt}_{ab}(c)} \\
        &= W_\CC(\{\psi^-_{ab}\})\,,
    \end{aligned}
\end{align}
where for $(a,b)\in\BF_p\times\BF_p$ we define
\begin{align}
\begin{aligned}
    \psi^-_{ab}(\tau,\bar{\tau}) &= (-1)^{a+b} \sum_{k_1,k_2\in\BZ}\, (-1)^{p\,(k_1+k_2)}\,q^{\frac{p}{4}\left(\frac{a+b}{p}+k_1+k_2\right)^2} \bar{q}^{\frac{p}{4}\left(\frac{a-b}{p}+k_1-k_2\right)^2}\\
    &= (-1)^{a+b}\left(\Theta_{a+b,p}(\tau)\,\bar{\Theta}_{a-b,p}(\bar{\tau}) - \Theta_{a+b-p,p}(\tau)\,\bar{\Theta}_{a-b-p,p}(\bar{\tau})\right)\,.
\end{aligned}
\end{align}
Since $\Theta'_{\widetilde{\Lambda}(\CC)}(\tau,\bar{\tau})$ is weighted by $(-1)^{\chi\odot\lambda}$, the theta functios of each set $\widetilde{\Lambda}_0$ and $\widetilde{\Lambda}_1$ are given by
\begin{align}
    \begin{aligned}
    \Theta_{\widetilde{\Lambda}_0} &= \frac{1}{2}\left[\Theta_{\widetilde{\Lambda}(\CC)}(\tau,\bar{\tau}) +\Theta'_{\widetilde{\Lambda}(\CC)}(\tau,\bar{\tau})\right]\,,\\
    \Theta_{\widetilde{\Lambda}_1} &= \frac{1}{2}\left[\Theta_{\Lambda(\CC)}(\tau,\bar{\tau}) -\Theta'_{\Lambda(\CC)}(\tau,\bar{\tau})\right]\,.
    \end{aligned}
\end{align}
From Proposition \ref{prop:enumerator=partition} and \eqref{eq:psi-}, we prove the proposition.
\end{proof}

\begin{proposition}
The theta functions of $\widetilde{\Lambda}_2$ and $\widetilde{\Lambda}_3$ are given by
\begin{align}
    \begin{aligned}
    \Theta_{\widetilde{\Lambda}_2}(\tau,\bar{\tau}) &= \frac{1}{2} \left[ W_\CC (\{\tilde{\psi}^+_{ab}\}) - W_\CC(\{\tilde{\psi}^-_{ab}\})\right]\,,\\
    \Theta_{\widetilde{\Lambda}_3}(\tau,\bar{\tau}) &= \frac{1}{2} \left[ W_\CC (\{\tilde{\psi}^+_{ab}\}) + W_\CC(\{\tilde{\psi}^-_{ab}\})\right]\,,
    \end{aligned}
\end{align}
where
\begin{align}
    \begin{aligned}
    \tilde{\psi}^+_{ab}(\tau,\bar{\tau}) &= \Theta_{a+b,p}(\tau)\,\bar{\Theta}_{a-b-p,p}(\bar{\tau}) + \Theta_{a+b-p,p}(\tau)\,\bar{\Theta}_{a-b,p}(\bar{\tau})\,,\\
    \tilde{\psi}^-_{ab}(\tau,\bar{\tau}) &=  e^{\pi\i\left(\frac{2ab}{p}+a+b-\frac{p}{2}\right)}\,\left(\Theta_{a+b,p}(\tau)\,\bar{\Theta}_{a-b-p,p}(\bar{\tau}) - \Theta_{a+b-p,p}(\tau)\,\bar{\Theta}_{a-b,p}(\bar{\tau})\right)\,.
    \end{aligned}
\end{align}
\label{prop:odd_p_theta23}
\end{proposition}

\begin{proof}
The theta function associated with the set $\Lambda_2\cup\Lambda_3$ can be computed easily because it is simply the shift of the original lattice $\Lambda(\CC)$ by $\delta=\frac{\chi}{2}$. We start with the theta function of $\Lambda_2\cup\Lambda_3$ and then apply an appropriate modular transformation to obtain the theta function of each set $\Lambda_2$ and $\Lambda_3$.

Let $(\lambda_1,\lambda_2)\in\Lambda(\CC)$ be an element of the Construction A lattice $\Lambda(\CC) =\Lambda_0\cup\Lambda_1$:
\begin{align}
    \lambda_1 = \frac{\alpha+p\,k_1}{\sqrt{p}}\,,\qquad
    \lambda_2 = \frac{\beta+p\,k_2}{\sqrt{p}}\,.
\end{align}
Shifting the Construction A lattice by $\delta=\frac{\chi}{2}=\frac{\sqrt{p}}{2}(1,1,\cdots,1)$, we obtain
\begin{align}
    N(\Lambda(\CC)) := (\Lambda_0 +\delta)\cup(\Lambda_1+\delta) = \Lambda_2 \cup \Lambda_3\,.
\end{align}
An element $(\lambda_1,\lambda_2)\in N(\Lambda(\CC))$ is given by
\begin{align}
\begin{aligned}
    \lambda_1
        =\frac{\alpha}{\sqrt{p}}+\sqrt{p}\,\left(k_1 +\frac{\Bone_n}{2}\right)\,,\qquad
    \lambda_2 
        = \frac{\beta}{\sqrt{p}}+\sqrt{p}\,\left(k_2 +\frac{\Bone_n}{2}\right)\,.
\end{aligned}
\end{align}
We denote $N(\Lambda(\CC))$ by  $\widetilde{N}(\Lambda(\CC))$ in the momentum basis and $(p_L,p_R)\in \widetilde{N}(\Lambda(\CC))$ is
\begin{align}
    p_L = \frac{1}{\sqrt{2}}\left(\frac{\alpha+\beta}{\sqrt{p}}+\sqrt{p}\,(k_1 + k_2 + \Bone_n)\right)\,,\qquad
    p_R = \frac{1}{\sqrt{2}}\left(\frac{\alpha-\beta}{\sqrt{p}}+\sqrt{p}\,(k_1-k_2)\right)\,.
\end{align}
The theta function associated with the set $\widetilde{N}(\Lambda(\CC))$ is
\begin{align}
    \begin{aligned}
    \Theta_{\widetilde{N}(\Lambda(\CC))}(\tau,\bar{\tau}) &= \sum_{(p_L,p_R)\,\in\,\widetilde{N}(\Lambda(\CC))}q^{\frac{p_L^2}{2}}\,\bar{q}^{\frac{p_R^2}{2}}\\
    &= \sum_{(\alpha,\beta)\,\in\,\CC}\,\prod_{i=1}^n\, \tilde{\psi}^+_{\alpha_i\beta_i}(\tau,\bar{\tau}) \\
    &= \sum_{c\,\in\,\CC}\,\prod_{(a,b)\,\in\,\BF_p\times\BF_p}\,\left(\tilde{\psi}^+_{ab}(\tau,\bar{\tau})\right)^{\mathrm{wt}_{ab}(c)} \\
    &= W_\CC(\{\tilde{\psi}^+_{ab}\})\,,
    \end{aligned}
\end{align}
where we define, for $(a,b)\in\BF_p\times\BF_p$,
\begin{align}
\begin{aligned}
 \tilde{\psi}^+_{ab}(\tau,\bar{\tau}) &= \sum_{k_1,k_2\,\in\,\BZ} \, q^{\frac{p}{4}\left(\frac{a+b}{p}+k_1+k_2+1\right)^2}\,\bar{q}^{\frac{p}{4}\left(\frac{a-b}{p}+k_1-k_2\right)^2}\,,\\
    &=  \Theta_{a+b,p}(\tau)\,\bar{\Theta}_{a-b-p,p}(\bar{\tau}) + \Theta_{a+b-p,p}(\tau)\,\bar{\Theta}_{a-b,p}(\bar{\tau})\,.
\end{aligned}
\end{align}
Therefore, the theta function of $\widetilde{N}(\Lambda(\CC))$ can be computed by the change of variables: $\psi_{ab}^+\mapsto \tilde{\psi}^+_{ab}$.
Now we aim to divide the theta function of the set $\widetilde{N}(\Lambda(\CC)) = \widetilde{\Lambda}_2\cup \widetilde{\Lambda}_3$ into two theta functions associated with $\widetilde{\Lambda}_2$ and $\widetilde{\Lambda}_3$, respectively.
We can exploit the property that the set $\widetilde{\Lambda}_2$ has only odd norm elements and the other $\widetilde{\Lambda}_3$ has only even norm ones, which appears in Proposition \ref{prop:norm_general}.

By the modular $T$ transformation $\tau\to \tau+1$, equivalently $q\to e^{2\pi\i}q$, we obtain the theta function:
\begin{align}
    \begin{aligned}
    \Theta_{\widetilde{N}(\Lambda(\CC))}(\tau+1,\bar{\tau}+1) = \sum_{(p
    _L,p_R)\,\in\,\widetilde{N}(\Lambda(\CC))} (-1)^{p_L^2-p_R^2}\,\, q^{\frac{p_L^2}{2}}\,\bar{q}^{\frac{p_R^2}{2}}\,,
    \end{aligned}
\end{align}
the theta function is weighted by the $\BZ_2$ grading according to the mod 2 value of $p_L^2-p_R^2$.
Since $p_L^2-p_R^2$ are odd for $\widetilde{\Lambda}_2$ and even for $\widetilde{\Lambda}_3$, we obtain the theta functions of each set
\begin{align}
    \begin{aligned}
    \Theta_{\widetilde{\Lambda}_2}(\tau,\bar{\tau}) &= \frac{1}{2}\left[\Theta_{\widetilde{N}(\Lambda(\CC))}(\tau,\bar{\tau}) -\Theta_{\widetilde{N}(\Lambda(\CC))}(\tau+1,\bar{\tau}+1)\right]\,,\\
    \Theta_{\widetilde{\Lambda}_3} (\tau,\bar{\tau}) &= \frac{1}{2}\left[\Theta_{\widetilde{N}(\Lambda(\CC))}(\tau,\bar{\tau}) +\Theta_{\widetilde{N}(\Lambda(\CC))}(\tau+1,\bar{\tau}+1)\right]\,.
    \end{aligned}
\end{align}
The theta function associated with $\widetilde{N}(\Lambda(\CC))$ depends on the moduli parameter $\tau$ only through $\tilde{\psi}^+_{ab}(\tau,\bar{\tau})$.
Under the modular $T$ transformation, it behaves as
\begin{align}
\begin{aligned}
    \tilde{\psi}^-_{ab}(\tau,\bar{\tau}) &:= \tilde{\psi}^+_{ab}(\tau+1,\bar{\tau}+1) \\
    &= e^{\pi\i\left(\frac{2ab}{p}+a+b-\frac{p}{2}\right)}\,\left(\Theta_{a+b,p}(\tau)\,\bar{\Theta}_{a-b-p,p}(\bar{\tau}) - \Theta_{a+b-p,p}(\tau)\,\bar{\Theta}_{a-b,p}(\bar{\tau})\right)\,.
\end{aligned}
\end{align}
Then, substituting the above term into the complete enumerator polynomial instead of the modular $T$ transformation, we obtain the proposition.
\end{proof}

\subsection{$\BZ_2$ gauging of Narain code CFTs for $p=2$}
\label{ss:gauging_code_p=2}

We implement the $\BZ_2$ gauging by a lattice shift of Narain code CFTs.
We begin to define the notion of $\BF_4$-evenness for classical codes $\CC$, which will be necessary to deform lattices with the $\BZ_2$ symmetry.

Qubit stabilizer codes $(p=2)$ can be represented by classical codes over $\BF_4$ \cite{Calderbank:1996aj}.
Suppose that a stabilizer code has the following check matrix:
\begin{align}
    H = \left[
    \begin{array}{c|c}
    \,\alpha_1 \,&\, \beta_1\, \\
    \,\vdots \,&\, \vdots\, \\
    \,\alpha_{n} \,&\, \beta_{n}\,
    \end{array}\right],
\end{align}
where $\alpha_i,\beta_i\in\BZ^n_2$. 
We regard it as a generator matrix of a classical code $\CC\subset\BF_2^{2n}$. 
We map a codeword $c=(\alpha,\beta)\in\BF_2^n\times\BF_2^n$ of $\CC$ to a codeword $c\in\BF_4^n$ of the associated classical code over $\BF_4$ by the Gray map
\begin{align}
    \begin{aligned}
    0&\leftrightarrow (0,0)\,,&\qquad 1 &\leftrightarrow (1,1)\,,\\
    \omega&\leftrightarrow (1,0)\,,&\qquad \bar{\omega}&\leftrightarrow (0,1)\,.
    \end{aligned}
\end{align}
This map is an isomorphism under addition between $\BF_4$ and $\BF_2\times\BF_2$. 
Through the Gray map, the $i$-th component of $c\in\BF_4^n$ gives the $i$-th component of $\alpha$ and $\beta$ for $(\alpha,\beta)\in \BF_2^n\times\BF_2^n$.
For example,
\begin{align}
\begin{aligned}
\label{eq:gray_ex}
    \BF_4^2&\ni(1,\omega) &\quad &\longleftrightarrow \quad &(1,1\,|\,1,0)&\in\BF_2^2\times\BF_2^2\,,\\
    \BF_4^2&\ni(0,1)&\quad & \longleftrightarrow &\quad (0,1\,|\,0,1)&\in\BF_2^2\times\BF_2^2\,.
\end{aligned}
\end{align}

The Hamming distance of a codeword $c\in\BF_4$ is the number of non-zero elements $1,\omega,\bar{\omega}\in\BF_4$ of $c\in\BF_4^n$.
We can formulate it in the language of  $(\alpha,\beta)\in\BF_2^n\times\BF_2^n$.
Let $c\in \BF_4^n$ be a codeword, which is mapped to $(\alpha,\beta)\in\BF_2^n\times\BF_2^n$ through the Gray map.
Then, the Hamming distance of a codeword $c\in\BF_4$ can be counted by
\begin{align}
\label{eq:gf4_hamming}
    \mathrm{wt}(c) = \alpha\cdot\alpha + \beta\cdot\beta - \alpha\cdot\beta\ .
\end{align}
Let us check it using the examples \eqref{eq:gray_ex}. The first example $(1,\omega)\in\BF_4^2$ is mapped to $\alpha=(1,1)$ and $\beta=(1,0)$, and \eqref{eq:gf4_hamming} returns $\mathrm{wt}(c)=\alpha\cdot\alpha + \beta\cdot\beta - \alpha\cdot\beta=2+1-1=2$, which is correct since $(1,\omega)$ has no non-zero elements.
On the other hand, the Hamming distance of $(0,1)\in\BF_4^2$ is $1$. This can be reproduced by \eqref{eq:gf4_hamming} because $\mathrm{wt}(c)=\alpha\cdot\alpha + \beta\cdot\beta - \alpha\cdot\beta= 1 + 1 -1=1$.

In analogy with classical binary codes, we call a code over $\BF_4$ even if it has only codewords with even Hamming distance.
We extend this notion for $\BF_4^n$ to $\BF_2^n\times\BF_2^n$ using the Gray map.
We call $\CC\subset\BF_2^n\times\BF_2^n$ $\BF_4$-even if $\mathrm{wt}(c) = \alpha\cdot\alpha + \beta\cdot\beta - \alpha\cdot\beta \in 2\BZ$ for all codewords $c=(\alpha,\beta)\in\CC$.

\begin{proposition}
Let $\CC\subset\BF_2^n\times\BF_2^n$ be a doubly-even self-dual code with respect to the off-diagonal Lorentzian metric $\eta$.
If $\CC$ is $\BF_4$-even, then $\Bone_{2n}\in\CC$. Otherwise, $\Bone_{2n}\notin\CC$.
\label{prop:GF4_even}
\end{proposition}

\begin{proof}
Let $\CC$ be $\BF_4$-even. Then, a codeword $c = (\alpha,\beta)\in\CC$ satisfies
\begin{align}
\label{eq:GF4_even}
    \alpha\cdot\alpha + \beta\cdot\beta - \alpha\cdot\beta\in2\BZ\,.
\end{align}
Since $\CC$ is doubly-even with respect to $\eta$, we have
\begin{align}
    c \odot c = \alpha\cdot\beta + \beta\cdot\alpha = 2\alpha\cdot\beta \in4\BZ\ ,
\end{align}
which gives $\alpha\cdot\beta\in2\BZ$.
Combining it with \eqref{eq:GF4_even}, we obtain $\alpha\cdot\alpha + \beta\cdot\beta\in 2\BZ$.
Using the relation $\alpha\cdot\alpha=\Bone_n\cdot\alpha$ mod 2, we can write $\alpha\cdot\alpha + \beta\cdot\beta\in 2\BZ$ as
\begin{align}
\label{eq:GF(4)-even_doubly-even}
    \Bone_{2n}\odot c = 0 \qquad\mathrm{mod}\;\,2\ ,
\end{align}
where $c=(\alpha,\beta)$ and $\Bone_{2n}$ is the $2n$-dimensional row vector $(1,1,\cdots,1)$. The classical code $\CC$ is self-dual, so the above relation shows $\Bone_{2n}\in\CC^\perp = \CC$. If a classical code $\CC$ is not $\BF_4$-even, there exists a codeword $c=(\alpha,\beta)$ that does not satisfy \eqref{eq:GF4_even}. Then there exists a codeword $c\in\CC$ that does not satisfy $\Bone_{2n}\odot c = 0$ mod $2$, which concludes $\Bone_{2n}\notin \CC^\perp = \CC$.
\end{proof}

\subsubsection{$\BF_4$-even code}
\label{ss:GF(4)even}
Let $\CC$ be $\BF_4$-even, equivalently $\Bone_{2n}\in\CC$.
We take an element of the Construction A lattice $\Lambda(\CC)$
\begin{align}
    \chi = \frac{1}{\sqrt{2}}\,(1,1,\cdots,1)\in \Lambda(\CC).
\end{align}
Since a $\BF_4$-even code $\CC$ contains the all-one vector $\Bone_{2n}$, the Construction A lattice includes $\chi\in\Lambda(\CC)$.
Note that for a non-$\BF_4$-even code $\CC$, the vector $\Bone_{2n}$ is not contained in a code $\CC$ by Proposition \ref{prop:GF4_even}. Then, the following construction does not hold.
We separately discuss non-$\BF_4$-even codes in the next section.

Since the Construction A lattice $\Lambda(\CC)$ is integral, we can divide it into
\begin{align}
    \Lambda(\CC) = \Lambda_0 \cup \Lambda_1\ ,
\end{align}
where
\begin{align}
    \begin{aligned}
    \label{eq:lambda0lambda1_p=2}
    \Lambda_0 &= \{\lambda\in\Lambda(\CC)\,|\,\chi\odot \lambda = 0 \;\;\mathrm{mod}\;2\}\ ,\\
    \Lambda_1 &= \{\lambda\in\Lambda(\CC)\,|\,\chi\odot\lambda = 1 \;\;\mathrm{mod}\;2\}\ .
    \end{aligned}
\end{align}
We define a shift vector of length $2n$ by
\begin{align}
    \delta = \frac{\chi}{2} = \frac{1}{2\sqrt{2}}\,(1,1,\cdots,1)\ .
\end{align}
The non-anomalous condition for the $\BZ_2$ symmetry reduces to $n\in4\BZ$.
We define the following sets according to \eqref{eq:lattice_shift_23}:
\begin{align}
\label{eq:lambda2lambda3_p=2}
    \begin{aligned}
    \Lambda_2 &= 
    \begin{dcases}
    \Lambda_1 + \delta & (n\in8\BZ)\ ,\\
    \Lambda_0 + \delta & (n\in 8\BZ+4)\ ,
    \end{dcases}\\[0.1cm]
    \Lambda_3 &= 
    \begin{dcases}
    \Lambda_0 + \delta & (n\in8\BZ)\ ,\\
    \Lambda_1 + \delta & (n\in 8\BZ+4)\ .
    \end{dcases}
    \end{aligned}
\end{align}
As in the case of an odd prime $p\neq2$, we denote them by $\widetilde{\Lambda}_i$ in the momentum basis.
By combining each set, we obtain the momentum lattice in the orbifold theory, the set of momenta of the NS and R sector in the fermionized theory
\begin{align}
    \widetilde{\Lambda}_\CO = \widetilde{\Lambda}_0\cup\widetilde{\Lambda}_3\,,\qquad
    \widetilde{\Lambda}_\mathrm{NS} = \widetilde{\Lambda}_0\cup\widetilde{\Lambda}_2\,,\qquad \widetilde{\Lambda}_\mathrm{R} = \widetilde{\Lambda}_1\cup\widetilde{\Lambda}_3\,.
\end{align}

\begin{proposition}
The theta functions of $\widetilde{\Lambda}_0$ and $\widetilde{\Lambda}_1$ are
\begin{align}
    \begin{aligned}
    \Theta_{\widetilde{\Lambda}_0}(\tau,\bar{\tau}) &= \frac{1}{2} \left( W_\CC (\{\psi^+_{ab}\}) + W_\CC(\{\psi^-_{ab}\})\right)\,,\\
    \Theta_{\widetilde{\Lambda}_1}(\tau,\bar{\tau}) &= \frac{1}{2} \left( W_\CC (\{\psi^+_{ab}\}) - W_\CC(\{\psi^-_{ab}\})\right)\,,
    \end{aligned}
\end{align}
where
\begin{align}\label{psi_pm}
    \begin{aligned}
    \psi^+_{ab}(\tau,\bar{\tau}) &= \Theta_{a+b,2}(\tau)\,\overline{\Theta_{a-b,2}({\tau})} + \Theta_{a+b-2,2}(\tau)\,\overline{\Theta_{a-b-2,2}({\tau})}\,,\\
    \psi^-_{ab}(\tau,\bar{\tau}) &= (-1)^{\frac{a+b}{2}}\left( \Theta_{a+b,2}(\tau)\,\overline{\Theta_{a-b,2}({\tau})} - \Theta_{a+b-2,2}(\tau)\,\overline{\Theta_{a-b-2,2}({\tau})}\right)\,.
    \end{aligned}
\end{align}
\end{proposition}

\begin{proof}
From Proposition \ref{prop:enumerator=partition}, the lattice theta function of the Construction A lattice $\Lambda(\CC)$ can be computed from the complete enumerator polynomial.
To divide $\Lambda(\CC)$ into two subsets $\widetilde{\Lambda}_0$ and $\widetilde{\Lambda}_1$, it is natural to introduce the $\BZ_2$ grading of $\Lambda(\CC)$.
Let $\lambda=(\lambda_1,\lambda_2)$ be an element of $\Lambda(\CC)$, which can be written as
\begin{align}
\label{eq:lattice_vector}
    \lambda_1 = \frac{\alpha+2k_1}{\sqrt{2}}\,,\qquad
    \lambda_2 = \frac{\beta+2k_2}{\sqrt{2}}\,,\qquad
    k_1\,,\,k_2\in\BZ^n\,,
\end{align}
for a codeword $(\alpha,\beta)\in\CC$. The $\BZ_2$ grading given by $\chi\odot\lambda$ reduces to the one by
\begin{align}
    \chi\odot \lambda = \frac{\Bone_n}{2}\cdot (\alpha+\beta) +  \,\Bone_n\cdot (k_1+k_2)\,,
\end{align}
where we have $\Bone_n\cdot (\alpha+\beta)\in2\BZ$ from \eqref{eq:GF(4)-even_doubly-even}. This ensures $\chi\odot\lambda\in\BZ$.
The lattice theta function weighted by $\chi\odot\lambda$ is
\begin{align}
    \begin{aligned}
    \Theta'_{\widetilde{\Lambda}(\CC)}(\tau,\bar{\tau}) &= \sum_{(\alpha,\beta)\,\in\,\CC}
    \,\sum_{k_1,k_2\,\in\,\BZ^n}\,(-1)^{\chi\odot\lambda}\,
    q^{\frac{1}{2}\left(\frac{\alpha+\beta}{2}+k_1+k_2\right)^2} \bar{q}^{\frac{1}{2}\left(\frac{\alpha-\beta}{2}+k_1-k_2\right)^2} \\
    &= \sum_{(\alpha,\beta)\,\in\,\CC} \,\prod_{i=1}^n \,\psi^-_{\alpha_i\beta_i}(\tau,\bar{\tau})\\ 
    &= \sum_{c\,\in\,\CC}\, \prod_{(a,b)\,\in\,\BF_2\times\BF_2}\,\left(\psi^-_{ab}(\tau,\bar{\tau}) \right)^{\mathrm{wt}_{ab}(c)} \\
    &= W_\CC(\{\psi^-_{ab}\})\,,
    \end{aligned}
\end{align}
where, for $(a,b)\in\BF_2\times\BF_2$, we define
\begin{align}
\begin{aligned}
    \psi^-_{ab}(\tau,\bar{\tau}) &= (-1)^{\frac{a+b}{2}} \sum_{k_1,k_2\in\BZ}\, (-1)^{k_1+k_2}\,q^{\frac{1}{2}\left(\frac{a+b}{2}+k_1+k_2\right)^2} \bar{q}^{\frac{1}{2}\left(\frac{a-b}{2}+k_1-k_2\right)^2} \\
    &= (-1)^{\frac{a+b}{2}}\left(\Theta_{a+b,2}(\tau)\,\overline{\Theta_{a-b,2}({\tau})} - \Theta_{a+b-2,2}(\tau)\,\overline{\Theta_{a-b-2,2}({\tau})}\right)\,.
\end{aligned}
\end{align}
Since the weighted theta function $\Theta'_{\widetilde{\Lambda}(\CC)}(\tau,\bar{\tau})$ has positive coefficients for $\Lambda_0$ and negative ones for $\Lambda_1$, the theta functions of each set are given by
\begin{align}
    \begin{aligned}
    \Theta_{\widetilde{\Lambda}_0}(\tau,\bar{\tau}) 
        &= \frac{1}{2}\left[\Theta_{\widetilde{\Lambda}(\CC)}(\tau,\bar{\tau}) +\Theta'_{\widetilde{\Lambda}(\CC)}(\tau,\bar{\tau})\right] \\
        &= \frac{1}{2}\left[ W_\CC(\{\psi^+_{ab}\})+W_\CC(\{\psi^-_{ab}\})\right] \,,\\
    \Theta_{\widetilde{\Lambda}_1}(\tau,\bar{\tau}) 
        &= \frac{1}{2}\left[ \Theta_{\widetilde{\Lambda}(\CC)}(\tau,\bar{\tau}) -\Theta'_{\widetilde{\Lambda}(\CC)}(\tau,\bar{\tau})\right] \\
        &= \frac{1}{2}\left[W_\CC(\{\psi^+_{ab}\})-W_\CC(\{\psi^-_{ab}\})\right]\,.
    \end{aligned}
\end{align}
\end{proof}

\begin{proposition}
The theta functions of $\widetilde{\Lambda}_2$ and $\widetilde{\Lambda}_3$ are
\begin{align}
    \begin{aligned}
    \Theta_{\widetilde{\Lambda}_2}(\tau,\bar{\tau}) &= \frac{1}{2} \left[ W_\CC (\{\tilde{\psi}^+_{ab}\}) - W_\CC(\{\tilde{\psi}^-_{ab}\})\right]\,,\\
    \Theta_{\widetilde{\Lambda}_3}(\tau,\bar{\tau}) &= \frac{1}{2} \left[ W_\CC (\{\tilde{\psi}^+_{ab}\}) + W_\CC(\{\tilde{\psi}^-_{ab}\})\right]\,,
    \end{aligned}
\end{align}
where
\begin{align}\label{psi_tilde_pm}
    \begin{aligned}
    \tilde{\psi}^+_{ab}(\tau,\bar{\tau}) &= \Theta_{a+b+1,2}(\tau)\,\overline{\Theta_{a-b,2}({\tau})} + \Theta_{a+b-1,2}(\tau)\,\overline{\Theta_{a-b-2,2}({\tau})}\,,\\
    \tilde{\psi}^-_{ab}(\tau,\bar{\tau}) &= e^{\pi\i\,\left(a+\frac{1}{2}\right) \left(b+\frac{1}{2}\right)}\left(\Theta_{a+b+1,2}(\tau)\,\overline{\Theta_{a-b,2}({\tau})} - \Theta_{a+b-1,2}(\tau)\,\overline{\Theta_{a-b-2,2}({\tau})}\right)\,.
    \end{aligned}
\end{align}
\end{proposition}

\begin{proof}
Let $(\lambda_1,\lambda_2)$ be an element of $\Lambda(\CC)$ given by \eqref{eq:lattice_vector}. 
The shift of the Construction A lattice by $\chi=\frac{\chi}{2}$ endows $N(\Lambda(\CC)): =\Lambda_2\cup\Lambda_3$.
Therefore, an element $(\lambda_1,\lambda_2)$ of $\Lambda_2\cup\Lambda_3$ can be written as
\begin{align}
\label{eq:element_23}
    \lambda_1 = \frac{\alpha}{\sqrt{2}} + \sqrt{2}\,\left(k_1 + \frac{\Bone_n}{4}\right)\,,\qquad
    \lambda_2 = \frac{\beta}{\sqrt{2}} + \sqrt{2}\,\left(k_2 + \frac{\Bone_n}{4}\right)\,,
\end{align}
where $(\alpha,\beta)\in\CC$ and $k_1$, $k_2\in\BZ^n$.

To compute the corresponding theta function, we move onto the momentum basis.
We denote $N(\Lambda(\CC))$ by $\widetilde{N}(\Lambda(\CC))$ in the momentum basis. Then, an element $(p_L,p_R) = (\frac{\lambda_1+\lambda_2}{\sqrt{2}},\frac{\lambda_1-\lambda_2}{\sqrt{2}})\in \widetilde{N}(\Lambda(\CC))$ is
\begin{align}
    p_L = \frac{\alpha+\beta+2(k_1+k_2)}{2}+\frac{\Bone_n}{2}\,,\qquad
    p_R = \frac{\alpha-\beta+2(k_1-k_2)}{2}\,,
\end{align}
where $(\lambda_1,\lambda_2)\in N(\Lambda(\CC))$ given in \eqref{eq:element_23}.
Then, the theta function associated with $\widetilde{N}(\Lambda(\CC))$ is
\begin{align}
\begin{aligned}
    \Theta_{\widetilde{N}(\Lambda(\CC))}(\tau,\bar{\tau}) 
        &= \sum_{(\alpha,\beta)\,\in\,\CC} \sum_{k_1,k_2\,\in\,\BZ^n}\,q^{\frac{1}{2}\left(\frac{\alpha+\beta}{2}+k_1+k_2 +\frac{\Bone_n}{2}\right)^2}\,\bar{q}^{\frac{1}{2}\left(\frac{\alpha-\beta}{2}+k_1-k_2 \right)^2}\,\\
        &= \sum_{(\alpha,\beta)\,\in\,\CC}\, \prod_{i=1}^n\,\tilde{\psi}^+_{\alpha_i\beta_i}(\tau,\bar{\tau}) \\
        &= \sum_{c\,\in\,\CC}\,\prod_{(a,b)\,\in\,\BF_2\times\BF_2}\,\left(\tilde{\psi}^+_{ab}(\tau,\bar{\tau})\right)^{\mathrm{wt}_{ab}(c)} \\
        &= W_\CC(\{\tilde{\psi}^+_{ab}\})\,,
\end{aligned}
\end{align}
where, for $(a,b)\in\BF_2\times\BF_2$, we define
\begin{align}
    \begin{aligned}
    \tilde{\psi}^+_{ab}(\tau,\bar{\tau})  
        &= \sum_{k_1,k_2\,\in\,\BZ} \, q^{\frac{1}{2}\left(\frac{a+b}{2}+k_1+k_2+\frac{1}{2}\right)^2}\,\bar{q}^{\frac{1}{2}\left(\frac{a-b}{2}+k_1-k_2\right)^2} \\
        &=  \Theta_{a+b+1,2}(\tau)\,\overline{\Theta_{a-b,2}({\tau})} + \Theta_{a+b-1,2}(\tau)\,\overline{\Theta_{a-b-2,2}({\tau})}\,.
    \end{aligned}
\end{align}
Therefore, the shift of the Construction A lattice $\Lambda(\CC)$ by $\delta=\frac{\chi}{2}$ changes the theta function by the replacement of variables: $\psi^+_{ab}\mapsto \tilde{\psi}^+_{ab}$.

Let us divide $\widetilde{N}(\Lambda(\CC))$ into two subsets $\widetilde{\Lambda}_2$ and $\widetilde{\Lambda}_3$.
Note that the set $\widetilde{\Lambda}_2$ has only odd norm, while $\widetilde{\Lambda}_3$ has only even norm, which can be read off from Proposition \ref{prop:norm_general}.
Since the modular $T$ transformation $\tau\to \tau+1$ acts as
\begin{align}
    \begin{aligned}
    \Theta_{\widetilde{N}(\Lambda(\CC))}(\tau+1,\bar{\tau}+1) = \sum_{(p
    _L,p_R)\,\in\,\widetilde{N}(\Lambda(\CC))} (-1)^{p_L^2-p_R^2}\,\, q^{\frac{p_L^2}{2}}\,\bar{q}^{\frac{p_R^2}{2}}\,,
    \end{aligned}
\end{align}
we obtain the theta functions of each set $\widetilde{\Lambda}_2$ and $\widetilde{\Lambda}_3$
\begin{align}
    \begin{aligned}
    \Theta_{\widetilde{\Lambda}_2}(\tau,\bar{\tau}) 
        &= \frac{1}{2}\left[\Theta_{\widetilde{N}(\Lambda(\CC))}(\tau,\bar{\tau}) -\Theta_{\widetilde{N}(\Lambda(\CC))}(\tau+1,\bar{\tau}+1)\right]\,,\\
    \Theta_{\widetilde{\Lambda}_3}(\tau,\bar{\tau}) 
        &= \frac{1}{2}\left[\Theta_{\widetilde{N}(\Lambda(\CC))}(\tau,\bar{\tau}) +\Theta_{\widetilde{N}(\Lambda(\CC))}(\tau+1,\bar{\tau}+1)\right]\,.
    \end{aligned}
\end{align}
The theta function of $\widetilde{N}(\Lambda(\CC))$ depends on the torus modulus $\tau$ only through $\tilde{\psi}^+_{ab}(\tau,\bar{\tau})$.
Under the modular $T$ transformation $\tau\to\tau+1$, it behaves as
\begin{align}
    \tilde{\psi}^+_{ab}(\tau+1,\bar{\tau}+1) = e^{\pi\i\,\left(a+\frac{1}{2}\right)\left(b+\frac{1}{2}\right)}\,\left(\Theta_{a+b+1,2}(\tau)\,\bar{\Theta}_{a-b,2}(\bar{\tau}) - \Theta_{a+b-1,2}(\tau)\,\bar{\Theta}_{a-b-2,2}(\bar{\tau})\right)\,.
\end{align}
Then, we substitute the above term into the complete enumerator polynomial instead of the modular $T$ transformation. 
\end{proof}

Note that, in the binary case $(p=2)$, the theta functions are
\begin{align}
    \begin{aligned}
    \label{eq:theta_p=2}
    \Theta_{0,2} = \frac{\vartheta_3+\vartheta_4}{2}\,,\qquad 
    \Theta_{1,2} = \frac{\vartheta_2}{2}\,,\qquad
    \Theta_{2,2} = \frac{\vartheta_3-\vartheta_4}{2}\,,\qquad
    \Theta_{3,2} = \frac{\vartheta_2}{2}\,,
    \end{aligned}
\end{align}
where $\vartheta_i(\tau)$ ($i=2,3,4$) are the Jacobi theta functions.\footnote{We use the convention of \cite{Polchinski:1998rq}.}
Then, we have the following:
\begin{align}
    \begin{aligned}
        \psi^+_{00} &= \frac{\vartheta_3\,\bar{\vartheta}_3+\vartheta_4\,\bar{\vartheta}_4}{2}\,,&\qquad
        \psi^+_{01} &=\psi^+_{10} = \frac{\vartheta_2\,\bar{\vartheta}_2}{2}\,,&\qquad
        \psi^+_{11} &= \frac{\vartheta_3\,\bar{\vartheta}_3-\vartheta_4\,\bar{\vartheta}_4}{2}\,,\\
    
        \psi^-_{00} &= \frac{\vartheta_3\,\bar{\vartheta}_4+\vartheta_3\,\bar{\vartheta}_4}{2}\,,&\qquad
        \psi^-_{01} &= 
        \psi^-_{10} = 0\,,&\quad 
        \psi^-_{11} &= \frac{\vartheta_4\,\bar{\vartheta}_3-\vartheta_3\,\bar{\vartheta}_4}{2}\,, \\
    
        \tilde{\psi}^+_{00} &= \frac{\vartheta_2\,\bar{\vartheta}_3}{2}\,,&\qquad
        \tilde{\psi}^+_{01}&= \tilde{\psi}^+_{10} = \frac{\vartheta_3\,\bar{\vartheta}_2}{2}\,,&\quad 
        \tilde{\psi}^+_{11} &= \frac{\vartheta_2\,\bar{\vartheta}_3}{2}\,,\\
    
        \tilde{\psi}^-_{00} &= e^{\frac{\pi\i}{4}}\, \frac{\vartheta_2\,\bar{\vartheta}_4}{2}\,,&\qquad
        \tilde{\psi}^-_{01} &= \tilde{\psi}^-_{10} =\,\, e^{-\frac{\pi\i}{4}}\, \frac{\vartheta_4\,\bar{\vartheta}_2}{2}\,,&\quad 
        \tilde{\psi}^-_{11} &= e^{\frac{\pi\i}{4}}\, \frac{\vartheta_2\,\bar{\vartheta}_4}{2}\,. 
    \end{aligned}
\end{align}
Substituting the above theta functions into the lattice theta function of $\Lambda_\CO$, it reproduces the result in Eq.(3.45) in \cite{Dymarsky:2020qom} for the case with $n\in 8\BZ+4$ which they focus on.

\subsubsection{Non-$\BF_4$-even code}
Let a classical code $\CC$ be non-$\BF_4$-even.
Then, the classical code does not contain the all-one vector: $\Bone_{2n}\notin\CC$.
In this case, we choose the following lattice vector of length $2n$:
\begin{align}
\label{eq:choice_nongf4}
    \chi = \sqrt{2}\,(1,1,\cdots,1)\in\Lambda(\CC)\,,
\end{align}
where its half $\delta=\frac{\chi}{2}$ is not an element of the lattice $\Lambda(\CC)$ due to the absence of $\Bone_{2n}$ in the code $\CC$. The vector $\chi$ satisfies the non-anomalous condition: $\delta\odot\delta=n\in\BZ$.
The Construction A lattice $\Lambda(\CC)$ is divided into two parts according to \eqref{eq:lattice_divide} as
\begin{align}
    \Lambda=\Lambda_0\cup\Lambda_1\,,
\end{align}
where
\begin{align}
\begin{aligned}
    \Lambda_0 
        &=
        \left\{\lambda\in\Lambda(\CC)\,|\,\chi\odot\lambda =0 \;\;\mathrm{mod}\;\,2\right\}\,,\\
    \Lambda_1 
        &=
        \left\{\lambda\in\Lambda(\CC)\,|\,\chi\odot\lambda =1 \;\;\mathrm{mod}\;\,2\right\}\,.
\label{eq:lambda0lambda1_nonGF(4)}
\end{aligned}
\end{align}
We define the sets shifted by the vector $\delta$
\begin{align}
    \begin{aligned}
    \label{eq:lambda2lambda3_nonGF(4)}
    \Lambda_2 &=
    \begin{dcases}
    \Lambda_1 + \delta & \quad (n\in2\BZ)\,,\\
    \Lambda_0 + \delta & \quad (n\in2\BZ+1)\,,
    \end{dcases}\\
    \Lambda_3 &=
    \begin{dcases}
    \Lambda_0 + \delta & \quad (n\in2\BZ)\,,\\
    \Lambda_1 + \delta & \quad (n\in2\BZ+1)\,.
    \end{dcases}
    \end{aligned}
\end{align}
In the momentum basis, we denote them by $\widetilde{\Lambda}_i$.

\begin{proposition}
The theta functions of $\widetilde{\Lambda}_0$ and $\widetilde{\Lambda}_1$ are given by
\begin{align}
    \begin{aligned}
    \Theta_{\widetilde{\Lambda}_0}(\tau,\bar{\tau}) &= \frac{1}{2} \left[ W_\CC (\{\psi^+_{ab}\}) + W_\CC(\{\psi^-_{ab}\})\right]\,,\\
    \Theta_{\widetilde{\Lambda}_1}(\tau,\bar{\tau}) &= \frac{1}{2} \left[ W_\CC (\{\psi^+_{ab}\}) - W_\CC(\{\psi^-_{ab}\})\right]\,,
    \end{aligned}
\end{align}
where 
\begin{align}
    \begin{aligned}
    \psi^+_{ab}(\tau,\bar{\tau}) &= \Theta_{a+b,2}(\tau)\,\bar{\Theta}_{a-b,2}(\bar{\tau}) + \Theta_{a+b-2,2}(\tau)\,\bar{\Theta}_{a-b-2,2}(\bar{\tau})\,,\\
    \psi^-_{ab}(\tau,\bar{\tau}) &= (-1)^{a+b}\,\psi^+_{ab}(\tau,\bar{\tau})\,.
    \end{aligned}
\end{align}
\end{proposition}

\begin{proof}
We aim to divide the theta function of the Construction A lattice into two sets: $\widetilde{\Lambda}_0$ and $\widetilde{\Lambda}_1$.
From the definition of these sets, we introduce $\BZ_2$ grading into the Construction A lattice by whether $\chi\odot\lambda$ is even or odd for $\lambda\in\Lambda(\CC)$.
Suppose that $\lambda=(\lambda_1,\lambda_2)\in\Lambda(\CC)$ is given by
\begin{align}
    \lambda_1 = \frac{\alpha+2\,k_1}{\sqrt{2}}\,,\qquad
    \lambda_2 = \frac{\beta+2\,k_2}{\sqrt{2}}\,,\qquad
    k_1,k_2\in\BZ^n\,,
\end{align}
for a codeword $(\alpha,\beta)\in\CC$. Then the $\BZ_2$ classification is given in terms of the following:
\begin{align}
    \chi\odot \lambda = \Bone_n\cdot (\alpha+\beta) + 2 \,\Bone_n\cdot (k_1+k_2)\,.
\end{align}
Then the inner product $\chi\odot\lambda$ is even if $\Bone_n\cdot (\alpha+\beta)\in2\BZ$ and odd if $\Bone_n\cdot(\alpha+\beta)\in2\BZ+1$.
We define the weighted theta function according to whether $\chi\odot\lambda$ is even or odd:
\begin{align}
    \begin{aligned}
    \Theta'_{\widetilde{\Lambda}(\CC)}(\tau,\bar{\tau}) 
        &= \sum_{(\alpha,\beta)\,\in\,\CC}
        \,\sum_{k_1,k_2\,\in\,\BZ^n}\,(-1)^{\chi\odot\lambda}\,
        q^{\frac{1}{2}\left(\frac{\alpha+\beta}{2}+k_1+k_2\right)^2} \bar{q}^{\frac{1}{2}\left(\frac{\alpha-\beta}{2}+k_1-k_2\right)^2}\\
        &= \sum_{(\alpha,\beta)\,\in\,\CC} \,\prod_{i=1}^n \,\psi^-_{\alpha_i\beta_i}(\tau,\bar{\tau}) \\
        &= \sum_{c\,\in\,\CC}\, \prod_{(a,b)\,\in\,\BF_p\times\BF_p}\,\left(\psi^-_{ab}(\tau,\bar{\tau}) \right)^{\mathrm{wt}_{ab}(c)} \\
        &= W_\CC(\{\psi^-_{ab}\})\,,
    \end{aligned}
\end{align}
where for $(a,b)\in\BF_p\times\BF_p$ we define
\begin{align}
\begin{aligned}
    \psi^-_{ab}(\tau,\bar{\tau}) &= (-1)^{a+b} \sum_{k_1,k_2\in\BZ}\, q^{\frac{1}{2}\left(\frac{a+b}{2}+k_1+k_2\right)^2} \bar{q}^{\frac{1}{2}\left(\frac{a-b}{2}+k_1-k_2\right)^2}\,,\\
    &= (-1)^{a+b}\left(\Theta_{a+b,2}(\tau)\,\bar{\Theta}_{a-b,2}(\bar{\tau}) + \Theta_{a+b-2,2}(\tau)\,\bar{\Theta}_{a-b-2,2}(\bar{\tau})\right)\,.
\end{aligned}
\end{align}
Since we weighted the theta function $\Theta'_{\widetilde{\Lambda}(\CC)}$ according to whether a vector $\lambda$ is in the set $\Lambda_0$ or $\Lambda_1$, we obtain the following:
\begin{align}
    \begin{aligned}
    \Theta_{\widetilde{\Lambda}_0} &= \frac{1}{2}\left[\Theta_{\widetilde{\Lambda}(\CC)}(\tau,\bar{\tau}) +\Theta'_{\widetilde{\Lambda}(\CC)}(\tau,\bar{\tau})\right]\,,\\
    \Theta_{\widetilde{\Lambda}_1} &= \frac{1}{2}\left[\Theta_{\widetilde{\Lambda}(\CC)}(\tau,\bar{\tau}) -\Theta'_{\widetilde{\Lambda}(\CC)}(\tau,\bar{\tau})\right]\,,
    \end{aligned}
\end{align}
where the theta function of the Construction A lattice $\Lambda(\CC)$ is given in terms of the complete enumerator polynomial in Proposition \ref{prop:enumerator=partition}.
\end{proof}

\begin{proposition}
The theta functions of $\widetilde{\Lambda}_2$ and $\widetilde{\Lambda}_3$ are 
\begin{align}
    \begin{aligned}
    \Theta_{\widetilde{\Lambda}_2}(\tau,\bar{\tau}) &= \frac{1}{2} \left[ W_\CC (\{\tilde{\psi}^+_{ab}\}) - W_\CC(\{\tilde{\psi}^-_{ab}\})\right]\,,\\
    \Theta_{\widetilde{\Lambda}_3}(\tau,\bar{\tau}) &= \frac{1}{2} \left[ W_\CC (\{\tilde{\psi}^+_{ab}\}) + W_\CC(\{\tilde{\psi}^-_{ab}\})\right]\,,
    \end{aligned}
\end{align}
where
\begin{align}
    \begin{aligned}
    \tilde{\psi}^+_{ab}(\tau,\bar{\tau}) &= \Theta_{a+b,2}(\tau)\,\bar{\Theta}_{a-b-2,2}(\bar{\tau}) + \Theta_{a+b-2,2}(\tau)\,\bar{\Theta}_{a-b,2}(\bar{\tau})\,,\\
    \tilde{\psi}^-_{ab}(\tau,\bar{\tau}) &= \tilde{\psi}^+_{ab}(\tau+1,\bar{\tau}+1)\,.
    \end{aligned}
\end{align}
\end{proposition}

\begin{proof}
The proof is identical to the case with an odd prime $p\neq2$.
\end{proof}

For non-$\BF_4$-even codes, we have
\begin{align}
    \begin{aligned}
    \label{eq:nongf4_psi_theta}
        \psi^+_{00} &= \frac{\vartheta_3\,\bar{\vartheta}_3+\vartheta_4\,\bar{\vartheta}_4}{2}\,,&\qquad
        \psi^+_{01} &=\psi^+_{10} = \frac{\vartheta_2\,\bar{\vartheta}_2}{2}\,,&\qquad
        \psi^+_{11} &= \frac{\vartheta_3\,\bar{\vartheta}_3-\vartheta_4\,\bar{\vartheta}_4}{2}\,,\\
    
        \psi^-_{00} &= \frac{\vartheta_3\,\bar{\vartheta}_3+\vartheta_4\,\bar{\vartheta}_4}{2}\,,&\qquad
        \psi^-_{01} &=\psi^-_{10} = -\frac{\vartheta_2\,\bar{\vartheta}_2}{2}\,,&\qquad
        \psi^-_{11} &= \frac{\vartheta_3\,\bar{\vartheta}_3-\vartheta_4\,\bar{\vartheta}_4}{2}\,,\\
    
        \tilde{\psi}^+_{00} &= \frac{\vartheta_3\,\bar{\vartheta}_3-\vartheta_4\,\bar{\vartheta}_4}{2}\,,&\qquad
        \tilde{\psi}^+_{01}&= \tilde{\psi}^+_{10} = \frac{\vartheta_2\,\bar{\vartheta}_2}{2}\,,&\quad 
        \tilde{\psi}^+_{11} &= \frac{\vartheta_3\,\bar{\vartheta}_3+\vartheta_4\,\bar{\vartheta}_4}{2}\,,\\
    
        \tilde{\psi}^-_{00} &= -\frac{\vartheta_3\,\bar{\vartheta}_3-\vartheta_4\,\bar{\vartheta}_4}{2}\,,&\qquad
        \tilde{\psi}^-_{01} &= \tilde{\psi}^-_{10} =\,\,  \frac{\vartheta_2\,\bar{\vartheta}_2}{2}\,,&\quad 
        \tilde{\psi}^-_{11} &= \frac{\vartheta_3\,\bar{\vartheta}_3+\vartheta_4\,\bar{\vartheta}_4}{2}\,. 
    \end{aligned}
\end{align}

By using the theta functions of $\widetilde{\Lambda}_i$, we can compute the partition functions of each sector in the orbifold and fermionized theories.
In the next section, we utilize the results to investigate the $\BZ_2$ gaugings of Narain code CFTs.

\section{Examples}
\label{sec:examples}

In this section, we use the construction of the orbifold and fermionized theory in section~\ref{sec:gauging_code_CFT} to analyze some examples of Narain code CFTs and their $\BZ_2$ gauging.

\subsection{Quantum codes of length $n=1$}
\label{ss:n=1}

Let us consider a classical Euclidean code $C\subset\BF_p$ whose generator matrix is $G_C = [1]$ and whose parity check matrix is $H_C = [0]$.
Then, the classical code and its dual are
\begin{align}
    C = \BF_p\,,\qquad C^\perp = \{0\}\,.
\end{align}
Using the CSS construction \eqref{eq:general_CSS_check}, the check matrix of the corresponding stabilizer code is
\begin{align}
    \mathsf{H}_{(C,C^\perp)} = 
    \begin{bmatrix}
        0 & 1
    \end{bmatrix}\,.
\end{align}
Therefore, the classical code generated by $G_\mathsf{H} = \mathsf{H}$ is
\begin{align}
    \CC = \{(0,0),(0,1),\cdots,(0, p-1)\}\subset\BF_p^2\,.
\end{align}
Via Construction A, we obtain the Lorentzian even self-dual lattice
\begin{align}
    \Lambda(\CC) = \left\{\,\left(\sqrt{p}\,k_1\,,\,\frac{k_2}{\sqrt{p}}\,\right)\in\BR^2\;\middle|\; k_1,k_2\in\BZ\,\right\}\,,
\end{align}
where the lattice has the off-diagonal Lorentzian metric $\eta$.
We depict the momentum lattice in figure\,\ref{fig:original_orbifold}.
Note that $\Lambda(\CC)$ exactly agrees with the momentum lattice of a compact boson at radius $R = \sqrt{2/p}$.
In terms of a compact boson, $k_1\in\BZ$ corresponds to the momentum number $m$, and $k_2\in\BZ$ corresponds to the winding number $w$.

For an odd prime $p\neq2$, we choose an element $\chi = \sqrt{p}\,(1,1)\in\Lambda(\CC)$ according to \eqref{eq:choice_oddprime}. 
This choice, however, does not satisfy the non-anomalous condition that requires $n\in2\BZ$, and we cannot perform orbifolding and fermionization by gauging the corresponding $\BZ_2$ symmetry.
Indeed, this choice corresponds to the diagonal $\BZ_2$ subgroup of the $\U(1)_m \times \U(1)_w$ symmetry, which acts on the vertex operators as
\begin{align}
    V_{p_L,\,p_R}(z,\bar{z}) \to (-1)^{k_1+k_2}\,V_{p_L,\,p_R}(z,\bar{z})\,,
\end{align}
and cannot be gauged due to the mixed anomaly between $\U(1)_m$ and $\U(1)_w$ symmetries~\cite{Ji:2019ugf}.

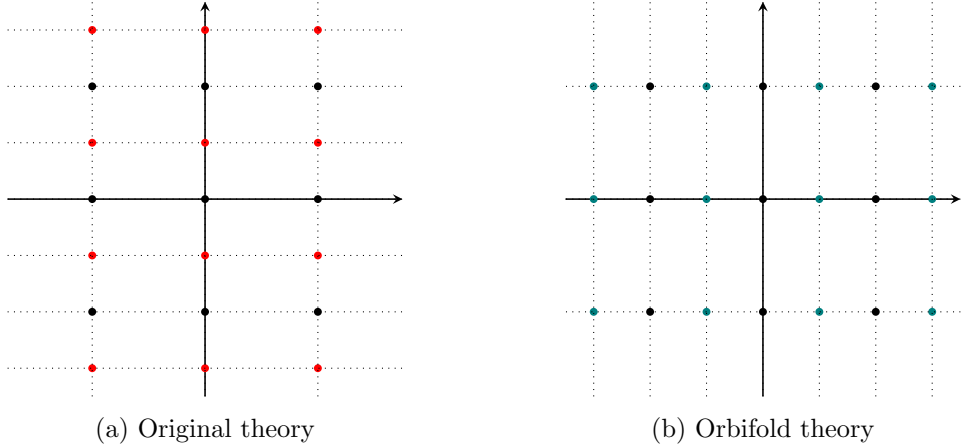
\begin{figure}
    \centering
    \begin{tabular}{cc}
        \begin{minipage}[t]{0.45\textwidth}
        \begin{center}
            \begin{tikzpicture}[scale=0.75, transform shape]
            \centering
    \draw[semithick,->,>=stealth](-3.5,0)--(3.5,0);
    \draw[semithick,->,>=stealth](0,-3.5)--(0,3.5);
\begin{scope}
    \fill[red](0,-1) circle[radius=0.07cm];
    \fill(0,-2) circle[radius=0.07cm];
    \fill[red](0,-3) circle[radius=0.07cm];
    \fill(0,0) circle[radius=0.07cm];
    \fill[red](0,1) circle[radius=0.07cm];
    \fill(0,2) circle[radius=0.07cm];
    \fill[red](0,3) circle[radius=0.07cm];
    \draw[dotted] (0,-3.5)--(0,3.5);
\end{scope}
\begin{scope}[xshift=2cm]
    \fill[red](0,-1) circle[radius=0.07cm];
    \fill(0,-2) circle[radius=0.07cm];
    \fill[red](0,-3) circle[radius=0.07cm];
    \fill(0,0) circle[radius=0.07cm];
    \fill[red](0,1) circle[radius=0.07cm];
    \fill(0,2) circle[radius=0.07cm];
    \fill[red](0,3) circle[radius=0.07cm];
    \draw[dotted] (0,-3.5)--(0,3.5);
\end{scope}
\begin{scope}[xshift=-2cm]
    \fill[red](0,-1) circle[radius=0.07cm];
    \fill(0,-2) circle[radius=0.07cm];
    \fill[red](0,-3) circle[radius=0.07cm];
    \fill(0,0) circle[radius=0.07cm];
    \fill[red](0,1) circle[radius=0.07cm];
    \fill(0,2) circle[radius=0.07cm];
    \fill[red](0,3) circle[radius=0.07cm];
    \draw[dotted] (0,-3.5)--(0,3.5);
\end{scope}

    \draw[dotted] (-3.5,-3)--(3.5,-3);
    \draw[dotted] (-3.5,-2)--(3.5,-2);
    \draw[dotted] (-3.5,-1)--(3.5,-1);
    \draw[dotted] (-3.5,0)--(3.5,0);
    \draw[dotted] (-3.5,1)--(3.5,1);
    \draw[dotted] (-3.5,2)--(3.5,2);
    \draw[dotted] (-3.5,3)--(3.5,3);
\end{tikzpicture}
        \end{center}
            \subcaption{Original theory}
        \end{minipage} & 
        \begin{minipage}[t]{0.45\textwidth}
        \begin{center}
            \begin{tikzpicture}[scale=0.75, transform shape, rotate=90]
\begin{scope}[rotate=-90]
    \draw[semithick,->,>=stealth](-3.5,0)--(3.5,0);
    \draw[semithick,->,>=stealth](0,-3.5)--(0,3.5);
\end{scope}
\begin{scope}
    \fill[teal](0,-1) circle[radius=0.07cm];
    \fill(0,-2) circle[radius=0.07cm];
    \fill[teal](0,-3) circle[radius=0.07cm];
    \fill(0,0) circle[radius=0.07cm];
    \fill[teal](0,1) circle[radius=0.07cm];
    \fill(0,2) circle[radius=0.07cm];
    \fill[teal](0,3) circle[radius=0.07cm];
    \draw[dotted] (0,-3.5)--(0,3.5);
\end{scope}
\begin{scope}[xshift=2cm]
    \fill[teal](0,-1) circle[radius=0.07cm];
    \fill(0,-2) circle[radius=0.07cm];
    \fill[teal](0,-3) circle[radius=0.07cm];
    \fill(0,0) circle[radius=0.07cm];
    \fill[teal](0,1) circle[radius=0.07cm];
    \fill(0,2) circle[radius=0.07cm];
    \fill[teal](0,3) circle[radius=0.07cm];
    \draw[dotted] (0,-3.5)--(0,3.5);
\end{scope}
\begin{scope}[xshift=-2cm]
    \fill[teal](0,-1) circle[radius=0.07cm];
    \fill(0,-2) circle[radius=0.07cm];
    \fill[teal](0,-3) circle[radius=0.07cm];
    \fill(0,0) circle[radius=0.07cm];
    \fill[teal](0,1) circle[radius=0.07cm];
    \fill(0,2) circle[radius=0.07cm];
    \fill[teal](0,3) circle[radius=0.07cm];
    \draw[dotted] (0,-3.5)--(0,3.5);
\end{scope}

    \draw[dotted] (-3.5,-3)--(3.5,-3);
    \draw[dotted] (-3.5,-2)--(3.5,-2);
    \draw[dotted] (-3.5,-1)--(3.5,-1);
    \draw[dotted] (-3.5,0)--(3.5,0);
    \draw[dotted] (-3.5,1)--(3.5,1);
    \draw[dotted] (-3.5,2)--(3.5,2);
    \draw[dotted] (-3.5,3)--(3.5,3);
\end{tikzpicture}
        \end{center}
            \subcaption{Orbifold theory}
        \end{minipage}
    \end{tabular}
    \caption{The Lorentzian lattice $\Lambda(\CC)$ for the original Narain code CFT (the left panel) and the one $\Lambda_\CO$ for its orbifold (the right panel). On the left panel, the black and red points represent $\Lambda_0$ and $\Lambda_1$, respectively. On the right panel, the black and green points express $\Lambda_0$ and $\Lambda_3$, respectively.}
    \label{fig:original_orbifold}
\end{figure}

For the binary case $(p=2)$, we take an element $\chi = \sqrt{2}\,(1,1)\in\Lambda(\CC)$ by \eqref{eq:choice_nongf4} because the classical code $\CC$ does not contain the all-ones vector: $(1,1)\notin \CC$, then $\CC$ is a non-$\BF_4$-even code.
In this case, the non-anomalous condition is $n\in\BZ$, so we can take the orbifold and fermionization by $\chi = \sqrt{2}\,(1,1)$.
Indeed, this choice corresponds to the winding $\BZ_2$ symmetry,\footnote{If we exchange the roles of $C$ and $C^\perp$, the resulting code CFT has the non-anomalous momentum $\BZ_2$ symmetry specified by the same vector $\chi=\sqrt{2}\,(1,1)$.
More generally, exchanging the roles of $C$ and $C^\perp$ in the CSS construction acts as a T-duality on Narain code CFTs.
} which acts on the vertex operators as
\begin{align}
    V_{p_L,\,p_R}(z,\bar{z}) \to (-1)^{k_2}\,V_{p_L,\,p_R}(z,\bar{z})\,.
\end{align}
We can divide the Construction A lattice into two parts
\begin{align}
    \begin{aligned}
        \Lambda_0 &= \left\{\,\left(\sqrt{2}\,k_1\,,\,\sqrt{2}\,k_2\,\right)\in\BR^2\;\middle|\; k_1\in\BZ\,,\; k_2\in\BZ\,\right\}\,,\\
        \Lambda_1 &= \left\{\,\left(\sqrt{2}\,k_1\,,\,\frac{k_2}{\sqrt{2}}\,\right)\in\BR^2\;\middle|\; k_1\in\BZ\,,\; k_2\in2\BZ+1\,\right\}\,,
    \end{aligned}
\end{align}
which implies that $\Lambda(\CC)$ is graded only by the value of $k_2$. Then, the $\BZ_2$ symmetry corresponding to the choice of $\chi\in\Lambda(\CC)$ is purely the winding $\BZ_2$ symmetry without anomaly.
The other shifted sets are
\begin{align}
    \begin{aligned}
        \Lambda_2 &= \left\{\,\left(\frac{l_1}{\sqrt{2}}\,,\,\frac{l_2}{\sqrt{2}}\,\right)\in\BR^2\;\middle|\; l_1\in2\BZ+1\,,\;l_2\in2\BZ+1\,\right\}\,,\\
        \Lambda_3 &= \left\{\,\left(\frac{l_1}{\sqrt{2}}\,,\,\sqrt{2}\,l_2\,\right)\in\BR^2\;\middle|\; l_1\in2\BZ+1\,,\; l_2\in\BZ\,\right\}\,.
    \end{aligned}
\end{align}

The orbifold gives the momentum lattice $\Lambda_\CO = \Lambda_0\cup\Lambda_3$, which can be written as
\begin{align}
    \Lambda_\CO = \left\{\,\left(\frac{k_1}{\sqrt{2}}\,,\,\sqrt{2}\,k_2\,\right)\in\BR^2\;\middle|\; k_1\in\BZ\,,\; k_2\in\BZ\,\right\}\,.
\end{align}
This matches the momentum lattice of a compact boson of radius $R=2$. Since the original Construction A lattice corresponds to the radius $R=1$ for $p=2$, the $\BZ_2$ orbifold by $\chi = \sqrt{2}\,(1,1)$ makes the target circle only twice as large.
In our notation, the T-duality acts on the radius $R$ as $R\to 2/R$. Then, the orbifolded theory, which is a compact boson of radius $R=2$, is equivalent to the original theory $(R=1)$.

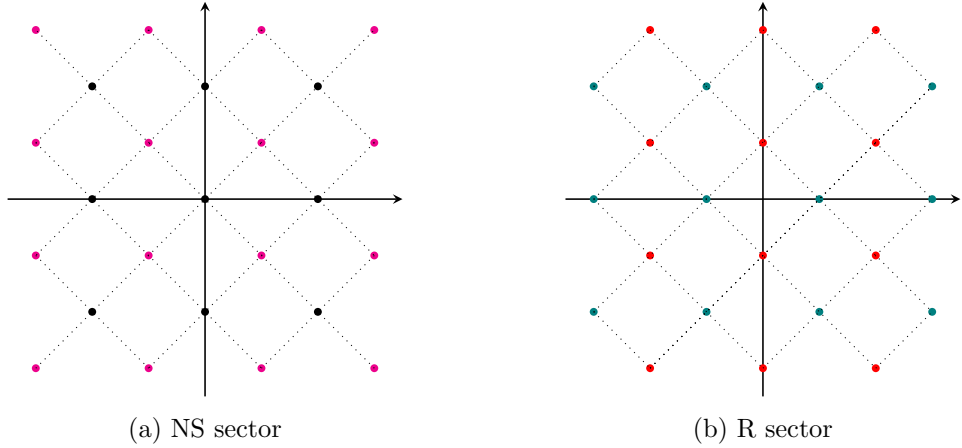
\begin{figure}
    \centering
    \begin{tabular}{cc}
        \begin{minipage}{0.45\textwidth}
            \begin{center}
\begin{tikzpicture}[scale=0.75, transform shape]
    \draw[semithick,->,>=stealth](-3.5,0)--(3.5,0);
    \draw[semithick,->,>=stealth](0,-3.5)--(0,3.5);
    \begin{scope}
    \fill[magenta](1,-1) circle[radius=0.07cm];
    \fill(0,-2) circle[radius=0.07cm];
    \fill[magenta](1,-3) circle[radius=0.07cm];
    \fill(0,0) circle[radius=0.07cm];
    \fill[magenta](1,1) circle[radius=0.07cm];
    \fill(0,2) circle[radius=0.07cm];
    \fill[magenta](1,3) circle[radius=0.07cm];
    \end{scope}
    \begin{scope}[xshift=2cm]
    \fill[magenta](1,-1) circle[radius=0.07cm];
    \fill(0,-2) circle[radius=0.07cm];
    \fill[magenta](1,-3) circle[radius=0.07cm];
    \fill(0,0) circle[radius=0.07cm];
    \fill[magenta](1,1) circle[radius=0.07cm];
    \fill(0,2) circle[radius=0.07cm];
    \fill[magenta](1,3) circle[radius=0.07cm];
    \end{scope}
    \begin{scope}[xshift=-2cm]
    \fill[magenta](1,-1) circle[radius=0.07cm];
    \fill(0,-2) circle[radius=0.07cm];
    \fill[magenta](1,-3) circle[radius=0.07cm];
    \fill(0,0) circle[radius=0.07cm];
    \fill[magenta](1,1) circle[radius=0.07cm];
    \fill(0,2) circle[radius=0.07cm];
    \fill[magenta](1,3) circle[radius=0.07cm];
    \end{scope}
    \begin{scope}[xshift=-4cm]
    \fill[magenta](1,-1) circle[radius=0.07cm];
    \fill[magenta](1,-3) circle[radius=0.07cm];
    \fill[magenta](1,1) circle[radius=0.07cm];
    \fill[magenta](1,3) circle[radius=0.07cm];
    \end{scope}
    
    \draw[dotted] (-3,-3)--(3,3);
    \draw[dotted] (1,3)--(-3,-1);
    \draw[dotted] (-1,3)--(-3,1);
    \draw[dotted] (3,1)--(-1,-3);
    \draw[dotted] (3,-1)--(1,-3);
    \draw[dotted] (3,-3)--(-3,3);
    \draw[dotted] (1,3)--(3,1);
    \draw[dotted] (-1,3)--(3,-1);
    \draw[dotted] (-3,1)--(1,-3);
    \draw[dotted] (-3,-1)--(-1,-3);
\end{tikzpicture}
\subcaption{NS sector}
\end{center}
        \end{minipage} & 
        \begin{minipage}{0.45\textwidth}
            \begin{center}
\begin{tikzpicture}[scale=0.75, transform shape]
    \draw[semithick,->,>=stealth](-3.5,0)--(3.5,0);
    \draw[semithick,->,>=stealth](0,-3.5)--(0,3.5);
    \begin{scope}[xshift=1cm]
    \begin{scope}
    \fill[red](1,-1) circle[radius=0.07cm];
    \fill[teal](0,-2) circle[radius=0.07cm];
    \fill[red](1,-3) circle[radius=0.07cm];
    \fill[teal](0,0) circle[radius=0.07cm];
    \fill[red](1,1) circle[radius=0.07cm];
    \fill[teal](0,2) circle[radius=0.07cm];
    \fill[red](1,3) circle[radius=0.07cm];
    \end{scope}
    \begin{scope}[xshift=2cm]
    \fill[teal](0,-2) circle[radius=0.07cm];
    \fill[teal](0,0) circle[radius=0.07cm];
    \fill[teal](0,2) circle[radius=0.07cm];
    \end{scope}
    \begin{scope}[xshift=-2cm]
    \fill[red](1,-1) circle[radius=0.07cm];
    \fill[teal](0,-2) circle[radius=0.07cm];
    \fill[red](1,-3) circle[radius=0.07cm];
    \fill[teal](0,0) circle[radius=0.07cm];
    \fill[red](1,1) circle[radius=0.07cm];
    \fill[teal](0,2) circle[radius=0.07cm];
    \fill[red](1,3) circle[radius=0.07cm];
    \end{scope}
    \begin{scope}[xshift=-4cm]
    \fill[red](1,-1) circle[radius=0.07cm];
    \fill[teal](0,-2) circle[radius=0.07cm];
    \fill[red](1,-3) circle[radius=0.07cm];
    \fill[teal](0,0) circle[radius=0.07cm];
    \fill[red](1,1) circle[radius=0.07cm];
    \fill[teal](0,2) circle[radius=0.07cm];
    \fill[red](1,3) circle[radius=0.07cm];
    \end{scope}
    \end{scope}
    \draw[dotted] (2,-3)--(3,-2);
    \draw[dotted] (0,-3)--(3,0);
    \draw[dotted] (-2,-3)--(3,2);
    \draw[dotted] (-3,-2)--(2,3);
    \draw[dotted] (-2,-3)--(3,2);
    \draw[dotted] (-3,0)--(0,3);
    \draw[dotted] (-3,2)--(-2,3);
    
    \draw[dotted] (-3,-2)--(-2,-3);
    \draw[dotted] (-3,0)--(0,-3);
    \draw[dotted] (-3,2)--(2,-3);
    \draw[dotted] (-2,3)--(3,-2);
    \draw[dotted] (-2,-3)--(3,2);
    \draw[dotted] (0,3)--(3,0);
    \draw[dotted] (2,3)--(3,2);
\end{tikzpicture}
\subcaption{R sector}
\end{center}
        \end{minipage}
    \end{tabular}
    \caption{The Lorentzian lattice $\Lambda_{\mathrm{NS}}$ for the NS sector (the left panel) and $\Lambda_{\mathrm{R}}$ for the R sector (the right panel) in the fermionized theory. On the left panel, the black and pink points represent $\Lambda_0$ and $\Lambda_2$, respectively. On the right panel, the red and green points express $\Lambda_1$ and $\Lambda_3$, respectively.}
    \label{fig:fermionized}
\end{figure}

In the fermionized theory, the NS sector can be built out of the vertex operators $V_{\lambda}(z)$ where $\lambda$ is in the momentum lattice $\Lambda_{\mathrm{NS}} = \Lambda_0 \cup \Lambda_2$ (see figure~\ref{fig:fermionized})
\begin{align}
    \Lambda_{\mathrm{NS}}= \left\{\,\left(\frac{l_1}{\sqrt{2}}\,,\,\frac{l_2}{\sqrt{2}}\,\right)\in\BR^2\;\middle|\;  l_1 = l_2 \;\;\mathrm{mod}\;2\,,\;l_1\in\BZ\,,\;l_2\in\BZ\,\right\}\,.
\end{align}
The Ramond sector can be obtained by the momentum set $\Lambda_{\mathrm{R}} = \Lambda_1\cup\Lambda_3$ where
\begin{align}
    \Lambda_{\mathrm{R}}= \left\{\,\left(\frac{l_1}{\sqrt{2}}\,,\,\frac{l_2}{\sqrt{2}}\,\right)\in\BR^2\;\middle|\;  l_1 \neq l_2 \;\;\mathrm{mod}\;2\,,\;l_1\in\BZ\,,\;l_2\in\BZ\,\right\}\,.
\end{align}

The partition functions for these theories can be easily computed using the technique given in the previous section.
For $p=2$, the complete weight enumerator of $\CC$ is
\begin{align}
    W_\CC (\{x_{ab}\}) = x_{00} + x_{10}\,.
\end{align}
From \eqref{eq:theta_p=2}, the partition function of the original bosonic theory is
\begin{align}
    Z_\CC (\tau,\bar{\tau}) = \frac{1}{|\eta(\tau)|^2}\frac{\vartheta_2\,\bar{\vartheta}_2+\vartheta_3\,\bar{\vartheta}_3+ \vartheta_4\,\bar{\vartheta}_4}{2}\,.
\end{align}
The orbifolded partition function is
\begin{align}
\begin{aligned}
    Z^\CO_\CC(\tau,\bar{\tau}) &= \frac{1}{2|\eta(\tau)|^2} \left[W_\CC (\{\psi^+_{ab}\}) + W_\CC(\{{\psi}^-_{ab}\}) +  W_\CC (\{\tilde{\psi}^+_{ab}\}) + W_\CC(\{\tilde{\psi}^-_{ab}\})\right]\,\\
    & = \frac{1}{|\eta(\tau)|^2}\frac{\vartheta_2\,\bar{\vartheta}_2+\vartheta_3\,\bar{\vartheta}_3+ \vartheta_4\,\bar{\vartheta}_4}{2}\,,
\end{aligned}
\end{align}
where we used \eqref{eq:nongf4_psi_theta}.
This is the same as the one for the original theory, which is consistent with the fact that they are related by T-duality.
On the other hand, the fermionized theory has four partition functions depending on the choice of spin structures $\rho$. On the torus, the spin structure is $\rho = (a,b)\in\BZ_2\times\BZ_2$. Following \eqref{eq:fermionized_part_gen}, we denote the corresponding partition functions by $Z^\CF_\CC[\rho]$. Then, we have
\begin{align}
    \begin{aligned}
        Z^\CF_\CC[00] &= \frac{1}{2|\eta(\tau)|^2} \left[W_\CC (\{\psi^+_{ab}\}) + W_\CC(\{{\psi}^-_{ab}\}) +  W_\CC (\{\tilde{\psi}^+_{ab}\}) - W_\CC(\{\tilde{\psi}^-_{ab}\})\right] = \left|\frac{\vartheta_3}{\eta}\right|^2\ ,\\
        
        Z^\CF_\CC[10] &= \frac{1}{2|\eta(\tau)|^2} \left[W_\CC (\{\psi^+_{ab}\}) + W_\CC(\{{\psi}^-_{ab}\}) -  W_\CC (\{\tilde{\psi}^+_{ab}\}) + W_\CC(\{\tilde{\psi}^-_{ab}\})\right] = \left|\frac{\vartheta_4}{\eta}\right|^2\ ,\\

        Z^\CF_\CC[01] &= \frac{1}{2|\eta(\tau)|^2} \left[W_\CC (\{\psi^+_{ab}\}) - W_\CC(\{{\psi}^-_{ab}\}) +  W_\CC (\{\tilde{\psi}^+_{ab}\}) + W_\CC(\{\tilde{\psi}^-_{ab}\})\right] = \left|\frac{\vartheta_2}{\eta}\right|^2\ ,\\

        Z^\CF_\CC[11] &= \frac{1}{2|\eta(\tau)|^2} \left[W_\CC (\{\psi^+_{ab}\}) - W_\CC(\{{\psi}^-_{ab}\}) -  W_\CC (\{\tilde{\psi}^+_{ab}\}) - W_\CC(\{\tilde{\psi}^-_{ab}\})\right] = 0\ .
    \end{aligned}
\end{align}
These are the partition functions of a free Dirac fermion. Therefore, we conclude that the fermionized theory of a compact boson of radius $R=1$ is a free Dirac fermion. This can be expected because a Dirac fermion has been known to be equivalent to a compact boson with $R=1$ via bosonization \cite{Ginsparg:1988ui}.

\subsection{CSS construction with self-dual codes}

In this subsection, we consider a CSS code defined by a single self-dual code $C=C^\perp\subset\BF_p^n$.
Then, the resulting classical code is $\CC = C\times C\subset\BF_p^{2n}$ via CSS construction \eqref{eq:general_CSS_check}.
In this case, the Construction A lattice $\Lambda(\CC)$ is given by 
\begin{align}
    \Lambda(\CC) = \Lambda_\mathrm{E}(C)\times \Lambda_\mathrm{E}(C)\ ,
\end{align}
where we denote by $\Lambda_\mathrm{E}(C)$ the Euclidean Construction A lattice \cite{nebe2006self}:
\begin{align}
    \Lambda_\mathrm{E}(C) 
        =
            \left\{ \frac{c + p\,m}{\sqrt{p}}\; \bigg|\; c\in C, ~ m\in \BZ^{n} \right\}\ .
\end{align}
The Euclidean Construction A lattice is odd self-dual when $C$ is a self-dual code for an odd prime $p$.
Since the Euclidean Construction A lattice is odd self-dual, one can introduce a characteristic vector $\mathbf{x}\in\Lambda_\mathrm{E}(C)$ \cite{conway2013sphere,milnor1973symmetric,serre2012course} such that for any $\lambda\in\Lambda_\mathrm{E}(C)$
\begin{align}
    \mathbf{x}\cdot \lambda = \lambda\cdot \lambda\mod 2\,.
\end{align}
For the Construction A lattice from a $p$-ary code, we can choose $\mathbf{x} = \sqrt{p}\,\mathbf{1}_n\in\Lambda_\mathrm{E}(C)$ \cite{Gaiotto:2018ypj,Kawabata:2023nlt}.
By using the characteristic vector, we divide $\Lambda_\mathrm{E}(C)$ into two parts
\begin{align}
    \Lambda_\mathrm{E}(C) = \Lambda_\mathrm{E}^{(0)}\cup\Lambda_\mathrm{E}^{(2)}
\end{align}
where
\begin{align}
    \begin{aligned}
    \Lambda_\mathrm{E}^{(0)} &=  \left\{ \lambda\in\Lambda_\mathrm{E}(C) \mid \mathbf{x}\cdot\lambda = 0 \mod 2 \right\}\,,
    \\
    \Lambda_\mathrm{E}^{(2)} &= \left\{ \lambda\in\Lambda_\mathrm{E}(C) \mid \mathbf{x}\cdot\lambda = 1 \mod 2 \right\}\,.
\end{aligned}
\end{align}
We can also introduce the shadow of $\Lambda_\mathrm{E}(C)$ \cite{conway1990new} by
\begin{align}
    S(\Lambda_\mathrm{E}(C)) = \Lambda_\mathrm{E}(C) +\frac{\mathbf{x}}{2} = \Lambda_\mathrm{E}^{(1)} \cup \Lambda_\mathrm{E}^{(3)}\,,
\end{align}
where we divide the shadow into two parts
\begin{align}
    \begin{aligned}
            \Lambda_\mathrm{E}^{(1)} &=  \left\{ \,\lambda+\frac{\mathbf{x}}{2}\in S(\Lambda(C)) \;\middle|\; \mathbf{x}\cdot\left(\lambda+\frac{\mathbf{x}}{2}\right) = 0 \mod 2\, \right\}\,,
    \\
    \Lambda_\mathrm{E}^{(3)} &= \left\{ \,\lambda+\frac{\mathbf{x}}{2}\in S(\Lambda(C)) \;\middle|\; \mathbf{x}\cdot\left(\lambda+\frac{\mathbf{x}}{2}\right) = 1 \mod 2\, \right\}\,.
    \end{aligned}
\end{align}

Let us consider the orbifold and fermionization of Narain CFTs based on the Construction A lattice $\Lambda(\CC)$.
The choice of an element $\chi = \sqrt{p}\,\mathbf{1}_{2n}\in\Lambda(\CC)$ is always non-anomalous.
Then, one can denote it by $\chi = (\mathbf{x},\mathbf{x})\in\Lambda_\mathrm{E}(C)\times \Lambda_\mathrm{E}(C)$.
The $\BZ_2$ grading of the Construction A lattice $\Lambda(\CC)$ gives $\Lambda(\CC) = \Lambda_0 \cup \Lambda_1$ where
\begin{align}
    \begin{aligned}
        \Lambda_0 &= \left\{(\lambda_1,\lambda_2)\in\Lambda(C)\times\Lambda(C)\mid \chi \odot (\lambda_1,\lambda_2) = 0\mod 2\right\}\\
        &= \left(\Lambda_\mathrm{E}^{(0)}\times \Lambda_\mathrm{E}^{(0)}\right) \cup \left(\Lambda_\mathrm{E}^{(2)} \times \Lambda_\mathrm{E}^{(2)}\right)\,,
    \end{aligned}
\end{align}
and
\begin{align}
    \Lambda_1 = \left(\Lambda_\mathrm{E}^{(0)}\times \Lambda_\mathrm{E}^{(2)}\right) \cup \left(\Lambda_\mathrm{E}^{(2)} \times \Lambda_\mathrm{E}^{(0)}\right)\,.
\end{align}
On the other hand, by shifting $\Lambda_0$ and $\Lambda_1$ by $\delta = \chi/2 = (\mathbf{x},\mathbf{x})/2$, we obtain
\begin{align}
\begin{aligned}
    \Lambda_2 &= \left(\Lambda_\mathrm{E}^{(1)}\times \Lambda_\mathrm{E}^{(1)}\right) \cup \left(\Lambda_\mathrm{E}^{(3)} \times \Lambda_\mathrm{E}^{(3)}\right)\,,\\
    \Lambda_3 &= \left(\Lambda_\mathrm{E}^{(1)}\times \Lambda_\mathrm{E}^{(3)}\right) \cup \left(\Lambda_\mathrm{E}^{(3)} \times \Lambda_\mathrm{E}^{(1)}\right)\,.
\end{aligned}
\end{align}
Therefore, the momentum lattice in the orbifold theory is
\begin{align}
    \Lambda_\CO = \left(\Lambda_\mathrm{E}^{(0)}\times \Lambda_\mathrm{E}^{(0)}\right) \cup \left(\Lambda_\mathrm{E}^{(2)} \times \Lambda_\mathrm{E}^{(2)}\right) \cup \left(\Lambda_\mathrm{E}^{(1)}\times \Lambda_\mathrm{E}^{(3)}\right) \cup \left(\Lambda_\mathrm{E}^{(3)} \times \Lambda_\mathrm{E}^{(1)}\right)\,.
\end{align}
The momenta in the NS sector and the R sector in the fermionized theory are given by
\begin{align}
    \begin{aligned}
        \Lambda_\mathrm{NS} &= \left(\Lambda_\mathrm{E}^{(0)}\times \Lambda_\mathrm{E}^{(0)}\right) \cup \left(\Lambda_\mathrm{E}^{(2)} \times \Lambda_\mathrm{E}^{(2)}\right) \cup \left(\Lambda_\mathrm{E}^{(1)}\times \Lambda_\mathrm{E}^{(1)}\right) \cup \left(\Lambda_\mathrm{E}^{(3)} \times \Lambda_\mathrm{E}^{(3)}\right)\,,\\
        \Lambda_\mathrm{R} &= \left(\Lambda_\mathrm{E}^{(0)}\times \Lambda_\mathrm{E}^{(2)}\right) \cup \left(\Lambda_\mathrm{E}^{(2)} \times \Lambda_\mathrm{E}^{(0)}\right) \cup \left(\Lambda_\mathrm{E}^{(1)}\times \Lambda_\mathrm{E}^{(3)}\right) \cup \left(\Lambda_\mathrm{E}^{(3)} \times \Lambda_\mathrm{E}^{(1)}\right)\,.
    \end{aligned}
\end{align}

\subsubsection{$n=2\,,\,$ $p=5$ case ($C_{5,2}$)}

Here, we consider the CSS code defined by a single self-dual code whose generator matrix is $G_{C_{5,2}} = [1 \, 2]$. Then, the codewords consist of
\begin{align}
\label{eq:5,2}
    C_{5,2} = \left\{\,(0,0),(1,2),(2,4),(3,1),(4,3)\,\right\}\,.
\end{align}
The Euclidean Construction A lattice from $C_{5,2}$ is
\begin{align}
    \Lambda_\mathrm{E}(C_{5,2}) = \left\{ \frac{c+5m}{\sqrt{5}}\in\BR^2\;\middle|\; c\in C_{5,2}\,,\,m\in\BZ^2\right\}\,.
\end{align}
We can take a basis of $\Lambda(C_{5,2})$ by $v_1 = \frac{1}{\sqrt{5}}(1,2)$ and $v_2 = \frac{1}{\sqrt{5}}(2,-1)$, which is orthonormal: $v_i\cdot v_j=\delta_{ij}$ $(i,j=1,2)$.
Then the elements $\lambda\in\Lambda(C_{5,2})$ of the Euclidean Construction A lattice are written as $\lambda = m_1\,v_1 + m_2\, v_2$ where $m_1$, $m_2\in\BZ$. Thus, the Euclidean Construction A lattice is a two-dimensional square lattice.

Choosing a characteristic vector by $\mathbf{x} = \sqrt{5}\,(1,1)$, we can divide the Euclidean Construction A lattice into $\Lambda_\mathrm{E}^{(i)}$ ($i=0,2$) where
\begin{align}
    \begin{aligned}
        \Lambda_\mathrm{E}^{(0)} &= \left\{ m_1\, v_1 + m_2\, v_2\mid m_1 + m_2 \in 2\BZ\,,\,m_1,m_2\in\BZ \right\}\,,\\
        \Lambda_\mathrm{E}^{(2)} &= \left\{ m_1 \, v_1 + m_2 \, v_2\mid m_1 + m_2 \in 2\BZ+1\,,\,m_1,m_2\in\BZ \right\}\,.
    \end{aligned}
\end{align}

On the other hand, the shadow of $\Lambda(C_{5,2})$ is
\begin{align}
    S(\Lambda_\mathrm{E}(C_{5,2})) = \Lambda_\mathrm{E}(C_{5,2}) +\frac{\mathbf{x}}{2} = \left\{\frac{c+5m}{\sqrt{5}}\in\BR^2 \;\middle|\;c\in C_{2}\,,\,\,m\in\left(\BZ+\frac{1}{2}\right)^2\right\}\,.
\end{align}
In the basis by $v_1,v_2$, the characteristic vector $\mathbf{x} = \sqrt{5}\,(1,1)$ is represented by $\mathbf{x} =  3v_1 + v_2$, so the shadow $S(\Lambda_\mathrm{E}(C_{5,2})) = \Lambda_\mathrm{E}(C_{5,2})+\frac{\mathbf{x}}{2}$ consists of $\widetilde{\lambda} = m_1'\,v_1 +m_2'\, v_2$ where $m_1'$, $m_2'\in \BZ+\frac{1}{2}$. 
Hence, the shadow of $\Lambda_\mathrm{E}(C_{5,2})$ consists of points lying at the center of each square in the two-dimensional square lattice $\Lambda_\mathrm{E}(C_{5,2})$.
The $\BZ_2$ grading of the shadow is
\begin{align}
    \begin{aligned}
        \Lambda_\mathrm{E}^{(1)} &= \left\{ m_1' \,v_1 + m_2' \,v_2\;\middle|\; m_1' + m_2' \in 2\BZ+1\,,\,\,m_1',m_2'\in\BZ+\frac{1}{2} \right\}\,,\\
        \Lambda_\mathrm{E}^{(3)} &= \left\{ m_1' \,v_1 + m_2' \,v_2\;\middle|\; m_1' + m_2' \in 2\BZ\,,\,\,m_1',m_2'\in\BZ+\frac{1}{2} \right\}\,.
    \end{aligned}
\end{align}

The above discussion in the Euclidean signature can be used to provide the set of momenta in the Lorentzian signature.
By combining $\Lambda_\mathrm{E}^{(i)}$, we obtain $\Lambda_i$ whose elements take the form $(m_1\, v_1 + m_2\, v_2\,,\,m_1'\,v_1 + m_2'\,v_2)$ where the parameters $m_1,m_2,m_1',m_2'$ run over the regions
\begin{align}
\begin{aligned}
\label{eq:exmp_lattice}
    \Lambda_0:&\qquad m_1\,,\,m_2\,,\,m_1'\,,\,m_2'\in\BZ\,\qquad\qquad \;\;\, m_1 + m_2 = m_1' + m_2'&  & \mod 2\,, \\
    \Lambda_1:&\qquad m_1\,,\,m_2\,,\,m_1'\,,\,m_2'\in\BZ\,\qquad \qquad \;\;\, m_1 + m_2 = m_1' + m_2'+1 &  & \mod 2\,, \\
    \Lambda_2:&\qquad m_1\,,\,m_2\,,\,m_1'\,,\,m_2'\in\BZ+1/2\,\qquad m_1 + m_2 = m_1' + m_2' &  & \mod 2\,, \\
    \Lambda_3:&\qquad m_1\,,\,m_2\,,\,m_1'\,,\,m_2'\in\BZ+1/2\,\qquad m_1 + m_2 = m_1' + m_2'+1 &  & \mod 2\,.
\end{aligned}
\end{align}
Then, the lattice of the orbifold theory $\CO$ is $\Lambda_\CO = \Lambda_0\cup \Lambda_3$.
Similarly, the lattice of the NS sector in the fermionized theory is $\Lambda_\mathrm{NS} = \Lambda_0\cup \Lambda_2$ and the set for the Ramond sector is $\Lambda_\mathrm{R} = \Lambda_1\cup \Lambda_3$.

Rotating the sets $\Lambda_i\ni (\lambda_1,\lambda_2)$ into the momentum basis by $p_L = (\lambda_1+\lambda_2)/\sqrt{2}$ and $p_R = (\lambda_1-\lambda_2)/\sqrt{2}$, we obtain the sets $\widetilde{\Lambda}_i$ of left- and right-moving moementa. Then, the theta functions of $\widetilde{\Lambda}_i$ are
\begin{align}
    \Theta_{\widetilde{\Lambda}_i}(\tau,\bar{\tau}) = \sum_{m_1,m_1',m_2,m_2'}\; q^{\frac{1}{4}\left\{(m_1+m_1')^2 + (m_2 + m_2')^2\right\}}\; \bar{q}^{\frac{1}{4}\left\{(m_1-m_1')^2 + (m_2 - m_2')^2\right\}}\,,
\end{align}
where the sum is taken under the conditions of \eqref{eq:exmp_lattice} depending on $i=0,1,2,3$. 
The orbifold partition function is given by 
\begin{align}
\begin{aligned}
    Z_\CC^\CO(\tau,\bar{\tau}) = \frac{1}{|\eta(\tau)|^4} \left[\Theta_{\widetilde{\Lambda}_0}(\tau,\bar{\tau}) + \Theta_{\widetilde{\Lambda}_3}(\tau,\bar{\tau})  \right]\,.
\end{aligned}
\end{align}
In the fermionized theory, depending on the choice of the spin structures, the partition functions become
\begin{align*}
    Z_\CC^\CF[00] = \frac{1}{|\eta(\tau)|^4} \left[\Theta_{\widetilde{\Lambda}_0}(\tau,\bar{\tau}) + \Theta_{\widetilde{\Lambda}_2}(\tau,\bar{\tau})  \right]\,,\qquad
    Z_\CC^\CF[10] = \frac{1}{|\eta(\tau)|^4} \left[\Theta_{\widetilde{\Lambda}_0}(\tau,\bar{\tau}) - \Theta_{\widetilde{\Lambda}_2}(\tau,\bar{\tau})  \right]\,,\\
    Z_\CC^\CF[01] = \frac{1}{|\eta(\tau)|^4} \left[\Theta_{\widetilde{\Lambda}_1}(\tau,\bar{\tau}) + \Theta_{\widetilde{\Lambda}_3}(\tau,\bar{\tau})  \right]\,,\qquad
    Z_\CC^\CF[11] = \frac{1}{|\eta(\tau)|^4} \left[\Theta_{\widetilde{\Lambda}_1}(\tau,\bar{\tau}) - \Theta_{\widetilde{\Lambda}_3}(\tau,\bar{\tau})  \right]\,.
\end{align*}

\subsubsection{$n=6\,,\,$ $p=5$ case ($C_{5,2}^3$)}

Let us consider the Euclidean classical code $C_{5,2}^3$ generated by
\begin{align}
    G_{C_{5,2}^3} = 
    \begin{bmatrix}
        1 & 2 & 0 & 0 & 0 & 0\\
        0 & 0 & 1 & 2 & 0 & 0\\
        0 & 0 & 0 & 0 & 1 & 2
    \end{bmatrix}\,.
\end{align}
The classical code is given by $C_{5,2}^3 = C_{5,2}\times C_{5,2} \times C_{5,2}$ where $C_{5,2}$ is the classical code in \eqref{eq:5,2}.
Then, the Euclidean Construction A lattice is also written as $\Lambda_\mathrm{E}(C_{5,2}^3) = \Lambda_\mathrm{E}(C_{5,2})\times\Lambda_\mathrm{E}(C_{5,2})\times\Lambda_\mathrm{E}(C_{5,2})$.
Hence, we can choose an orthonormal basis of $\Lambda_\mathrm{E}(C_{5,2}^3)$ as follows:
\begin{align}
\begin{aligned}
    v_1 = \frac{1}{\sqrt{5}}\,(1,2,0,0,0,0)\,,\qquad v_2 = \frac{1}{\sqrt{5}}\,(2,-1,0,0,0,0)\,,\\
    v_3 = \frac{1}{\sqrt{5}}\,(0,0,1,2,0,0)\,,\qquad v_4 = \frac{1}{\sqrt{5}}\,(0,0,2,-1,0,0)\,,\\
    v_5 = \frac{1}{\sqrt{5}}\,(0,0,0,0,1,2)\,,\qquad v_6 = \frac{1}{\sqrt{5}}\,(0,0,0,0,2,-1)\,,
\end{aligned}
\end{align}
where $v_i\cdot v_j=\delta_{ij}$ ($i=1,2,3,4,5,6$).
Then any element of $\Lambda_\mathrm{E}(C_{5,2}^3)$ can be written in the form $\lambda = \sum_{i=1}^6 m_i\, v_i$ where $m_i\in\BZ$.
For example, the characteristic vector is $\mathbf{x} =\sqrt{5}\,(1,1,1,1,1,1) = 3v_1 +v_2 + 3v_3+v_4+3v_5+v_6$.

Therefore, we can decompose the Construction A lattice $\Lambda_\mathrm{E}(C_{5,2}^3)$ into $\Lambda_\mathrm{E}^{(0)}$ and $\Lambda_\mathrm{E}^{(2)}$ where
\begin{align}
    \begin{aligned}
        \Lambda_\mathrm{E}^{(0)} &= \left\{ \sum_{i=1}^6 m_i\, v_i\;\middle|\; \sum_{i=1}^6 m_i  \in 2\BZ\,,\,m_i\in\BZ \right\}\,,\\
        \Lambda_\mathrm{E}^{(2)} &= \left\{ \sum_{i=1}^6 m_i\, v_i\;\middle|\; \sum_{i=1}^6 m_i  \in 2\BZ+1\,,\,m_i\in\BZ \right\}\,.
    \end{aligned}
\end{align}

On the other hand, the shadow $S\left(\Lambda_\mathrm{E}(C_{5,2}^3)\right)$, which is the translation of $\Lambda_\mathrm{E}(C_{5,2}^3)$ by $\frac{\mathbf{x}}{2}$, is given by
\begin{align}
    S\left(\Lambda_\mathrm{E}(C_{5,2}^3)\right) = \left\{\sum_{i=1}^6 m_i'\,v_i\in\BR^6\;\middle|\; m_i'\in \BZ+\frac{1}{2} \right\}\,.
\end{align}
Therefore, the shadow of $\Lambda_\mathrm{E}(C_{5,2}^3)$ is $S(\Lambda_\mathrm{E}(C_{5,2}^3)) = S(\Lambda_\mathrm{E}(C_{5,2}))\times S(\Lambda_\mathrm{E}(C_{5,2}))\times S(\Lambda_\mathrm{E}(C_{5,2}))$.
This can be divided as
\begin{align}
    \begin{aligned}
        \Lambda_\mathrm{E}^{(1)} &= \left\{ \sum_{i=1}^6 m_i'\, v_i\;\middle|\; \sum_{i=1}^6 m_i'  \in 2\BZ+1\,,\,m_i'\in\BZ+1/2 \right\}\,,\\
        \Lambda_\mathrm{E}^{(3)} &= \left\{ \sum_{i=1}^6 m_i'\, v_i\;\middle|\; \sum_{i=1}^6 m_i'  \in 2\BZ\,,\,m_i'\in\BZ+1/2 \right\}\,.
    \end{aligned}
\end{align}

Using the above sets, we obtain the sets of momenta in the Lorentzian signature. An element of each $\Lambda_{i=0,1,2,3}$ takes the form 
\begin{align}
    (\lambda_1,\lambda_2) = \left(\,\sum_i m_i\, v_i\,,\, \sum_i m_i'\, v_i\,\right)\,,
\end{align}
where $m_i$ and $m_i'$ run over the region
\begin{align}
\begin{aligned}
\label{eq:exmp_lattice_2}
    \Lambda_0:&\qquad m_i\,,\,m_i'\in\BZ\, \;\, &\sum_i m_i &= \sum_i m_i' & & \mod 2\,, \\
    \Lambda_1:&\qquad m_i\,,\,m_i'\in\BZ\,\;\, &\sum_i m_i &= \sum_i m_i'+1 & & \mod 2\,, \\
    \Lambda_2:&\qquad m_i\,,\,m_i'\in\BZ+\frac{1}{2}\, \;\;\, \qquad&\sum_i m_i &= \sum_i m_i' & & \mod 2\,, \\
    \Lambda_3:&\qquad m_i\,,\,m_i'\in\BZ+\frac{1}{2}\,\;\;\, &\sum_i m_i &= \sum_i m_i'+1 & & \mod 2\,.
\end{aligned}
\end{align}
The orbifold theory is based on the momentum lattice $\Lambda_\CO = \Lambda_0\cup \Lambda_3$. In the fermionized theory, the NS sector and the R sector are given by $\Lambda_\mathrm{NS} = \Lambda_0 \cup \Lambda_2$ and $\Lambda_\mathrm{R} = \Lambda_1 \cup \Lambda_3$, respectively.
Recently, it has been shown in \cite{Kawabata:2023usr} that the fermionized theory is an $\CN=4$ superconformal field theory called the GTVW model \cite{Gaberdiel:2013psa}.

\section{Chern-Simons theories for Narain code CFTs and their $\mathbb{Z}_2$ gauging}
\label{sec:chern-simons}

In this section, we argue that abelian Chern-Simons (CS) theories naturally capture some properties of a Narain code CFT, its orbifold, and its fermionization.
Our consideration has some similarity to~\cite{Buican:2021uyp} in spirit but is distinct in many crucial ways.
See footnote~\ref{footnote:folding}.
The CS theories can be viewed as the symmetry topological field theories (TFTs)~\cite{Gaiotto:2020iye,Apruzzi:2021nmk,Freed:2022qnc} for the Narain code CFTs.

As a preparation for the subsequent subsections, here we review the basics of general abelian Chern-Simons theories.
An abelian CS theory is defined by the action
\begin{equation}
    S = \frac{\i}{4\pi} \int  K_{IJ} A^I \wedge \d A^J \,,
\end{equation}
where $A^I$ ($I=1,\ldots,2n$) are U(1) gauge fields and $K=(K_{IJ})$ is a non-degenerate symmetric integer matrix.
It is convenient to represent the level matrix $K$ in terms of an integral lattice~$\Gamma$ with an inner product $\oslash$ and a basis $\{e_I\}$: 
\begin{equation}
    K_{IJ} = e_I \oslash e_J \,.
\end{equation}
We denote the CS theory by ${\rm CS}(\Gamma)$.
The CS theory ${\rm CS}(\Gamma)$ is spin (fermionic), i.e., it depends on the spin structure of the spacetime manifold, when at least one of the diagonal components~$K_{II}$ is odd, i.e., when $\Gamma$ is an odd lattice~\cite{Dijkgraaf:1989pz,Belov:2005ze}.
If $K_{II}$ are even, or equivalently if $\Gamma$ is even, the CS theory is non-spin (bosonic).

Suppose that the Chern-Simons theory is non-spin.
The data for the modular tensor category corresponding to the non-spin abelian CS theory can be expressed in terms of $K$, or equivalently by the pair $(\Gamma,\oslash)$.
See~\cite{Lee:2018eqa} for a review.
A Wilson line 
\begin{equation}
W_n:= \exp\left[\i\, n_I \! \oint A^I\right]
\end{equation} 
(the worldline of an anyon) is labeled by a set of integers $n_I$ with identification $n_I\sim n_I + K_{IJ}\,\ell^J$ for $\ell^J\in\mathbb{Z}$, implying that the anyon label $(n_I)$, or equivalently $n:=n_I\, e^I$ with $\{e^I\}$ a dual basis, takes values in the quotient $\Gamma^*/\Gamma$, where $\Gamma^*$ is the dual of $\Gamma$ with respect to $\oslash$.
The spin of the line is $h(n)=\frac{1}{2} K^{IJ} n_I\, n_J = \frac{1}{2} n \oslash n$,  where $(K^{IJ})$ is the inverse of $K$.
The spin $h(n)$ mod $\mathbb{Z}$ is called the topological spin.
The braiding matrix between $W_{m}$ and $W_{n}$ is 
\begin{equation}
B_{m n}= \exp\left(2\pi\i\, m \oslash n\right) \,,
\end{equation}
and the modular $S$- and $T$-matrices are given as
\begin{equation}\label{eq:S-T-CS}
S_{m n}=|\Gamma^*/\Gamma|^{-\frac{1}{2}}\,B_{m n} \,,
\qquad
T_{m n}= e^{-\frac{\pi\i}{6}}\, e^{2\pi\i\, h(m)}\, \delta_{m n} \,.
\end{equation}

\subsection{Chern-Simons theories for Narain code CFTs}

Recall that the torus partition function of a Narain code CFT can be expressed as a finite sum (specified by the complete weight enumerator) involving simple building blocks ($\psi^+_{ab}\ ,\psi^-_{ab}\ ,\ldots$), as shown for $p=2$ in~\cite{Dymarsky:2020bps,Dymarsky:2020qom},  for $p\neq 2$ in \cite{Yahagi:2022idq}, and for the orbifolded and fermionized versions in section~\ref{sec:gauging_code_CFT}.
See~\eqref{eq:partition_theta_enumerator}.
This structure is similar to that of a rational CFT, whose partition function on a torus is a finite sum involving the products of the characters (more generally conformal blocks for the partition function on a general Riemann surface) of a chiral algebra and their complex conjugate.
In the following, we will interpret the building blocks as the basis states of a canonically quantized abelian Chern-Simons theory on the torus, in the same way as the characters and conformal blocks of a WZW model are interpreted as the basis states of the associated non-abelian Chern-Simons theory~\cite{Witten:1988hf}.
The crucial difference between the two structures 
is that the building blocks for a Narain code CFT depend on both $\tau$ and $\bar\tau$, where $\tau$ is the modulus of the torus.

Technically, this structure of the partition function arises from the Construction A lattice~$\Lambda(\mathcal{C})$ in~\eqref{eq:constructon-A-lattice}, which we write as
\begin{equation}\label{eq:Lambda-C-rewriting}
\Lambda(\mathcal{C}) = \bigcup_{c\,\in\,\mathcal{C}} \left(\Gamma_0 +\frac{c}{\sqrt p}\right) \subset \Gamma^*_0 \,.
\end{equation}
Here
\begin{equation}
\Gamma_0 =  \sqrt{p}\,\mathbb{Z}^{2n}    
\end{equation}
is an even lattice with respect to the inner product $\odot$ determined by the Lorentzian metric~$\eta$ defined in~\eqref{eq:eta-def}.
The partition function is given by a double summation.
The first summation is finite and is over the classical code $\mathcal{C}$.
The second summation is infinite and is over $\gamma\in \Gamma_0$.
The infinite sum over $\gamma$ for fixed $c$ is an analog of a character for the algebra generated by (appropriately normal-ordered versions of)
\begin{equation}\label{eq:non-chiral-generators}
\exp \left[  \frac{\i}{\sqrt{2}}\,\big(
(\gamma_1+\gamma_2)\cdot X(z) + (\gamma_1-\gamma_2) \cdot \bar X(\bar z)
\big)\right]
\end{equation}
for $\gamma = (\gamma_1,\gamma_2) \in \Gamma_0$.
The algebra extends in a non-chiral way the product of chiral and anti-chiral algebras corresponding to $[\U(1)_{2p}]^{n}$ and  $[\U(1)_{-2p}]^{n}$, respectively.
The characters of the latter are holomorphic or anti-holomorphic and are given by products of the theta functions in~\eqref{eq:Theta-m-k-def} (divided by eta functions) and their complex conjugate, which appear in non-holomorphic functions $\psi_{ab}$, etc.%
\footnote{
See, for example, \cite{Dymarsky:2021xfc,Angelinos:2022umf} for the construction of code-based Narain CFTs that are not rational.
}

The structure of the momentum lattice shown in~\eqref{eq:Lambda-C-rewriting} is captured via the bulk-boundary correspondence by a non-spin $\U(1)^{2n}$ CS theory~${\rm CS}(\Gamma_0)$ determined by the even lattice $\Gamma_0$, which is a special case of the Chern-Simons theory discussed at the beginning of this section.\footnote{%
Indeed, the wave function obtained by explicit quantization has a non-holomorphic dependence on $\tau$ when the signature of $K$ is indefinite~\cite{Gukov:2004id,Belov:2005ze}.
}
The codewords $c\in\mathcal{C} \simeq \Lambda(\mathcal{C})/\Gamma_0
\subset 
\Gamma_0 ^*/\Gamma_0  $ correspond to a subset of Wilson lines~$W_c$ (selected by a boundary condition as we will see below).
More general lines $W_\gamma$ are given by arbitrary $\gamma\in\Gamma_0$ with spin $\gamma\odot\gamma/2$ mod $\mathbb{Z}$.
The level matrix is given by
 $   K^{(0)}_{IJ} := b_I \odot b_J = p \,\eta_{IJ}  $, i.e., 
\begin{align}
K^{(0)}= \left[
    \begin{array}{c|c}
0_{n\times n} &\; p 1_{n\times n} \\
p 1_{n\times n}     &0_{n\times n}
    \end{array}\right],
\end{align}
where $b_I = \sqrt{p}\, (\delta_{iI})_{i=1}^{2n}$ form a basis of $\Gamma_0 $.
The theory consists of $n$ copies of the same BF theory, which is equivalent to the topological $\mathbb{Z}_p$ gauge theory~\cite{Maldacena:2001ss,Banks:2010zn} and the low-energy limit of the $\mathbb{Z}_p$ toric code model~\cite{Kitaev:1997wr}.
The diagonal components of $K^{(0)} $ are even (zero) as they should be for a non-spin CS theory.
The modular $S$- and $T$-transformations of $\psi_{ab}(\tau,\bar\tau)$ were computed in~\cite{Kawabata:2022jxt}.
It is consistent with the modular $S$- and $T$-matrices of the CS theory~(\ref{eq:S-T-CS}) determined by $K^{(0)}$.

We propose that the Narain code CFT on a Riemann surface $\Sigma$ is realized by the Chern-Simons theory~${\rm CS}(\Gamma_0)$ on the spacetime $\Sigma\times$(interval), which has two boundaries as illustrated in figure~\ref{fig:3d}.
The bulk theory supports 2d degrees of freedom (DOF) on one boundary (edge modes)~\cite{1995AdPhy..44..405W}, and is subject to a topological boundary condition on the other boundary.
An equivalence class of vertex operators, where two vertex operators are equivalent if they are related by the multiplication of an operator of the form~(\ref{eq:non-chiral-generators}), corresponds to a Wilson line that stretches between the two boundaries.
It is well-known~\cite{Kapustin:2010hk} that a topological boundary condition is in one-to-one correspondence with a Lagrangian subgroup of $\Gamma_0^*/\Gamma_0$ with respect to the bilinear form given by the level matrix.
In the present set-up, the ``Lagrangian subgroup'' is nothing but the classical code~$\mathcal{C}\simeq\Lambda(\mathcal{C})/\Gamma_0$ that is self-dual with respect to the metric $\eta$.
The Lagrangian subgroup specifies the Wilson lines that can end on the boundary.
The evenness of $\Lambda(\mathcal{C})$ implies that these Wilson lines have integer spins, i.e., the boundary condition is bosonic.

\begin{figure}
\centering
\begin{tikzpicture}[scale=1.2]

    \begin{scope}
            \coordinate (O) at (0,0,0);
            \coordinate (A) at (0,2,0);
            \coordinate (B) at (0,2,2);
            \coordinate (C) at (0,0,2);
            \coordinate (D) at (2,0,0);
            \coordinate (E) at (2,2,0);
            \coordinate (F) at (2,2,2);
            \coordinate (G) at (2,0,2);
            
            \draw (O) -- (C) -- (G) -- (D) -- cycle;
            \draw (O) -- (A) -- (E) -- (D) -- cycle;
            \draw[fill=red!50, opacity=0.5] (O) -- (A) -- (B) -- (C) -- cycle;
            \draw[fill=green!20,opacity=0.5] (D) -- (E) -- (F) -- (G) -- cycle;
            \draw (C) -- (B) -- (F) -- (G) -- cycle;
            \draw (A) -- (B) -- (F) -- (E) -- cycle;

            \node at (1.63, 0.5) {\parbox{1cm}{\centering 2d\\ DOF}};

            \node at (-0.36, 0.5) {\parbox{1cm}{\centering top.\\ b.c.}};

            \node at (0.63, 0.5) {CS};
    \end{scope}

    \begin{scope}[xshift=3.25cm, yshift=0.5cm]
        \node at (0,0) {\Large $=$};
    \end{scope}

    \begin{scope}[xshift=5.25cm]
            \coordinate (O) at (0,0,0);
            \coordinate (A) at (0,2,0);
            \coordinate (B) at (0,2,2);
            \coordinate (C) at (0,0,2);

            \draw[thick, fill=cyan!50, opacity=0.5] (O) -- (A) -- (B) -- (C) -- cycle;

            \node at (-0.37, 0.5) {\parbox{1cm}{\centering 2d\\ CFT}};
    \end{scope}

\end{tikzpicture}

\caption{The bulk CS theory supports 2d degrees of freedom (DOF) on one boundary and is subject to a topological boundary condition on the other boundary specified by the classical code~$\mathcal{C}$.
At low energies, the Narain code CFT describes the combined 3d$+$2d system.
}
\label{fig:3d}
\end{figure}
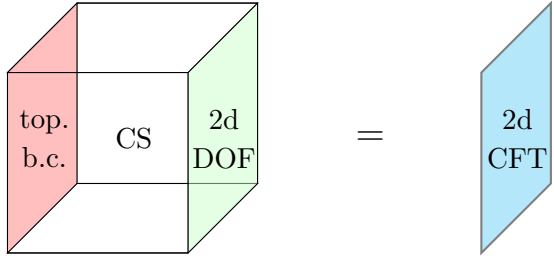

The CS theory $\mathrm{CS}(\Gamma_0)$ on the space $\Sigma$ has a Hilbert space spanned by the building blocks ($\prod_{i=1}^n \psi_{\alpha_i\beta_i}(\tau,\bar\tau)$ if $\Sigma$ is a torus) of the partition function.
The topological boundary condition specified by $\mathcal{C}$ defines a boundary state~$\langle\mathcal{C}|$.
Suppose that the classical code $\mathcal{C}$ corresponds to a quantum stabilizer code, i.e.,  $\mathcal{C}$ satisfies the condition~\eqref{stabilizer_condition}.
Then, a Wilson line~$W$ wrapping a cycle of $\Sigma$ satisfies $\langle\mathcal{C}|W=\langle\mathcal{C}|$ if and only if $W$ corresponds to an element of~$\mathcal{C}$.
This means that we can interpret $\langle\mathcal{C}|$ itself as the unique state of a quantum stabilizer code (with zero logical qubits).
In this interpretation, the general Wilson lines of $\mathrm{CS}(\Gamma_0)$ represent Pauli operators, and those labeled by elements of~$\mathcal{C}$ represent stabilizers.\footnote{\label{footnote:folding}
Our picture is similar to~\cite{Kapustin:2010if,Buican:2021uyp}, but our Chern-Simons theory differs from theirs.
The bulk Chern-Simons theories in these works are products of two decoupled theories corresponding to chiral and anti-chiral sectors of the 2d CFTs.
Our Chern-Simons theories are generally not such products and correspond to non-chiral CFTs with partition functions built from non-holomorphic functions $\psi_{ab}(\tau,\bar\tau)$.}

\subsection{Chern-Simons theories for orbifolded and fermionized CFTs}

Next, we study the bulk-boundary correspondence for the orbifolded and fermionized Narain code CFTs.
We will often make use of an even lattice
\begin{equation}\label{eq:Gamma-e-def}
\Gamma_e := \{ \gamma \in\Gamma_0 \,|\, \chi\odot \gamma \in 2\mathbb{Z}\}  \,.
\end{equation}
We assume that the value of $n$ is such that the $\mathbb{Z}_2$ symmetry used for orbifolding and fermionization is non-anomalous as described in section~\ref{sec:gauging_code_CFT}. 

\subsubsection{$p\neq 2$}

In the case of $p$ odd prime, we have $\chi = \sqrt p \, \mathbf{1}_{2n}$.
The absence of anomaly for the zero-form symmetry $(\mathbb{Z}_2^{[0]})_\chi$ generated by $\chi$ as in~(\ref{eq:Z2-sym-def}) requires that $n\in2\mathbb{Z}$.
We note that $\chi\in \Gamma_e$.
The lattice~$\Lambda_\mathcal{O}$ for the orbifold can be written as
\begin{equation}
\Lambda_\mathcal{O}=
\bigcup_{c\,\in\,\mathcal{C}} \left(\Gamma_\mathcal{O} + \nu(c) \right) 
\,,
\qquad
\Gamma_\mathcal{O} = \Gamma_e \cup \left(\Gamma_e+\zeta\right) \,,\end{equation}
where $\nu(c)\in\Lambda_0$ is a certain element determined by $c$,%
\footnote{%
We omit the explicit expressions for $\nu(c)$, which are easy to calculate but are rather complicated.
The elements $\nu(c)$ specify the representations of the non-chirally extended algebra generated by the operators~(\ref{eq:non-chiral-generators}) with $\gamma\in\Gamma_\mathcal{O}$.
}
and
\begin{equation}
\zeta =
\left\{
\begin{array}{cl}
  \delta + b_{2n}   &  \text{ if $n\in 4\mathbb{Z}+2$}\,, \\
 \delta    & \text{ if  $n\in 4\mathbb{Z}$}\,. 
\end{array}
\right.
\end{equation}
Note that $2\zeta \in \Gamma_e$.
One can check that $\Gamma_\mathcal{O}$ is an even  lattice, with $\zeta\odot\zeta$ even.
Thus the CS theory~$\mathrm{CS}(\Gamma_\mathcal{O})$ is non-spin.
The vertex operators specified by $\nu(c)$ for $c\in\mathcal{C}$ correspond to a subset of its Wilson lines that can end on the left boundary in figure~\ref{fig:3d}.
This means that the quotient group $\Lambda_\mathcal{O}/\Gamma_e$ specifies the topological boundary condition.
The self-duality of the lattice~$\Lambda_\mathcal{O}$ established in section~\ref{sec:gauging_code_CFT} guarantees that $\Lambda_\mathcal{O}/\Gamma_e$ is a Lagrangian subgroup of $\Gamma_e^*/\Gamma_e$ as it should be~\cite{Kapustin:2010hk}.

The non-spin theory~${\rm CS}(\Gamma_\mathcal{O})$ arises by 
gauging a bosonic one-form symmetry of a non-spin CS theory ${\rm CS}(\Gamma_e)$~\cite{Kapustin:2014gua,Gaiotto:2014kfa}.
(See also Appendix~C of~\cite{Seiberg:2016rsg}.)
To see this, we note that the even lattice $\Gamma_e$ 
can be written as
\begin{equation}
\Gamma_e = \{\gamma\in\Gamma_0 \, | \, 
\zeta\odot \gamma \in \mathbb{Z} \} \subset \Gamma_0 \,.
\end{equation}
It follows that $\zeta\in\Gamma_e^*$ corresponds to a Wilson line $W_\zeta$ of $\mathrm{CS}(\Gamma_e)$.
The line $W_\zeta$ has an integer spin, i.e., it is bosonic.
Let us consider the CS theory obtained by gauging the one-form symmetry~$(\BZ_2^{[1]})_\zeta$ generated by $W_\zeta$.
(In the condensed matter language, $W_\zeta$ represents the bosonic anyon that condenses.)%
\footnote{%
We denote by $(G^{[1]})_\lambda$ the one-form $G$ symmetry generated by the line $W_\lambda$.}
This eliminates the Wilson lines that have non-trivial braiding with~$W_\zeta$, i.e., the new labels of lines take values in $\{\gamma\in\Gamma_e^*\,|\, \zeta\odot\gamma\in\mathbb{Z} \} = \Gamma_\mathcal{O}^*$.
Moreover, two lines that differ by the fusion with $W_\zeta$ are identified, i.e., the new labels are identified if their difference is in $\Gamma_e\cup (\Gamma_e+\zeta) = \Gamma_\mathcal{O}$.
We conclude that $\mathrm{CS}(\Gamma_\mathcal{O})$ arises by gauging $(\BZ_2^{[1]})_\zeta$.

For the fermionized theory, the relevant lattice is $\Lambda_\mathrm{NS}=\Lambda_0\cup\Lambda_2$, which we write as
\begin{equation}
    \Lambda_\mathrm{NS}=
    \bigcup_{c\,\in\,\mathcal{C}} \left(\Gamma_\mathrm{NS}
    +\mu(c)
    \right) 
    \,,
    \qquad
    \Gamma_\mathrm{NS} = \Gamma_e \cup \left(\Gamma_e+\xi\right) \,. 
\end{equation}
Here $\mu(c)$ is an element of $\Lambda_0$ uniquely determined by $c$,%
\footnote{
For simplicity, we omit the explicit expressions for $\mu(c)$.
The elements $\mu(c)$ specify the representations of the non-chirally extended algebra generated by the operators~(\ref{eq:non-chiral-generators}) with $\gamma\in\Gamma_\mathrm{NS}$.
}
and
\begin{equation}
\xi =
\left\{
\begin{array}{cl}
   \delta  &  \text{ if  $n\in 4\mathbb{Z}+2$,} \\
\delta + b_{2n}     & \text{ if  $n\in 4\mathbb{Z}$.}
\end{array}
\right.
\end{equation}
One can check that $\Gamma_\mathrm{NS}$ is an odd lattice with $\xi\odot\xi$ being an odd integer, showing that the CS theory~${\rm CS}(\Gamma_\mathrm{NS})$ is spin.
The vertex operators specified by $\mu(c)$ correspond to a subset of Wilson lines in ${\rm CS}(\Gamma_\mathrm{NS})$.

\begin{figure}
\centering

\begin{tikzpicture}[thick, >=stealth, box_orange/.style={rectangle, draw, fill=orange!30, text width=2.2cm, text centered, minimum height=3em}, box_cyan/.style={rectangle, draw, fill=cyan!20, text width=2.2cm, text centered, minimum height=3em}]

    \node[box_orange] (CS0) {CS($\Gamma_0$)};

    \node[box_orange, right of=CS0, node distance=5cm] (CSNS){CS($\Gamma_\text{NS}$)};

    \node[box_orange, left of=CS0, node distance=5cm] (CSOrb) {CS($\Gamma_\CO$)};

    \node[box_orange, above of=CS0, node distance=3cm] (CSe){CS($\Gamma_{e}$)};

    \node[box_cyan, below of=CS0, node distance=3cm] (Narain) {Narain code CFT};

    \node[box_cyan, right of=Narain, node distance=5cm] (fermionized) {Fermionized CFT};

    \node[box_cyan, left of=Narain, node distance=5cm] (orbifolded) {Orbifolded CFT};

    \draw[<-] (0, 0.75) --+ (0, 1.5) node[midway, left=0.1cm] {\parbox{1cm}{\centering gauge\\\vspace*{0.2cm} $(\BZ_2^{[1]})_\iota$}};

     \draw[<-] (-5.2, 0.75) --+ (3.5, 2) node[midway, above left= -0.1cm
     ] {gauge $(\BZ_2^{[1]})_\zeta$};
     
     \draw[<-] (-4.5, -2.25) --+ (3, 4.5) node[midway,  right= 0.1cm
     ] {b.c.};

     \draw[<-] (5.2, 0.75) --+ (-3.5, 2) node[midway, above right= -0.1cm
     ] { 
     gauge $(\BZ_2^{[1]})_\xi$
     };
     
     \draw[<-]
     (4.5, -2.25) --+ (-3, 4.5)
     node[midway, above right = -0.15cm
     ] {b.c.};

     \draw[<-] (1,-2.2) arc (-26:26:5cm) node[midway, below right = 0cm] {b.c.};

     \draw[->]  (0, -0.75) --+ (0, -1.5) node[midway, left=0.1cm
     ] {
     b.c.};

     \draw[->]  (-5, -0.75) --+ (0, -1.5) node[midway, left=0.1cm
     ] {
     b.c.};

     \draw[->]  (5, -0.75) --+ (0, -1.5) node[midway, right=0.1cm] {
     b.c.};

     \draw[<-] (-3.5, -2.9) --+ (2, 0) node[midway, above=0.1cm] {gauge $(\BZ_2^{[0]})_\chi$};
     \draw[->] (-3.5, -3.1) --+ (2, 0) node[midway, below=0.1cm] {gauge $(\BZ_2^{[0]})_\text{qu}$};

     \draw[<-] (3.5, -2.9) --+ (-2, 0) node[midway, above=0.1cm] {fermionize};
     \draw[->] (3.5, -3.1) --+ (-2, 0) node[midway, below=0.1cm] {bosonize};
\end{tikzpicture}

\caption{The interrelations among 3d Chern-Simons theories and 2d CFTs.
The zero-form symmetry $(\BZ_2^{[0]})_\chi$ is the $\mathbb{Z}_2$ symmetry that defines the orbifold.
$(\BZ_2^{[0]})_\text{qu}$ is the corresponding quantum symmetry~\cite{Vafa:1989ih}.
}
\label{fig:CS-relations}
\end{figure}
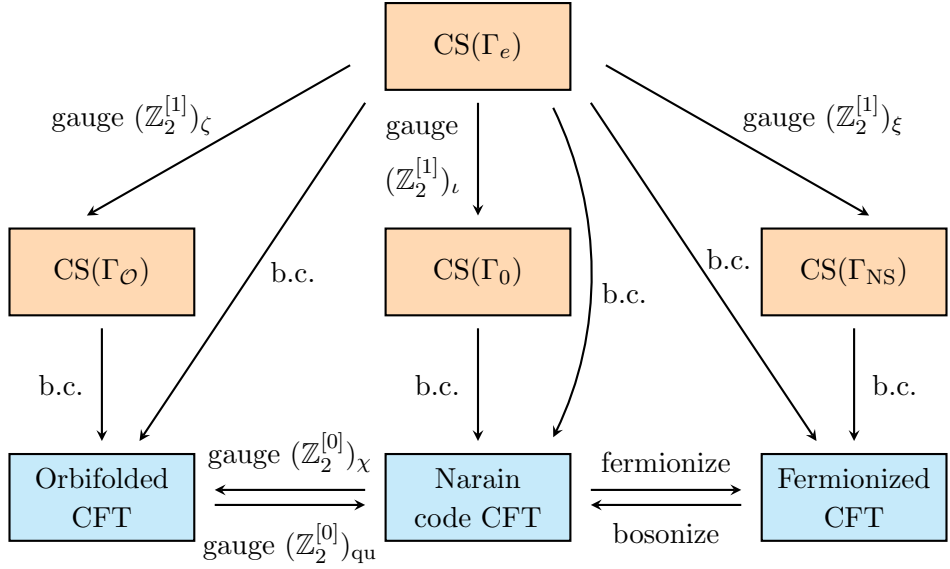
This spin CS theory~${\rm CS}(\Gamma_\mathrm{NS})$ arises by 
gauging a fermionic one-form symmetry of a non-spin CS theory ${\rm CS}(\Gamma_e)$.
The even lattice $\Gamma_e$ admits another rewriting as
\begin{equation}
    \Gamma_e = \{\gamma\in\Gamma_0 \, | \, 
    \xi\odot \gamma \in \mathbb{Z} \} \subset \Gamma_0 \,.
\end{equation}
It follows that $\xi\in\Gamma_e^*$ corresponds to a Wilson line $W_\xi$ of $\mathrm{CS}(\Gamma_e)$.
The line $W_\xi$ has a half-integer spin, i.e., it is fermionic.
Let us consider the CS theory obtained by gauging the one-form symmetry~$(\BZ_2^{[1]})_\xi$ generated by $W_\xi$.
(In the condensed matter language, $W_\xi$ represents the fermionic anyon that condenses.)
This eliminates the Wilson lines that have non-trivial braiding with $W_\xi$, i.e., the new labels of lines take values in $\{\gamma\in\Gamma_e^*\,|\, \xi\odot\gamma\in\mathbb{Z} \} = \Gamma_\mathrm{NS}^*$.
Thus, the spin CS theory~$\mathrm{CS}(\Gamma_\mathrm{NS})$
arises by gauging~$(\BZ_2^{[1]})_\xi$.\footnote{
For $\gamma\in\Gamma_e$, $W_{\gamma}$ and $W_{\gamma+\xi}$ have different spins mod $\mathbb{Z}$ and should not be identified~\cite{Seiberg:2016rsg}.
}

We note that using $\Gamma_e$ in~(\ref{eq:Gamma-e-def}), there is yet another rewriting of $\Gamma_0$, namely  $\Gamma_0=\Gamma_e\cup(\Gamma_e+\iota)$, where $\iota\in\Gamma_0$ is an element such that $\chi\odot\iota\in2\mathbb{Z}+1$.
Because $\Gamma_0$ is even, $\iota\odot\iota\in2\mathbb{Z}$.
Thus, the Wilson line $W_\iota$ of ${\rm CS}(\Gamma_e)$ is bosonic and generates a $\mathbb{Z}_2$ one-form symmetry~$(\BZ_2^{[1]})_\iota$.
Gauging this one-form symmetry produces ${\rm CS}(\Gamma_0)$.
The topological boundary condition is specified by the Lagrangian subgroup $\Lambda(\mathcal{C})/\Gamma_e$ of $\Gamma_e^*/\Gamma_e$.

Thus, the three Chern-Simons theories ${\rm CS}(\Gamma_0)$, ${\rm CS}(\Gamma_\mathcal{O})$, and ${\rm CS}(\Gamma_\mathrm{NS})$ all arise by gauging one-form~$\mathbb{Z}_2$ symmetries of  ${\rm CS}(\Gamma_e)$.
We can realize the Narain code CFT, its orbifold, and its fermionization by imposing the  topological boundary conditions on the fields of the non-spin theory ${\rm CS}(\Gamma_e)$ such that the Wilson lines labeled by the Lagrangian subgroups $\Lambda(\mathcal{C})/\Gamma_e$, $\Lambda_\mathcal{O}/\Gamma_e$, and $\Lambda_\mathrm{NS}/\Gamma_e$ of $\Gamma_e^*/\Gamma_e$ can end on the left boundary of figure~\ref{fig:3d}, respectively.
The first two boundary conditions are bosonic while the last is fermionic.
In this picture, the bulk topological field theory~${\rm CS}(\Gamma_e)$ is invariant under the topological manipulations (orbifolding, fermionization, and bosonization) on the 2d CFTs~\cite{Gaiotto:2020iye,Kaidi:2022cpf}.
The relations we have found are summarized in figure~\ref{fig:CS-relations}.

\subsubsection{$p=2$ and $\mathbb{F}_4$-even $\mathcal{C}$}

\begin{figure}
\centering

\begin{tikzpicture}[thick, >=stealth, box_orange/.style={rectangle, draw, fill=orange!30, text width=2.2cm, text centered, minimum height=2em}, box_cyan/.style={rectangle, draw, fill=cyan!20, text width=2.2cm, text centered, minimum height=2em}]

    \node[box_orange] (CS0) {CS($\Gamma_0$)};

    \node[box_orange, right of=CS0, node distance=5cm] (CSNS){CS($\Gamma_\text{NS}$)};

    \node[box_orange, left of=CS0, node distance=5cm] (CSOrb) {CS($\Gamma_\CO$)};

    \node[box_orange, above of=CS0, node distance=2cm] (CSe){CS($\Gamma_{e}$)};

    \node[box_orange, above of=CSe, node distance=2cm] (CS1){CS($\Gamma_{1}$)};

    \node[box_cyan, below of=CS0, node distance=2cm] (Narain) {Narain code CFT};

    \node[box_cyan, right of=Narain, node distance=5cm] (fermionized) {Fermionized CFT};

    \node[box_cyan, left of=Narain, node distance=5cm] (orbifolded) {Orbifolded CFT};

    \draw[<-, >=stealth] (0,3.5) to (0, 2.5);
    \node at (0.07,3) {\centering gauge\, $(\BZ_2^{[1]})_\chi$};
    
    \draw[->, >=stealth] (0,1.5) to (0, 0.5);
    \node at (0.07,1) {\centering gauge\,    $(\BZ_2^{[1]})_\iota$};

     \draw[<-] (-3.6,0.5) --+ (2.2,1) node[midway, above=0.4cm
     ] {gauge $(\BZ_{4}^{[1]})_\zeta$};
     
     \draw[->] (-1.3,1.5) to (-4,-1.3);
     \node at (-2.2,0) {b.c.};

     \draw[<-] (3.6,0.5) --+ (-2.2,1) node[midway, above=0.4cm
     ] {
     gauge $(\BZ_4^{[1]})_\xi$
     };
     
     \draw[->] (1.3,1.5) to (4,-1.3);
     \node at (3.1,0.2) {b.c.};

     \draw[<-] (1.35,-1.2) arc (-14:14:5cm) node[midway,right = 0cm] {b.c.};

     \draw[->] (0,-0.5) to (0,-1.4);
     \node at (0.4,-1) {b.c.};

     \begin{scope}[xshift=-5cm]
         \draw[->] (0,-0.5) to (0,-1.4);
     \node at (0.4,-1) {b.c.};
     \end{scope}

     \begin{scope}[xshift=5cm]
         \draw[->] (0,-0.5) to (0,-1.4);
     \node at (0.4,-1) {b.c.};
     \end{scope}

    \begin{scope}[yshift=-0.5cm]
        \draw[<-] (-3.5, -1.4) --+ (2, 0) node[midway, above=0.0cm] {gauge $(\BZ_2^{[0]})_\chi$};
     \draw[->] (-3.5, -1.6) --+ (2, 0) node[midway, below=0.0cm] {gauge $(\BZ_2^{[0]})_\text{qu}$};

     \draw[<-] (3.5, -1.4) --+ (-2, 0) node[midway, above=0.0cm] {fermionize};
     \draw[->] (3.5, -1.6) --+ (-2, 0) node[midway, below=0.0cm] {bosonize};
    \end{scope}
     
    \draw[->] (CS1) to[out=-160,in=80] (CSOrb);
    \node at (-4.8,1) {\centering \;\,gauge\, $(\BZ_2^{[1]})_\zeta$};

    \draw[->] (CS1) to[out=-20,in=100] (CSNS);
    \node at (4.8,1) {\centering \;gauge \,$(\BZ_2^{[1]})_\xi$\;};

    \draw[->] (CS1) to[out=180,in=145] (orbifolded);
    \node at (-5,2.8) {b.c.};

    \draw[->] (CS1) to[out=0,in=35] (fermionized);
    \node at (5,2.8) {b.c.};
    
\end{tikzpicture}

\caption{The relations among 3d Chern-Simons theories and 2d CFTs for $p=2$ and an $\mathbb{F}_4$-even classical code~$\mathcal{C}$.
}
\label{fig:CS-relations_p=2_F4-even}
\end{figure}
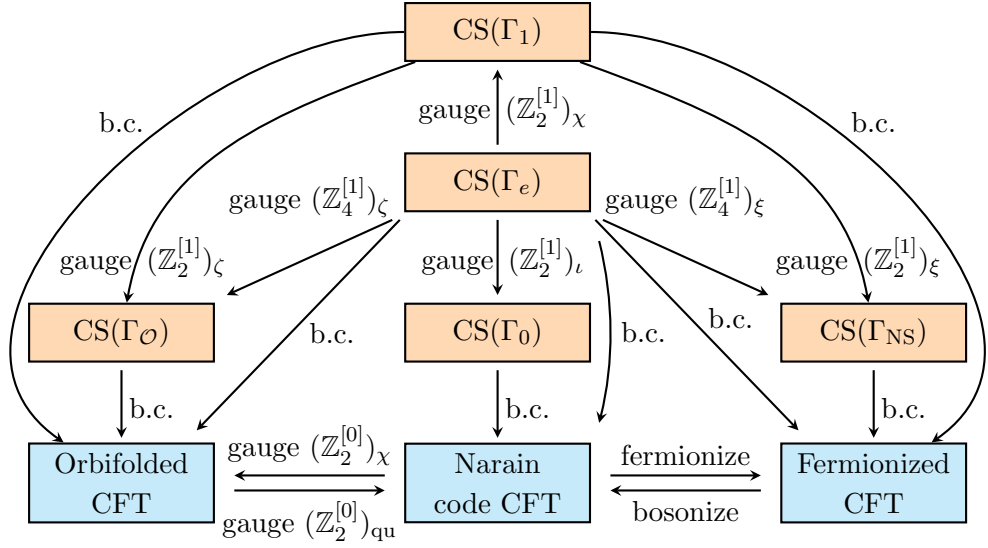

Suppose next that $p=2$ and that the classical code $\mathcal{C}$ is $\mathbb{F}_4$-even.
The absence of anomaly for $(\mathbb{Z}_2^{[0]})_\chi$ generated by $\chi$ requires that $n\in 4\mathbb{Z}$.
(See section~\ref{ss:gauging_code_p=2} for the definition.)
In this case, we have $\mathbf{1}_{2n} \in \mathcal{C}$, $\chi=(1/\sqrt{2}) \mathbf{1}_{2n}$ and $\Gamma_e\neq\Gamma_0$.
Let us define
\begin{equation}
\Gamma_1:= \Gamma_e \cup (\Gamma_e+\chi) \subset \Lambda_0 \,.
\end{equation}
Since $2\chi \in \Gamma_e$, $\Gamma_1$ is closed under addition and is an even lattice.
We note that $\Gamma_1$ is contained in all of $\Lambda(\mathcal{C})$, $\Lambda_\mathcal{O}$, and $\Lambda_\mathrm{NS}$.
Both of the two CS theories $\mathrm{CS}(\Gamma_e)$ and $\mathrm{CS}(\Gamma_1)$ reduce to the Narain code CFT, its orbifold, and it fermionization, with appropriate topological boundary conditions specified by the Lagrangian subgroups $\Lambda(\mathcal{C})/\Gamma_e, \Lambda_\mathcal{O}/\Gamma_e, \Lambda_\mathrm{NS}/\Gamma_e\subset\Gamma_e^*/\Gamma_e$ and $\Lambda(\mathcal{C})/\Gamma_1, \Lambda_\mathcal{O}/\Gamma_1, \Lambda_\mathrm{NS}/\Gamma_1 \subset\Gamma_1^*/\Gamma_1$ imposed, respectively.
The theory $\mathrm{CS}(\Gamma_1)$ arises by gauging the one-form symmetry $(\mathbb{Z}_2^{[1]})_\chi$ generated by $W_\chi$ of $\mathrm{CS}(\Gamma_e)$.

Moreover, we again have $\Gamma_0 = \Gamma_e\cup(\Gamma_e+\iota)$.
The Wilson line $W_\iota$ of ${\rm CS}(\Gamma_e)$ is bosonic and gauging~$(\BZ_2^{[1]})_\iota$ produces ${\rm CS}(\Gamma_0)$.
The topological boundary condition for the Narain code CFT is specified by the Lagrangian subgroup $\Lambda(\mathcal{C})/\Gamma_e$ of $\Gamma_e^*/\Gamma_e$.

We can also introduce
\begin{align}
\Gamma_\mathcal{O} = \bigcup_{j=0}^3 (\Gamma_e + j \cdot \zeta)
\,, \qquad
\Gamma_\mathrm{NS} = \bigcup_{j=0}^3 (\Gamma_e + j \cdot \xi)
\end{align}
with
\begin{align}
\zeta &= \left\{
\begin{array}{cl}
\delta     & \text{ if } n\in 8\mathbb{Z}\,,  \\
\delta+b_{2n}     & \text{ if } n\in 8\mathbb{Z}+4  
\end{array}
\right.
\end{align}
and
\begin{align}
\xi &= \left\{
\begin{array}{cl}
\delta + b_{2n}     & \text{ if } n\in 8\mathbb{Z}\,,  \\
\delta     & \text{ if } n\in 8\mathbb{Z}+4   \,,
\end{array}
\right.
\end{align}
and consider the corresponding CS theories $\mathrm{CS}(\Gamma_\mathcal{O})$ and $\mathrm{CS}(\Gamma_\mathrm{NS})$.

The relations among the CS theories and the CFTs are summarized in figure~\ref{fig:CS-relations_p=2_F4-even}.
We note that $W_\zeta$ of $\mathrm{CS}(\Gamma_\mathcal{O})$ and $W_\xi$ of $\mathrm{CS}(\Gamma_\mathrm{NS})$ generate one-form $\mathbb{Z}_4$ symmetries.

\subsubsection{$p=2$ and non-$\mathbb{F}_4$-even $\mathcal{C}$}

\begin{figure}
\centering

\begin{tikzpicture}[thick, >=stealth, box_orange/.style={rectangle, draw, fill=orange!30, text width=2.2cm, text centered, minimum height=3em}, box_cyan/.style={rectangle, draw, fill=cyan!20, text width=2.2cm, text centered, minimum height=3em}]

    \node[box_orange] (CSe) {CS($\Gamma_e$)};

    \node[box_orange, left of=CSe, node distance=5cm] (CSOrb) {CS($\Gamma_\CO$)};

    \node[box_cyan, below of=CS0, node distance=3cm] (Narain) {Narain code CFT};

    \node[box_cyan, right of=Narain, node distance=5cm] (fermionized) {Fermionized CFT};

    \node[box_cyan, left of=Narain, node distance=5cm] (orbifolded) {Orbifolded CFT};

    \draw[<-]  (-3.5, 0) --+ (2, 0) node[midway, above=0.1cm
     ] {gauge $(\mathbb{Z}_2^{[1]})_\zeta$};

    \draw[->]  (0, -0.75) --+ (0, -1.5) node[midway, left=0.1cm
     ] {
     b.c.};
    \draw[->]  (0.3, -0.75) --+ (4, -1.5) node[midway, above right =0cm
     ] {
     b.c.};
    \draw[->]  (-0.3, -0.75) --+ (-4, -1.5) node[midway, above left =0cm
     ] {
     b.c.};
     
    \draw[->]  (-5, -0.75) --+ (0, -1.5) node[midway, left=0.1cm
     ] {
     b.c.};    

     \draw[<-] (-3.5, -2.9) --+ (2, 0) node[midway, above=0.1cm] {gauge $(\BZ_2^{[0]})_\chi$};
     \draw[->] (-3.5, -3.1) --+ (2, 0) node[midway, below=0.1cm] {gauge $(\BZ_2^{[0]})_\text{qu}$};

     \draw[<-] (3.5, -2.9) --+ (-2, 0) node[midway, above=0.1cm] {fermionize};
     \draw[->] (3.5, -3.1) --+ (-2, 0) node[midway, below=0.1cm] {bosonize};
\end{tikzpicture}
\caption{The interrelations among 3d Chern-Simons theories and 2d CFTs for $p=2$, $n\in2\mathbb{Z}$, and a non-$\mathbb{F}_4$-even classical code~$\mathcal{C}$.
}
\label{fig:CS-relations_p=2_n-even_non-F4-even}
\end{figure}

In this case,
we have $\mathbf{1}_{2n} \notin \mathcal{C}$ and we take $\chi=\sqrt{2}\, \mathbf{1}_{2n}$, so $\Gamma_e = \Gamma_0$.

For $n \in 2\mathbb{Z}$, we get
\begin{equation}
\Lambda_{\mathcal{O}} = \mathop{\bigcup_{c\in\mathcal{C}}}_{\mathbf{1}_{2n}\cdot c \in 2\mathbb{Z}} \left(\Gamma_\mathcal{O}+ \frac{c}{\sqrt 2}\right) 
\,,
\qquad
\Lambda_{\mathrm{NS}} = \bigcup_{c\in\mathcal{C}} \left(\Gamma_0+\mu(c)\right)
\end{equation}
with
\begin{equation}
\Gamma_\mathcal{O}=\Gamma_0 \cup \left(\Gamma_0+ \frac{\mathbf{1}_{2n}}{\sqrt 2}\right)
\end{equation}
and
\begin{equation}
\mu(c)  = \left\{
\begin{array}{ll}
c & \text{ if } \mathbf{1}_{2n}\cdot c = 0 \text{ mod } 2 \,, \\
c + \mathbf{1}_{2n} & \text{ if } \mathbf{1}_{2n}\cdot c = 1 \text{ mod } 2 \,.
\end{array}
\right.
\end{equation}
The Narain code CFT, its orbifold, and its fermionized CFT are realized by the single non-spin CS theory $\mathrm{CS}(\Gamma_0)$ with the topological boundary conditions specified by the Lagrangian subgroups $\Lambda(\mathcal{C})/\Gamma_0, \Lambda_\mathcal{O}/\Gamma_0, \Lambda_\mathrm{NS}/\Gamma_0\subset\Gamma_0^*/\Gamma_0$, respectively.
Moreover, the orbifolded CFT can be obtained by imposing the topological boundary condition specified by $\Lambda_\mathcal{O}/\Gamma_\mathcal{O}$ on the non-spin theory~$\mathrm{CS}(\Gamma_\mathcal{O})$.
We summarize the relations in figure~\ref{fig:CS-relations_p=2_n-even_non-F4-even}.

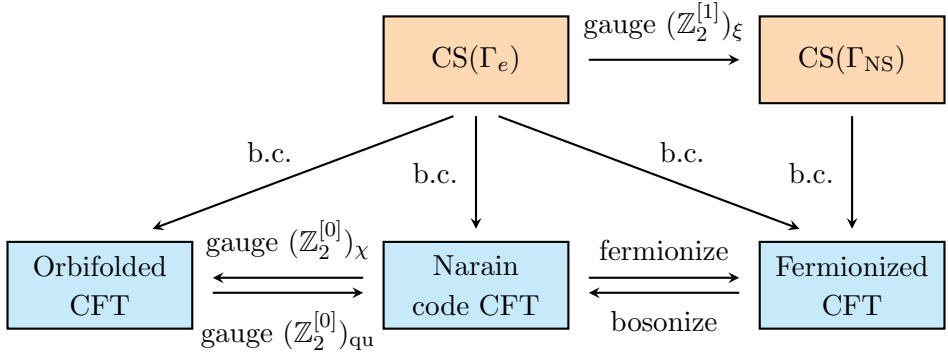
\begin{figure}
\centering

\begin{tikzpicture}[thick, >=stealth, box_orange/.style={rectangle, draw, fill=orange!30, text width=2.2cm, text centered, minimum height=3em}, box_cyan/.style={rectangle, draw, fill=cyan!20, text width=2.2cm, text centered, minimum height=3em}]

    \node[box_orange] (CSe) {CS($\Gamma_e$)};

    \node[box_orange, right of=CSe, node distance=5cm] (CSNS){CS($\Gamma_\text{NS}$)};

    \node[box_cyan, below of=CS0, node distance=3cm] (Narain) {Narain code CFT};

    \node[box_cyan, right of=Narain, node distance=5cm] (fermionized) {Fermionized CFT};

    \node[box_cyan, left of=Narain, node distance=5cm] (orbifolded) {Orbifolded CFT};
     
    \draw[<-]  (3.5, 0) --+ (-2, 0) node[midway, above=0.1cm] {gauge $(\mathbb{Z}_2^{[1]})_\xi$};

    \draw[->]  (0, -0.75) --+ (0, -1.5) node[midway, left=0.1cm
     ] {
     b.c.};
     
    \draw[->]  (0.3, -0.75) --+ (4, -1.5) node[midway, above right =0cm
     ] {
     b.c.};
    \draw[->]  (-0.3, -0.75) --+ (-4, -1.5) node[midway, above left =0cm
     ] {
     b.c.};

    \draw[->]  (5, -0.75) --+ (0, -1.5) node[midway, left=0.1cm      ] {      b.c.};
     ] {
     b.c.};    

     \draw[<-] (-3.5, -2.9) --+ (2, 0) node[midway, above=0.1cm] {gauge $(\BZ_2^{[0]})_\chi$};
     \draw[->] (-3.5, -3.1) --+ (2, 0) node[midway, below=0.1cm] {gauge $(\BZ_2^{[0]})_\text{qu}$};

     \draw[<-] (3.5, -2.9) --+ (-2, 0) node[midway, above=0.1cm] {fermionize};
     \draw[->] (3.5, -3.1) --+ (-2, 0) node[midway, below=0.1cm] {bosonize};
\end{tikzpicture}
\caption{The interrelations among 3d Chern-Simons theories and 2d CFTs for $p=2$, $n\in2\mathbb{Z}+1$, and a non-$\mathbb{F}_4$-even classical code~$\mathcal{C}$.
}
\label{fig:CS-relations_p=2_n-odd_non-F4-even}
\end{figure}

For $n\in2\mathbb{Z}+1$,%
\footnote{%
The analysis here applies to the example in section~\ref{ss:n=1}.
}
we get
\begin{equation}
\Lambda_{\mathcal{O}} = \bigcup_{c\in\mathcal{C}} (\Gamma_0+\nu(c))
\,,
\qquad
\Lambda_\text{NS} = \mathop{\bigcup_{c\in\mathcal{C}}}_{\mathbf{1}_{2n}\cdot c \in 2\mathbb{Z}} \left(\Gamma_\mathrm{NS}+ \frac{c}{\sqrt 2}\right) 
\end{equation}
with
\begin{equation}
\Gamma_\mathrm{NS}=\Gamma_0 \cup \left(\Gamma_0+ \frac{\mathbf{1}_{2n}}{\sqrt 2}\right)
\end{equation}
and
\begin{equation}
\nu(c)  = \left\{
\begin{array}{ll}
c & \text{ if } \mathbf{1}_{2n}\cdot c = 0 \text{ mod } 2 \,, \\
c + \mathbf{1}_{2n} & \text{ if } \mathbf{1}_{2n}\cdot c = 1 \text{ mod } 2 \,.
\end{array}
\right.
\end{equation}
The Narain code CFT, its orbifold, and its fermionized CFT are realized by the single non-spin CS theory $\mathrm{CS}(\Gamma_0)$ with the topological boundary conditions specified by the Lagrangian subgroups $\Lambda(\mathcal{C})/\Gamma_0, \Lambda_\mathcal{O}/\Gamma_0, \Lambda_\mathrm{NS}/\Gamma_0\subset\Gamma_0^*/\Gamma_0$, respectively.
Moreover, the fermionized CFT can be obtained by imposing a topological boundary condition  on the spin theory~$\mathrm{CS}(\Gamma_\mathrm{NS})$.
We summarize the relations in figure~\ref{fig:CS-relations_p=2_n-odd_non-F4-even}.

\section{Ensemble average}
\label{sec:ensemble}

In this section, we will consider ensemble averages of the orbifolded and fermionized Narain code CFTs over CSS codes defined by classical self-dual codes, and calculate their partition functions in the large central charge (large-$n$) limit.

To this end, we first rewrite the building blocks \eqref{eq:rel_four_partition} for the partition functions in terms of the weight enumerator polynomial:
\begin{align}
    S 
        &= 
            \frac{1}{2\,|\eta(\tau)|^{2n}}\left[  W_\CC (\{\psi_{ab}^+\}) + W_\CC (\{\psi_{ab}^-\}) \right], \\
    T 
        &=
            \frac{1}{2\,|\eta(\tau)|^{2n}}\left[  W_\CC (\{\psi_{ab}^+\}) - W_\CC (\{\psi_{ab}^-\}) \right], \\
    U 
        &=
            \frac{1}{2\,|\eta(\tau)|^{2n}}\left[  W_\CC (\{\tilde{\psi}_{ab}^+\}) + W_\CC (\{\tilde{\psi}_{ab}^-\}) \right], \\
    V 
        &=
            \frac{1}{2\,|\eta(\tau)|^{2n}}\left[  W_\CC (\{\tilde{\psi}_{ab}^+\}) - W_\CC (\{\tilde{\psi}_{ab}^-\}) \right], 
\end{align}
where $\psi_{ab}^\pm$, $\tilde\psi_{ab}^\pm$ are given in \eqref{psi_pm}, \eqref{psi_tilde_pm} and can be written as 
\begin{align}
    \psi_{ab}^+
        &= 
            \Theta\genfrac{[}{]}{0.0pt}{}{\Balpha}{\Bzero}
            \left(\Bzero\,\middle|\,\BOmega\right) , \\
    \psi_{ab}^-
        &=  
            \Theta\genfrac{[}{]}{0.0pt}{}{\Balpha}{\Bzero}
            \left(p\,\Bdelta\,\middle|\,\BOmega\right), \\
    \tilde\psi_{ab}^+
        &= 
            \Theta\genfrac{[}{]}{0.0pt}{}{\Balpha + \Bdelta}{\Bzero}
            \left(\Bzero\,\middle|\,\BOmega\right), \\
    \tilde\psi_{ab}^-
        &=  
            \Theta\genfrac{[}{]}{0.0pt}{}{\Balpha + \Bdelta}{\Bzero}
            \left(\Bzero\,\middle|\,\BOmega+\BDelta\right) \ .
\end{align}
$\Theta$ is the Riemann-Siegel theta function of genus-two defined by
\begin{align}
    \Theta\genfrac{[}{]}{0.0pt}{}{\Balpha}{\Bbeta}
\left(\Bz\,\middle|\,\BOmega\right)
        =\!
        \sum_{\Bn\in
        \BZ^{2}}e^{2\pi \i\left[\frac{(\Bn+\Balpha)\,
        \BOmega\,(\Bn+\Balpha)^T}{2}+(\Bn+
        \Balpha)(\Bz+\Bbeta)^T\right]}\ ,
\end{align}
and we define the parameters as
\begin{align}
    \Balpha 
            =
            \left( \frac{a}{p}, \frac{b}{p}\right), \quad
    \Bdelta
            =
            \left(\frac{1}{2}, \frac{1}{2}\right),\quad
    \BOmega
            =
            p\begin{bmatrix}
                    \i\,\tau_2 & \tau_1 \\
                    \tau_1 & \i\, \tau_2
                \end{bmatrix}, \quad
    \BDelta
            =
            \begin{bmatrix}
                    ~0~ & p~ \\
                    ~p~ & 0~
                \end{bmatrix}.
\end{align} 

Let us consider Narain code CFTs based on CSS codes defined by a single self-dual code and average their partition functions over the set of self-dual codes
\begin{align}
\label{eq:set_of_self-dual_codes}
    \CM_{n,p}=\{\text{self-dual codes over $\BF_p$ of length $n$}\}\,.
\end{align}
After averaging over the CSS codes, the complete weight enumerator polynomial becomes \cite{Kawabata:2022jxt}
\begin{align}
\begin{aligned}
\label{eq:averaged_CSS}
    \overline{W}^{\mathrm{(CSS)}}_{n,p}(\{x_{ab}\})
     &= \frac{1}{|\CM_{n,p}|}\,\sum_{C\,\in\,\CM_{n,p}}\,W_{C,C}^{(\mathrm{CSS})}(\{x_{ab}\})\\
    &=
        \begin{dcases}
        \sum_A\, \frac{1}{\left(2^{\frac{n}{2}-1}+1\right)\cdots \left(2^{\frac{n}{2}-\mathrm{dim}_2(A)+1}+1\right)}\,\binom{n}{A}
        \,x^A\, & \text{if\, $p=2$}\,,
        \\
        \sum_A\, \frac{1}{\left(p^{\frac{n}{2}-1}+1\right)\cdots \left(p^{\frac{n}{2}-\mathrm{dim}_p(A)}+1 \right)}\,\binom{n}{A}
        \,x^A\, & \text{if\, $p$ odd prime}\,.
        \end{dcases}
\end{aligned}
\end{align}
In the large-$n$ limit, it reduces to
\begin{equation}\label{eq:W-CSS-saddle-result}
 \overline{W}^{\mathrm{(CSS)}}_{n,p}(\{x_{ab}\})
 = p^{-n}
 \left(\sum_{a,b} x_{ab}\right)^n \left(1+\mathcal{O}(n^{-1})\right) \,.
\end{equation}

To proceed, we will be focused on the averaged partition function by fixing the torus moduli $\tau=\i\, \beta/2\pi$.
Then we have $q=\bar{q}=e^{-\beta}$ and $\psi_{ab}^\pm, \tilde\psi_{ab}^\pm$ simplify to
\begin{align}
    \begin{aligned}
        \psi_{ab}^+ 
            &= 
                \psi_{a}^+ \,\psi_{b}^+ \ , \qquad \psi^+_a \equiv \sum_{k\in \BZ}\,e^{-\frac{p\beta}{2}\left( k+ \frac{a}{p}\right)^2} \ , \\
        \psi_{ab}^-
            &= 
                \psi_{a}^- \,\psi_{b}^- \ , \qquad \psi^-_a \equiv (-1)^a \sum_{k\in \BZ}\,(-1)^{pk}\,e^{-\frac{p\beta}{2}\left( k+ \frac{a}{p}\right)^2} \ , \\   
        \tilde\psi_{ab}^+ 
            &= 
                \tilde\psi_{a}^+ \,\tilde\psi_{b}^+ \ , \qquad \tilde\psi^+_a \equiv \sum_{k\in \BZ}\,e^{-\frac{p\beta}{2}\left( k+ \frac{a}{p} + \frac{1}{2}\right)^2} \ , \\
        \tilde\psi_{ab}^- 
            &=
                e^{2\pi\i\,p \left( \frac{a}{p} + \frac{1}{2}\right)\left( \frac{b}{p} + \frac{1}{2}\right)}\,
                \left( \sum_{k_1\in \BZ} (-1)^{k_1}\,e^{-\frac{p\beta}{2}\left( k_1 + \frac{a}{p} + \frac{1}{2}\right)^2 } \right)
                \left( \sum_{k_2\in \BZ} (-1)^{k_2}\,e^{-\frac{p\beta}{2}\left( k_2 + \frac{b}{p} + \frac{1}{2}\right)^2 } \right)
            \ .
    \end{aligned}
\end{align}
We also find the summations over $a,b \in \BF_p$ become
\begin{align}
    \begin{aligned}
        \sum_{a,b\,\in \,\BF_p} \psi_{ab}^+
            &=
                \vartheta_{3}\left(\frac{\i\,\beta}{2\pi p}\right)^2 
                \ , \qquad
        \sum_{a,b\,\in\, \BF_p} \psi_{ab}^-
            =
                \vartheta_{4}\left(\frac{\i\,\beta}{2\pi p}\right)^2 
                \ , \\
        \sum_{a,b\,\in\, \BF_p} \tilde\psi_{ab}^+
            &=
                \vartheta_{2}\left(\frac{\i\,\beta}{2\pi p}\right)^2 
                \ , \qquad
        \sum_{a,b\,\in\, \BF_p} \tilde\psi_{ab}^-
            =
                \Theta 
                    \begin{bmatrix}
                        \Bdelta \\
                        {\bf 0}
                    \end{bmatrix}
                \left( {\bf 0}\,|\BOmega'\right)\ ,
    \end{aligned}
\end{align}
where 
\begin{align}
    \Bdelta
        =
        \left( \frac{1}{2},\frac{1}{2}\right)\ , \qquad
    \BOmega'
        =
            \begin{bmatrix}
                \frac{\i\,\beta}{2\pi p} & \frac{1}{p} \\
                \frac{1}{p} & \frac{\i\,\beta}{2\pi p}
            \end{bmatrix}\ .
\end{align}

The averaged partition functions of the bosonic, orbifold, and fermionic theories without background $\BZ_2$ gauge fields are
\begin{align}\label{averaged_PF}
    \begin{aligned}
        \bar Z 
            &=
            S + T
            =
            \frac{1}{p^n\,\left|\eta\left(\frac{\i\,\beta}{2\pi}\right)\right|^{2n}}\,
            \vartheta_{3}\left(\frac{\i\,\beta}{2\pi p}\right)^{2n}
            \ , \\
        \bar Z_\CO[00]
            &=
            S + U \\
            &=
            \frac{1}{2\,p^n\,\left|\eta\left(\frac{\i\,\beta}{2\pi}\right)\right|^{2n}}\,\left[ 
            \vartheta_{3}\left(\frac{\i\,\beta}{2\pi p}\right)^{2n} + \vartheta_{4}\left(\frac{\i\,\beta}{2\pi p}\right)^{2n} + \vartheta_{2}\left(\frac{\i\,\beta}{2\pi p}\right)^{2n} 
            + \Theta 
                    \begin{bmatrix}
                        \Bdelta \\
                        {\bf 0}
                    \end{bmatrix}
                \left( {\bf 0}\,|\BOmega'\right)^{n}\right]
            \ , \\
        \bar Z_\CF[00]
            &=
            S + V \\
            &=
            \frac{1}{2\,p^n\,\left|\eta\left(\frac{\i\,\beta}{2\pi}\right)\right|^{2n}}\,\left[
            \vartheta_{3}\left(\frac{\i\,\beta}{2\pi p}\right)^{2n} + \vartheta_{4}\left(\frac{\i\,\beta}{2\pi p}\right)^{2n} + \vartheta_{2}\left(\frac{\i\,\beta}{2\pi p}\right)^{2n} 
            - \Theta 
                    \begin{bmatrix}
                        \Bdelta \\
                        {\bf 0}
                    \end{bmatrix}
                \left( {\bf 0}\,|\BOmega'\right)^{n}\right]
            \ .
    \end{aligned}
\end{align}

To see the difference between the two bosonic partition functions $\bar Z$ and $\bar Z_\CO [00]$, we numerically plot their ratio $\bar Z_\CO [00]/\bar Z$ as a function of $\beta$ 
in figure \ref{fig:averaged_PF_ratio}.
\begin{figure}[t]
    \centering
    \includegraphics[width=9cm]{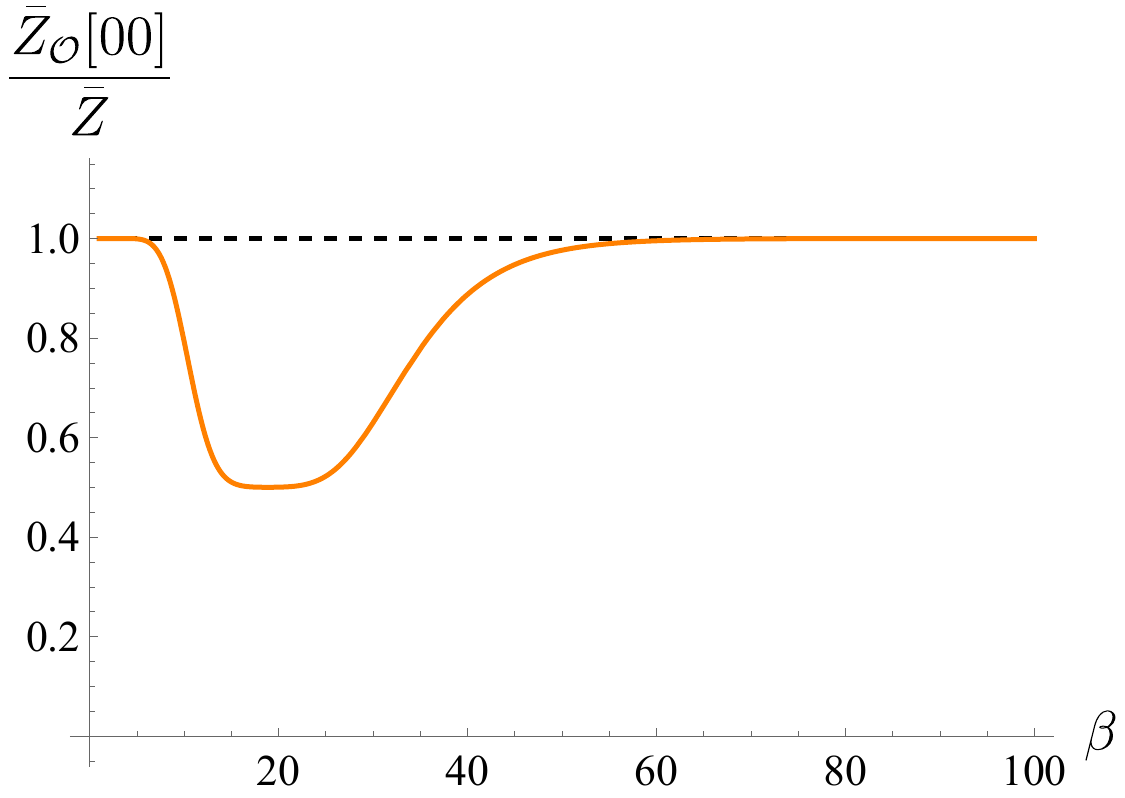}
    \caption{The ratio of the averaged partition functions $\bar Z$ and $\bar Z_\CO [00]$ of the Narain code CFTs constructed from the CSS codes and their orbifolds for $p=3$ and $n=25$.
    The ratio approaches one in $\beta\to 0$ and $\beta\to \infty$ limits, but deviates from one in general, implying that the averages of Narain code CFTs and their orbifolds are different theories.
    The position of the bump shifts to the right as $p$ increases.}
    \label{fig:averaged_PF_ratio}
\end{figure}
We find the ratio approaches one for small and large $\beta$, but deviates from one for the intermediate value.
We can understand the large $\beta$ behavior as follows.
In $\beta \to \infty$ limit, assuming $n\,e^{-\frac{\beta}{2p}}\ll 1$, we find
\begin{align}
    S = O\left( e^{\frac{n\beta}{12}} \right) \ , \qquad T = O\left( n\,e^{\frac{n\beta}{12} - \frac{\beta}{2p}}\right) \ , \qquad 
    U, V = O\left(e^{\frac{n\beta}{12} \left( 1 - \frac{3}{p}\right)}\right)  \ ,
\end{align}
hence, the $S$ term dominates in both $\bar Z$ and $\bar Z_\CO [00]$, and their ratio approaches one.
In $\beta \to 0$ limit, with the help of the modular transformation laws of the elliptic theta, eta, and Riemann-Siegel theta function (see e.g., \cite[\href{http://dlmf.nist.gov/21.5.E9}{(21.5.9)}]{NIST:DLMF}), we find
\begin{align}
    S = O\left( e^{\frac{n\pi^2}{3\beta}}\right) \ , \qquad T = O\left(n\,e^{\frac{n\pi^2}{3\beta} - \frac{2\pi^2 p}{\beta}} \right) \ , \qquad 
    U, V =  O\left(e^{\frac{n\pi^2}{\beta}\left( \frac{1}{3} - \frac{1}{\pi p}\right)}\right)\ ,
\end{align}
where we assume $n\,e^{- \frac{2\pi^2 p}{\beta}} \ll 1$.
Thus, the $S$ term also dominates in both $\bar Z$ and $\bar Z_\CO [00]$, and their ratio approaches one.

The above result shows that the averaged theories of the bosonic code CFTs and their orbifolds have the same spectrum for small and large $\beta$ in the large-$c$ (large-$n$) limit, but they are different ensemble averaging of Narain CFTs.

\section{Discussion}
\label{sec:discussion}

In this paper, we considered the gauging of a $\BZ_2$ symmetry in the bosonic Narain CFTs constructed from qudit stabilizer codes and obtained their orbifolded and fermionized CFTs.
We established the correspondence between the $\BZ_2$ even/odd Hilbert spaces in the untwisted/twisted sectors of the bosonic CFTs with the four sets of the lattice points associated with the codes.
Under this identification, the orbifolding and fermionization of Narain code CFTs are realized as deformations of the lattices by vectors that characterize the $\BZ_2$ symmetry.

In section~\ref{sec:gauging-Z2-bosonic}, we discussed the general lattice modifications to construct new even and odd self-dual lattices from an already-known one.
This technique has been used in Euclidean lattices to realize dense sphere packings such as the Leech lattice~\cite{conway2013sphere}.
Our formulation is valid not only for the Euclidean signature but also for the Lorentzian one, so we expect the modified Lorentzian lattices to have good lattice sphere packings and equivalently large spectral gaps.
It would be interesting to discuss the relationship between the modified Lorentzian lattices and the optimal sphere packings found by the modular bootstrap in~\cite{Hartman:2019pcd,Afkhami-Jeddi:2020hde,Afkhami-Jeddi:2020ezh}.

When a CFT has more than two $\BZ_2$ symmetries, there is a rich structure known as orbifold groupoid relating between the orbifolded and fermionized CFTs by several choices of the $\BZ_2$ subgroup \cite{Dixon:1988qd,Dolan:1994st,Gaiotto:2018ypj,Gaiotto:2020iye}.
In a Narain code CFT, such a groupoid structure may appear if there is more than one vector $\chi$ in the Construction A lattice with which one can deform the original lattice and obtain a family of new lattices.
It remains open from what class of quantum codes one can construct Narain code CFTs with symmetries including multiple $\BZ_2$ subgroups.

We note that the orbifolding and fermionization of Narain code CFTs are not necessarily Narain code CFTs constructed from qudit stabilizer codes in general.
For chiral cases, there are fermionic CFTs that can be built directly from classical codes without performing fermionization \cite{Gaiotto:2018ypj,Kawabata:2023nlt}.
It would be worthwhile to examine if there are similar constructions of non-chiral fermionic CFTs from quantum stabilizer codes without gauging a $\BZ_2$ symmetry.

In section \ref{sec:ensemble}, we considered the ensemble average of Narain code CFTs of CSS type, their orbifolded and fermionized theories, and derived the averaged partition functions in the large central charge (large-$n$) limit.
We pointed out that the averaged partition functions of the bosonic Narain code CFTs and their orbifolds are different.
This observation leads us to the conclusion that the ensemble averages of the Narain code CFTs and their orbifolds describe two different bosonic CFTs.
Since the two theories are discrete subsets of Narain CFTs, it may be reasonable to consider a more general ensemble that includes both Narain code CFTs and their orbifolds and take the weighted average.
It would be interesting to investigate if the resulting CFT can have a holographic description such as $\U(1)$ gravity as in \cite{Maloney:2020nni,Afkhami-Jeddi:2020ezh}.
Also, the average of fermionic CFTs is considered and proposed to be holographically dual to a spin Chern-Simons theory in \cite{Ashwinkumar:2021kav}.
A similar consideration may be applied to the average of fermionized code CFTs, and the partition function $\bar Z_\CF$ obtained in \eqref{averaged_PF} would be useful to identify the holographic description.

Fermionization of Narain code CFTs has been applied to searching for supersymmetric CFTs in the recent paper \cite{Kawabata:2023usr}, where the necessary conditions for CFTs to have supersymmetry \cite{Bae:2020xzl,Bae:2021jkc,Bae:2021lvk} are reformulated in terms of quantum codes.
On the other hand, a new class of Narain code CFTs has been constructed from quantum stabilizer codes over rings and finite fields in \cite{Alam:2023qac}.
Exploring supersymmetric CFTs through the fermionization of these novel code CFTs would be a promising avenue for future research.

\acknowledgments
We are grateful to S.\,M\"oller, Y.\,Moriwaki and H.\,Wada for valuable discussions.
The work of T.\,N. was supported in part by the JSPS Grant-in-Aid for Scientific Research (C) No.\,19K03863, Grant-in-Aid for Scientific Research (A) No.\,21H04469, and Grant-in-Aid for Transformative Research Areas (A) ``Extreme Universe''
No.\,21H05182 and No.\,21H05190.
The work of T.\,O. was supported in part by Grant-in-Aid for Transformative Research Areas (A) ``Extreme Universe'' No.\,21H05190.
The work of K.\,K. was supported by FoPM, WINGS Program, the University of Tokyo and JSPS KAKENHI Grant-in-Aid for JSPS fellows Grant No.\,23KJ0436.

\bibliographystyle{JHEP}
\bibliography{QEC_CFT}

\end{document}